\documentclass[copyright,creativecommons]{eptcs}
\pdfoutput=1

\usepackage{amsmath} % AMS Math Package
\usepackage{amsthm} % Theorem Formatting
\usepackage{amssymb}	% Math symbols such as \mathbb
\usepackage{bold-extra} % required for bold small-caps

\newenvironment{menum}
{\begin{enumerate}
  \setlength{\itemsep}{1pt}
  \setlength{\parskip}{0pt}
  \setlength{\parsep}{0pt}}
{\end{enumerate}}

\newenvironment{mitem}
{\begin{itemize}
  \setlength{\itemsep}{1pt}
  \setlength{\parskip}{0pt}
  \setlength{\parsep}{0pt}}
{\end{itemize}}

\newtheorem{thm}{Theorem}
\newtheorem*{thm*}{Theorem}
\newtheorem{prop}[thm]{Proposition}
\newtheorem*{prop*}{Proposition}
\newtheorem{lem}[thm]{Lemma}

\theoremstyle{definition}
\newtheorem{dfn}{Definition}
\newtheorem*{ex}{Example}

\newcommand{\ZX}{\textsc{zx}}

\newcommand{\abs}[1]{\left| #1 \right|}

\newcommand{\ket}[1]{\left| #1 \right>} % Dirac bras
\newcommand{\bra}[1]{\left< #1 \right|} % Dirac kets
\renewcommand{\t}[1]{\ensuremath{^{\otimes #1}}}

\usepackage[svgnames]{xcolor} % if want to use svgnames option, package must be loaded *before* tikz

\usepackage{tikz}
\usetikzlibrary{shapes.geometric} % required for ellipse
\usetikzlibrary{shapes.misc} % required for rounded rectangle
\usetikzlibrary{shapes.symbols} % required for clou

\pgfdeclarelayer{edgelayer}
\pgfdeclarelayer{nodelayer}
\pgfsetlayers{edgelayer,nodelayer,main}

\tikzstyle{none}=[inner sep=0pt]
\tikzstyle{rn}=[circle,fill=Red,draw=Black,line width=0.8 pt,minimum size=5pt,inner sep=0pt]
\tikzstyle{gn}=[circle,fill=Lime,draw=Black,line width=0.8 pt,minimum size=5pt,inner sep=0pt]
\tikzstyle{Hadamard}=[rectangle,fill=Yellow,draw=Black,minimum size=8pt,inner sep=0pt,label={center:$\scriptstyle\mathrm{H}$}]
\tikzstyle{gphase}=[rounded rectangle,rounded rectangle arc length=90,fill=Lime!20,inner sep=2pt]
\tikzstyle{rphase}=[rounded rectangle,rounded rectangle arc length=90,fill=Red!20,inner sep=2pt]
\tikzstyle{normalrect}=[rectangle,fill=white,draw=black,minimum height=10pt,minimum width=12pt,inner sep=0pt]

\tikzstyle{bn}=[circle,fill=black,draw=black,inner sep=1pt]
\tikzstyle{bigcloud}=[cloud,fill=white,draw=black]

\newcommand{\effect}[1]{\begin{tikzpicture}[baseline=.1cm]
	\begin{pgfonlayer}{nodelayer}
		\node [style=none] (0) at (0, 0) {};
		\node [style={#1}] (1) at (0, 0.25) {};
	\end{pgfonlayer}
	\begin{pgfonlayer}{edgelayer}
		\draw [] (0.center) to (1);
	\end{pgfonlayer}
\end{tikzpicture}}

\newcommand{\state}[1]{\begin{tikzpicture}[baseline=-.1cm]
	\begin{pgfonlayer}{nodelayer}
		\node [style={#1}] (0) at (0, 0) {};
		\node [style=none] (1) at (0, 0.25) {};
	\end{pgfonlayer}
	\begin{pgfonlayer}{edgelayer}
		\draw (0) to (1.center);
	\end{pgfonlayer}
\end{tikzpicture}}

\newcommand{\phase}[1]{\begin{tikzpicture}[baseline=-.1cm]
	\begin{pgfonlayer}{nodelayer}
		\node [style=none] (0) at (0, 0.2) {};
		\node [style={#1}] (1) at (0, -0) {};
		\node [style=none] (2) at (0, -0.2) {};
	\end{pgfonlayer}
	\begin{pgfonlayer}{edgelayer}
		\draw [] (0.center) to (1);
		\draw [] (1) to (2.center);
	\end{pgfonlayer}
\end{tikzpicture}}

\newcommand{\splitnode}[0]{\begin{tikzpicture}[baseline=-0.1cm]
	\begin{pgfonlayer}{nodelayer}
		\node [style=gn] (0) at (0, -0) {};
		\node [style=none] (1) at (0, -0.25) {};
		\node [style=none] (2) at (-0.25, 0.25) {};
		\node [style=none] (3) at (0.25, 0.25) {};
	\end{pgfonlayer}
	\begin{pgfonlayer}{edgelayer}
		\draw [bend right=15] (2.center) to (0);
		\draw [bend left=15] (3.center) to (0);
		\draw (0) to (1.center);
	\end{pgfonlayer}
\end{tikzpicture}}

\newcommand{\joinnode}[0]{\begin{tikzpicture}[baseline=-0.1cm]
	\begin{pgfonlayer}{nodelayer}
		\node [style=gn] (0) at (0, -0) {};
		\node [style=none] (1) at (0, 0.25) {};
		\node [style=none] (2) at (0.25, -0.25) {};
		\node [style=none] (3) at (-0.25, -0.25) {};
	\end{pgfonlayer}
	\begin{pgfonlayer}{edgelayer}
		\draw [bend right=15] (2.center) to (0);
		\draw [bend left=15] (3.center) to (0);
		\draw (0) to (1.center);
	\end{pgfonlayer}
\end{tikzpicture}}

\newcommand{\Hadamard}[0]{\begin{tikzpicture}[baseline=-0.1cm]
	\begin{pgfonlayer}{nodelayer}
		\node [style=Hadamard] (0) at (0, 0) {};
		\node [style=none] (1) at (0, 0.25) {};
		\node [style=none] (2) at (0, -0.25) {};
	\end{pgfonlayer}
	\begin{pgfonlayer}{edgelayer}
		\draw (0.center) to (1);
		\draw (2.center) to (1);
	\end{pgfonlayer}
\end{tikzpicture}}

\title{The \ZX-calculus is complete for stabilizer quantum mechanics}
\author{Miriam Backens
\institute{Department of Computer Science, University of Oxford}
\email{miriam.backens@cs.ox.ac.uk}
}
\def\titlerunning{The \ZX-calculus is complete for stabilizer quantum mechanics}
\def\authorrunning{Miriam Backens}

\begin{document}

\maketitle

\begin{abstract}
The \ZX-calculus is a graphical calculus for reasoning about quantum systems and processes. It is known to be \emph{universal} for pure state qubit quantum mechanics, meaning any pure state, unitary operation and post-selected pure projective measurement can be expressed in the \ZX-calculus. The calculus is also \emph{sound}, i.e.\ any equality that can be derived graphically can also be derived using matrix mechanics. Here, we show that the \ZX-calculus is \emph{complete} for pure qubit stabilizer quantum mechanics, meaning any equality that can be derived using matrices can also be derived pictorially. The proof relies on bringing diagrams into a normal form based on graph states and local Clifford operations.
\end{abstract}

\section{Introduction}

The success of the quantum circuit notation shows the value of graphical languages for quantum processes. Using both dimensions of a sheet of paper allows the parallel composition of operations (say, several operations happening to different systems at the same time) to be separated from serial composition (say, different operations, possibly happening to the same system, but at different times). This makes graphical notation much easier for humans to read than the standard Dirac or matrix notations, where parallel and serial composition of operations are both represented in one dimension, namely a line of text. Yet the quantum circuit notation has one big disadvantage: There are no transformation rules for quantum circuit diagrams. The only way to simplify or compare quantum circuit diagrams is by translating them back into matrices, thereby losing the advantages of the graphical notation.

Unlike quantum circuit notation, the \ZX-calculus developed in \cite{coecke_interacting_2008,coecke_interacting_2011} is not just a graphical notation: It has built-in rewrite rules, which transform one diagram into a different diagram representing the same overall process. These rewrite rules make the \ZX-calculus into a formal system with non-trivial equalities between diagrams. In the following, we will thus distinguish between diagrams which are \emph{identical} -- i.e.\ they consist of the same elements combined in the same way -- and diagrams which are \emph{equal}, meaning one can be rewritten into the other. Two identical diagrams are necessarily equal to each other, but two equal diagrams may not be identical.

As a formal system modelling pure state qubit quantum mechanics (QM), there are several properties the \ZX-calculus must have to be useful. One of these is \emph{universality}: the question whether any pure state, unitary operator, or post-selected measurement can be represented by a \ZX-calculus diagram. The \ZX-calculus is indeed universal \cite{coecke_interacting_2011}. A second important property is \emph{soundness}: can any equality which can be derived in the \ZX-calculus also be derived using other formalisms, such as matrix mechanics? By considering the rewrite rules one-by-one, it is not too difficult to show that the \ZX-calculus is sound \cite{coecke_interacting_2011}. As a result of this, the \ZX-calculus can be used to analyse a variety of questions, e.g.\ quantum non-locality \cite{coecke_strong_2012} and the verification of measurement-based quantum computations \cite{coecke_interacting_2011,duncan_rewriting_2010,horsman_quantum_2011}.

The converse of the soundness property is \emph{completeness}: The \ZX-calculus is complete if any equality that can be derived using matrices can also be derived graphically. It has been conjectured that the \ZX-calculus is not complete for general pure state qubit QM, but in this paper we show that it is complete for qubit stabilizer quantum mechanics. Stabilizer QM is an extensively studied part of quantum theory, which can be operationally described as the fragment of pure state QM where the only allowed operations are preparations or measurements in the computational basis and unitary transformations belonging to the Clifford group. While stabilizer quantum computation is significantly less powerful than general quantum computation -- it can be efficiently simulated on classical computers and is provably less powerful than even general classical computation \cite{aaronson_improved_2004} -- stabilizer QM is nevertheless of central importance in areas such as error-correcting codes \cite{nielsen_quantum_2010} or measurement-based quantum computation \cite{raussendorf_one-way_2001}, and it is non-local.

A pure stabilizer state on $n$ qubits is a state that can be created by applying some Clifford unitary to the state $\ket{0}\t{n}$. Graph states are a special class of stabilizer states, whose entanglement structure can be described by a simple undirected graph. In the \ZX-calculus, graph states have a particularly elegant representation \cite{duncan_graph_2009}. Furthermore, any stabilizer state is equivalent to some graph state under local Clifford operations, which are tensor products of single qubit Clifford operators \cite{van_den_nest_graphical_2004}. The first part of our completeness proof is a proof that this equivalence also holds in the \ZX-calculus, i.e.\ there is a non-unique normal form for stabilizer state diagrams consisting of a graph state diagram and local Clifford operators. Based on work by Elliott et al. \cite{elliott_graphical_2008}, we then show that even though this normal form is not unique, there is a straightforward algorithm for testing equality of diagrams given in this form. In particular, this algorithm shows that two diagrams are equal if and only if they correspond to the same quantum mechanical state. By the Choi-Jamio{\l}kowski isomorphism, this result extends to diagrams which represent not states but operators. Thus, for any pair of \ZX-calculus diagrams representing the same state or operator in stabilizer QM, the equality testing algorithm can be used to construct a sequence of rewrites obeying the rules of the calculus, which shows that the diagrams are equal. But this is just the definition of completeness, proving that the \ZX-calculus for stabilizer QM is complete.

The basic definitions and properties of stabilizer quantum mechnics are given in section \ref{s:stabilizer}. In section \ref{s:ZX-calculus}, the elements and rules of the \ZX-calculus are laid out. Section \ref{s:ZX_graph_states} contains the definition of graph state diagrams and the normal form, as well as the proof that any stabilizer state diagram can be brought into normal form. The completeness proof can be found in section \ref{s:completeness}, followed by an example in section \ref{s:example} and conclusions in section \ref{s:conclusions}.
\section{Stabilizer quantum mechanics}\label{s:stabilizer}

\subsection{The Pauli group and the Clifford group}\label{s:Pauli_Clifford}

The Pauli operators
 \[
  X = \begin{pmatrix}0&1\\1&0\end{pmatrix},\quad Y = \begin{pmatrix}0&-i\\i&0\end{pmatrix},\quad\text{and}\quad Z = \begin{pmatrix}1&0\\0&-1\end{pmatrix}
 \]
have a central role in quantum mechanics because, together with the identity, they form a basis for all single qubit unitaries under linear combinations. Under multiplication, this set of operators gives rise to the following group.

\begin{dfn}
 The \emph{Pauli group} $P_1$ is the closure of the set $\{I,X,Y,Z\}$ under multiplication. It consists of the identity and Pauli matrices with multiplicative factors $\{\pm 1,\pm i\}$. This definition generalises to multiple qubits as follows: The Pauli group on $n$ qubits, $P_n$, consists of all tensor products of Pauli and identity matrices with phase factors $\{\pm 1, \pm i\}$, i.e.
 \[
  P_n = \Big\{\alpha g_1\otimes g_2\otimes\ldots\otimes g_n \Big| \alpha\in\{\pm 1,\pm i\}\text{ and } g_k\in\{I,X,Y,Z\}\text{ for } k=1,2,\ldots,n\Big\}.
 \]
 Elements of $P_n$ are often called Pauli products.
\end{dfn}

\begin{dfn}
 A unitary operator $U$ is said to \emph{stabilize} a quantum state $\ket{\psi}$ if $U\ket{\psi}=\ket{\psi}$.
\end{dfn}

The unitaries stabilizing a given quantum state can easily be seen to form a group. This group uniquely defines the state.

\begin{dfn}
 An $n$-qubit quantum state is called a \emph{stabilizer state} if it is stabilized by a subgroup of $P_n$.
\end{dfn}

Most unitary operators do not preserve stabilizer states, i.e.\ they map some stabilizer states to non-stabilizer states or conversely. Yet there are some unitary operators which map stabilizer states to stabilizer states. These operators form the Clifford group.

\begin{dfn}
 The \emph{Clifford group} on $n$ qubits, denoted $C_n$, is the group of operators which normalize the Pauli group, i.e. $C_n = \{ U | \forall g\in P_n: UgU^\dagger\in P_n\}$.
\end{dfn}

Any $n$-qubit stabilizer state can be expressed as $U\ket{0}\t{n}$ for some (non-unique) $U\in C_n$. It can furthermore be shown that the Clifford group is generated by two single qubit operators and one two-qubit operator \cite{nielsen_quantum_2010}, namely the phase operator $S = \left(\begin{smallmatrix}1&0\\0&i\end{smallmatrix}\right)$, the Hadamard operator $H = \frac{1}{\sqrt{2}}\left(\begin{smallmatrix}1&1\\1&-1\end{smallmatrix}\right)$ and the controlled-NOT operator
\[
 \Lambda X = \begin{pmatrix}1&0&0&0\\0&1&0&0\\0&0&0&1\\0&0&1&0\end{pmatrix}.
\]
Ignoring global phases, the group $C_1$ of single qubit Clifford unitaries has 24 elements. It is generated by the phase and Hadamard operators, or, alternatively, by $R_Z$ and $R_X$, where $R_Z = S$ and $R_X = HSH$.

\begin{dfn}\label{dfn:local_Clifford_group}
 The \emph{local Clifford group} on $n$ qubits, $C_1\t{n}$, consists of all $n$-fold tensor products of single qubit Clifford operators.
\end{dfn}

The Clifford group contains all unitary operators that map stabilizer states to stabilizer states. To generate all linear operators which do the same, we must also allow measurements whose results are stabilizer states. Like any $n$-qubit stabilizer state can be expressed as a Clifford unitary applied to the state $\ket{0}\t{n}$, any measurement in stabilizer quantum mechanics can be realised by applying a Clifford unitary, followed by a measurement of some number of qubits in the computational (or Z-) basis $\{\ket{0},\ket{1}\}$. Thus, stabilizer quantum mechanics encompasses the following three types of operations: preparation of qubits in the state $\ket{0}$, Clifford unitaries, and measurements in the computational basis.

\subsection{Graph states}\label{s:graph_states}

An important subset of the stabilizer states are the graph states, which consist of a number of qubits entangled together according to the structure of a mathematical graph.

\begin{dfn}
 A finite \emph{graph} is a pair $G=(V,E)$ where $V$ is a finite set of vertices and $E$ is a collection of edges, which are denoted by pairs of vertices. A graph is \emph{undirected} if its edges are unordered pairs of vertices. It is \emph{simple} if it has no self-loops and there is at most one edge connecting any two vertices.
\end{dfn}

In the following, unless stated otherwise, all graphs will be assumed to be undirected and simple. For such graphs, the collection of edges is in fact a set (as opposed to, say, a multi-set) and each edge is an unordered set of size two (rather than a tuple). For an $n$-vertex graph, we will often take $V=\{1,2,\ldots,n\}$.

\begin{dfn}
 A simple undirected graph $G$ with $n=\abs{V}$ vertices can be described by a symmetric $n$ by $n$ matrix $\theta$ with binary entries such that $\theta_{ij}=1$ if and only if there is an edge connecting vertices $i$ and $j$. This matrix is known as the \emph{adjacency matrix}.
\end{dfn}

\begin{dfn}\label{dfn:graph_state}
 Given a graph $G=(V,E)$ with $n=\abs{V}$ vertices and adjacency matrix $\theta$, the corresponding \emph{graph state} $\ket{G}$ is the $n$-qubit state whose stabilizer subgroup is generated by the operators
 \[
  X_v\otimes\bigotimes_{u\in V} Z_u^{\theta_{uv}} \quad\text{for all } v\in V.
 \]
 Here, subscripts indicate to which qubit the operator is applied.
\end{dfn}

All graph states are pure stabilizer states by definition. On the other hand, it is obvious that not all stabilizer states are graph states. Yet there exists an interesting relationship between arbitrary stabilizer states and graph states. Consider the equivalence relation on stabilizer states given by the local Clifford group.

\begin{dfn}
 Two $n$-qubit stabilizer states $\ket{\psi}$ and $\ket{\phi}$ are \emph{equivalent under local Clifford operations} if there exists $U\in C_1\t{n}$ such that $\ket{\psi}=U\ket{\phi}$.
\end{dfn}

\begin{thm}[\cite{van_den_nest_graphical_2004}]\label{thm:stabilizer_graph_state}
 Any pure stabilizer state is equivalent to some graph state under local Clifford operations, i.e.\ any $n$-qubit stabilizer state $\ket{\psi}$ can be written, not generally uniquely, as $U\ket{G}$, where $U\in C_1\t{n}$ and $\ket{G}$ is an $n$-qubit graph state.
\end{thm}

A single stabilizer state may well be equivalent to more than one graph state under local Clifford operations. To organize these equivalence classes we require the following definition and theorem.

\begin{dfn}
 Let $G=(V,E)$ be a graph. The \emph{local complementation about the vertex $v$} is the operation that inverts the subgraph generated by the neighbourhood of $v$ (but not including $v$ itself). Formally, a local complementation about $v\in V$ sends $G$ to the graph
 \[
  G\star v = \left(V,E \triangle \big\{\{b,c\}\big|\{b,v\},\{c,v\}\in E\wedge b\neq c\big\}\right),
 \]
 where $\triangle$ denotes the symmetric set difference, i.e.\ $A\triangle B$ contains all elements that are contained either in $A$ or in $B$ but not in both.
\end{dfn}

\begin{ex}
 Consider the line graph on four vertices. Applying local complementations about vertex 3 and then vertex 2 yields the following sequence of graphs:
 \begin{center}
  \begin{tikzpicture}
	\begin{pgfonlayer}{nodelayer}
		\node [style=bn,label={left:$1$}] (0) at (-3, 0.25) {};
		\node [style=bn,label={right:$2$}] (1) at (-2.5, 0.25) {};
		\node [style=bn,label={right:$3$}] (2) at (-2.5, -0.25) {};
		\node [style=bn,label={left:$4$}] (3) at (-3, -0.25) {};
		\node [style=none] (4) at (-1.5, -0) {$\mapsto$};
		\node [style=bn,label={left:$1$}] (5) at (-0.5, 0.25) {};
		\node [style=bn,label={right:$2$}] (6) at (0, 0.25) {};
		\node [style=bn,label={left:$4$}] (7) at (-0.5, -0.25) {};
		\node [style=bn,label={right:$3$}] (8) at (0, -0.25) {};
		\node [style=bn,label={left:$4$}] (9) at (2, -0.25) {};
		\node [style=none] (10) at (1, -0) {$\mapsto$};
		\node [style=bn,label={right:$3$}] (11) at (2.5, -0.25) {};
		\node [style=bn,label={right:$2$}] (12) at (2.5, 0.25) {};
		\node [style=bn,label={left:$1$}] (13) at (2, 0.25) {};
	\end{pgfonlayer}
	\begin{pgfonlayer}{edgelayer}
		\draw (0) to (1);
		\draw (1) to (2);
		\draw (2) to (3);
		\draw (5) to (6);
		\draw (6) to (8);
		\draw (7) to (8);
		\draw (7) to (6);
		\draw (13) to (12);
		\draw (12) to (11);
		\draw (9) to (12);
		\draw (13) to (11);
		\draw (13) to (9);
	\end{pgfonlayer}
  \end{tikzpicture}
 \end{center}
\end{ex}

\begin{thm}[\cite{van_den_nest_graphical_2004}]\label{thm:graph_states_LC}
 Two graph states on the same number of qubits are equivalent under local Clifford operations if and only if there is a sequence of local complementations that transforms one graph into the other.
\end{thm}
\section{The \ZX-calculus for stabilizer theory}\label{s:ZX-calculus}

\subsection{Categorical quantum mechanics}

In 2004, Abramsky and Coecke introduced a formalism for describing quantum mechanics using category theory \cite{abramsky_categorical_2004}. This formalism gives rise to a graphical calculus, i.e.\ a graphical representation of quantum states and operations which can be manipulated according to some set of rules. Selinger shows in \cite{selinger_dagger_2007} that this graphical calculus is indeed equivalent to the equational reasoning in dagger compact closed categories, which are the category theoretical framework for quantum mechanics.

The graphical calculus is interesting not just because it can make computations more straightforward for humans to follow, but also because it allows mechanised reasoning, e.g. using a software system like \texttt{Quantomatic} \cite{quantomatic}.

There are different graphical calculi for categorical quantum mechanics; the one we are using here is that of Coecke and Duncan in \cite{coecke_interacting_2008,coecke_interacting_2011}. It is based on maps in the computational (or Z-) basis $\{\ket{0},\ket{1}\}$ and the complementary X-basis $\{\ket{+},\ket{-}\}$, and therefore known as the \ZX-calculus. In this work, we do not consider the \ZX-calculus for all of quantum mechanics, but only the subcategory that represents pure state stabilizer quantum mechanics. The elements of stabilizer \ZX-calculus diagrams are introduced in section \ref{s:ZX_elements} and the rules for manipulating these diagrams in section \ref{s:rules}. In section \ref{s:properties}, we discuss properties of the \ZX-calculus as a formal system.

\subsection{The \ZX-calculus elements}
\label{s:ZX_elements}

The diagrams of the \ZX-calculus consist of nodes connected by edges and are read from bottom to top. Some edges may only be connected to a node at one end, these are considered to be inputs (if the open end is pointing down) or outputs (if the open end is pointing up) for the whole diagram. There are three types of nodes:
\begin{mitem}
 \item green nodes with $n$ inputs and $m$ outputs and a phase $\alpha\in\{0,\pi/2,\pi,-\pi/2\}$, representing the maps
   \begin{center}
    \begin{tikzpicture}
	\begin{pgfonlayer}{nodelayer}
		\node [style=none] (0) at (-1.8, 0.5) {};
		\node [style=none] (1) at (-1.5, 0.45) {$\ldots$};
		\node [style=none] (2) at (-1.2, 0.5) {};
		\node [style=gn,label={[gphase]right:$\alpha$}] (3) at (-1.5, -0) {};
		\node [style=none] (4) at (1.25, -0) {$= \; \begin{cases} \ket{0}\t{n}\mapsto\ket{0}\t{m} \\ \ket{1}\t{n}\mapsto e^{i\alpha}\ket{1}\t{m}, \end{cases}$};
		\node [style=none] (5) at (-1.8, -0.5) {};
		\node [style=none] (6) at (-1.5, -0.45) {$\ldots$};
		\node [style=none] (7) at (-1.2, -0.5) {};
	\end{pgfonlayer}
	\begin{pgfonlayer}{edgelayer}
		\draw [bend left=15] (3) to (7.center);
		\draw [bend right=15] (3) to (5.center);
		\draw [bend right=15] (0.center) to (3);
		\draw [bend right=15] (3) to (2.center);
	\end{pgfonlayer}
    \end{tikzpicture}
   \end{center}
 \item red nodes with $n$ inputs and $m$ outputs and a phase $\alpha\in\{0,\pi/2,\pi,-\pi/2\}$, representing the maps
   \begin{center}
    \begin{tikzpicture}
	\begin{pgfonlayer}{nodelayer}
		\node [style=none] (0) at (-1.8, 0.5) {};
		\node [style=none] (1) at (-1.5, 0.45) {$\ldots$};
		\node [style=none] (2) at (-1.2, 0.5) {};
		\node [style=rn,label={[rphase]right:$\alpha$}] (3) at (-1.5, -0) {};
		\node [style=none] (4) at (1.35, -0) {$= \; \begin{cases} \ket{+}\t{n}\mapsto\ket{+}\t{m} \\ \ket{-}\t{n}\mapsto e^{i\alpha}\ket{-}\t{m}, \end{cases}$};
		\node [style=none] (5) at (-1.8, -0.5) {};
		\node [style=none] (6) at (-1.5, -0.45) {$\ldots$};
		\node [style=none] (7) at (-1.2, -0.5) {};
	\end{pgfonlayer}
	\begin{pgfonlayer}{edgelayer}
		\draw [bend left=15] (3) to (7.center);
		\draw [bend right=15] (3) to (5.center);
		\draw [bend right=15] (0.center) to (3);
		\draw [bend right=15] (3) to (2.center);
	\end{pgfonlayer}
    \end{tikzpicture}
   \end{center}
 where $\ket{\pm}=\frac{1}{\sqrt{2}}\left(\ket{0}\pm\ket{1}\right)$, and
 \item yellow nodes with one input and one output, labelled $H$, representing the Hadamard operator
   \begin{center}
    \begin{tikzpicture}
	\begin{pgfonlayer}{nodelayer}
		\node [style=none] (0) at (-1, 0.5) {};
		\node [style=Hadamard] (1) at (-1, -0) {};
		\node [style=none] (2) at (0.75, -0) {$= \; \begin{cases} \ket{0}\mapsto\ket{+}\\ \ket{1}\mapsto\ket{-}. \end{cases}$};
		\node [style=none] (3) at (-1, -0.5) {};
	\end{pgfonlayer}
	\begin{pgfonlayer}{edgelayer}
		\draw (0.center) to (1);
		\draw (1) to (3.center);
	\end{pgfonlayer}
    \end{tikzpicture}
   \end{center}
\end{mitem}
If a node has phase zero, the phase label is usually left out. Note that \phase{gn,label={[gphase]right:$\pi/2$}} corresponds to $R_Z$ and \phase{rn,label={[rphase]right:$\pi/2$}} corresponds to $R_X$. Red or green nodes with one input and one output are occasionally called \emph{phase operators}. Diagrams with no inputs denote states, in particular \state{gn} is $\ket{+}$ and \state{rn} denotes $\ket{0}$.

The Hermitian conjugate of a diagram, denoted by a superscript $\dagger$, is the diagram that results by interchanging the roles of inputs and outputs in the original diagram (i.e. mirroring the diagram in a horizontal line), and flipping the signs of all phases.

Measurements are represented in the \ZX-calculus in post-selected form. A diagram $D$ with $m$ inputs and no outputs is called an \emph{effect} on $m$ qubits. The interpretation of the effect $D$ is that of having found state $D^\dagger$ upon performing some apropriate measurement.

The \ZX-calculus for all of quantum mechanics has the same elements as the calculus for stabilizer quantum mechanics, the only difference being that arbitrary phases $\alpha$ in the interval $-\pi<\alpha\leq\pi$ are allowed.

\subsection{Rules of the \ZX-calculus}
\label{s:rules}

The diagrams of the \ZX-calculus satisfy a number of rewrite rules, i.e.\ purely graphical rules for manipulating diagrams. All rules are also true with the colours red and green reversed, or with inputs and outputs reversed. Subdiagrams with no inputs or outputs are global phase and normalization factors. Since global phases have no physical effect and the normalization of pure states is fixed, we will ignore them. The nodes
\begin{tikzpicture}[baseline=-.1cm]
	\begin{pgfonlayer}{nodelayer}
		\node [style=gn,label={[gphase]right:$\pi$}] (0) at (0, -0) {};
	\end{pgfonlayer}
\end{tikzpicture}
and
\begin{tikzpicture}[baseline=-.1cm]
	\begin{pgfonlayer}{nodelayer}
		\node [style=rn,label={[rphase]right:$\pi$}] (0) at (0, -0) {};
	\end{pgfonlayer}
\end{tikzpicture}
denote a global factor of 0, these can of course not be ignored.

\textbf{Only the topology matters}: As long as the topology of the diagram remains the same, nodes and lines can be moved around freely.

\textbf{Spider rule and identity rules}: Two adjacent nodes of the same colour merge, their phases add. A node with phase zero and exactly two incident edges can be removed.
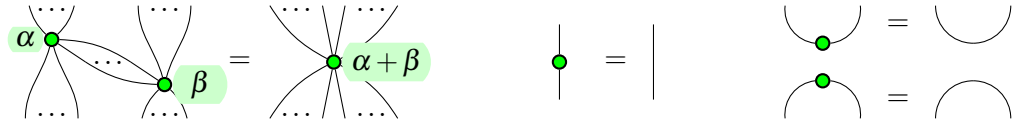
\begin{center}
   \begin{tikzpicture}
	\begin{pgfonlayer}{nodelayer}
		\node [style=none] (0) at (-5.8, 0.75) {};
		\node [style=none] (1) at (-5.5, 0.7) {$\ldots$};
		\node [style=none] (2) at (-5.2, 0.75) {};
		\node [style=none] (3) at (-4.35, 0.75) {};
		\node [style=none] (4) at (-4, 0.7) {$\ldots$};
		\node [style=none] (5) at (-3.65, 0.75) {};
		\node [style=none] (6) at (-2.6, 0.75) {};
		\node [style=none] (7) at (-2.25, 0.7) {$\ldots$};
		\node [style=none] (8) at (-1.9, 0.75) {};
		\node [style=none] (9) at (-1.6, 0.75) {};
		\node [style=none] (10) at (-1.25, 0.7) {$\ldots$};
		\node [style=none] (11) at (-0.9, 0.75) {};
		\node [style=none] (12) at (1.25, 0.5) {};
		\node [style=none] (13) at (2.5, 0.5) {};
		\node [style=none] (14) at (4.25, 0.75) {};
		\node [style=none] (15) at (5.25, 0.75) {};
		\node [style=none] (16) at (6.25, 0.75) {};
		\node [style=none] (17) at (7.25, 0.75) {};
		\node [style=gn,label={[gphase]left:$\alpha$}] (18) at (-5.5, 0.3) {};
		\node [style=none] (19) at (5.75, 0.5) {$=$};
		\node [style=gn] (20) at (4.75, 0.25) {};
		\node [style=none] (21) at (6.75, 0.25) {};
		\node [style=none] (22) at (-5.05, -0) {};
		\node [style=none] (23) at (-4.75, -0) {$\ldots$};
		\node [style=none] (24) at (-4.45, -0) {};
		\node [style=none] (25) at (-3, -0) {$=$};
		\node [style=gn,label={[gphase]right:$\alpha+\beta$}] (26) at (-1.75, -0) {};
		\node [style=gn] (27) at (1.25, -0) {};
		\node [style=none] (28) at (2, -0) {$=$};
		\node [style=gn] (29) at (4.75, -0.25) {};
		\node [style=none] (30) at (6.75, -0.25) {};
		\node [style=gn,label={[gphase]right:$\beta$}] (31) at (-4, -0.3) {};
		\node [style=none] (32) at (5.75, -0.5) {$=$};
		\node [style=none] (33) at (-2.6, -0.75) {};
		\node [style=none] (34) at (-2.25, -0.7) {$\ldots$};
		\node [style=none] (35) at (-1.9, -0.75) {};
		\node [style=none] (36) at (-1.6, -0.75) {};
		\node [style=none] (37) at (-1.25, -0.7) {$\ldots$};
		\node [style=none] (38) at (-0.9, -0.75) {};
		\node [style=none] (39) at (1.25, -0.5) {};
		\node [style=none] (40) at (2.5, -0.5) {};
		\node [style=none] (41) at (4.25, -0.75) {};
		\node [style=none] (42) at (5.25, -0.75) {};
		\node [style=none] (43) at (6.25, -0.75) {};
		\node [style=none] (44) at (7.25, -0.75) {};
		\node [style=none] (45) at (-5.85, -0.75) {};
		\node [style=none] (46) at (-5.5, -0.7) {$\ldots$};
		\node [style=none] (47) at (-5.15, -0.75) {};
		\node [style=none] (48) at (-4.3, -0.75) {};
		\node [style=none] (49) at (-4, -0.7) {$\ldots$};
		\node [style=none] (50) at (-3.7, -0.75) {};
	\end{pgfonlayer}
	\begin{pgfonlayer}{edgelayer}
		\draw (18) to (22.center);
		\draw [bend right=45] (14.center) to (20);
		\draw [bend left=45] (29) to (42.center);
		\draw [bend right=15] (26) to (33.center);
		\draw (13.center) to (40.center);
		\draw (9.center) to (26);
		\draw (26) to (36.center);
		\draw [bend left=45] (30.center) to (44.center);
		\draw [bend right=45] (20) to (15.center);
		\draw [bend right=45] (21.center) to (17.center);
		\draw [in=75, out=270] (5.center) to (31);
		\draw [in=90, out=255] (18) to (45.center);
		\draw [bend right=15] (22.center) to (31);
		\draw [bend left=45] (41.center) to (29);
		\draw [bend right=15] (6.center) to (26);
		\draw [in=90, out=-75] (18) to (47.center);
		\draw (24.center) to (31);
		\draw [bend left=15] (18) to (24.center);
		\draw [bend right=15] (0.center) to (18);
		\draw [bend left=15] (26) to (38.center);
		\draw (27) to (39.center);
		\draw [bend left=15] (31) to (50.center);
		\draw [bend right=15] (18) to (2.center);
		\draw (12.center) to (27);
		\draw [bend left=15] (11.center) to (26);
		\draw [bend right=15] (31) to (48.center);
		\draw (26) to (35.center);
		\draw [bend left=45] (43.center) to (30.center);
		\draw [in=105, out=270] (3.center) to (31);
		\draw (8.center) to (26);
		\draw [bend right=45] (16.center) to (21.center);
	\end{pgfonlayer}
   \end{tikzpicture}
\end{center}

\textbf{Bialgebra law, Hopf law and copying}: The bialgebra law allows a certain pattern of two red and two green nodes to be replaced by just one red and green node. If two nodes of different colours are connected by exactly two edges, then by the Hopf law those edges can be removed. Finally, a node of one colour with one input and two outputs copies the zero phase state of the other colour.
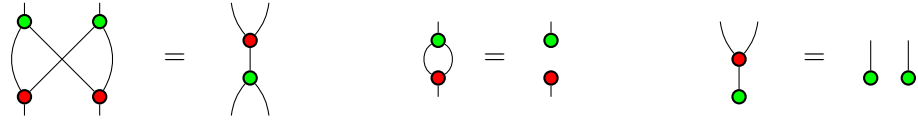
\begin{center}
 \begin{tikzpicture}
	\begin{pgfonlayer}{nodelayer}
		\node [style=none] (0) at (-6, 0.75) {};
		\node [style=none] (1) at (-5, 0.75) {};
		\node [style=none] (2) at (-3.25, 0.75) {};
		\node [style=none] (3) at (-2.75, 0.75) {};
		\node [style=none] (4) at (-0.5, 0.5) {};
		\node [style=none] (5) at (1, 0.5) {};
		\node [style=gn] (6) at (-6, 0.5) {};
		\node [style=gn] (7) at (-5, 0.5) {};
		\node [style=gn] (8) at (3.5, -0.5) {};
		\node [style=rn] (9) at (-3, 0.25) {};
		\node [style=gn] (10) at (-0.5, 0.25) {};
		\node [style=gn] (11) at (1, 0.25) {};
		\node [style=gn] (12) at (5.25, -0.25) {};
		\node [style=gn] (13) at (5.75, -0.25) {};
		\node [style=none] (14) at (-4, -0) {$=$};
		\node [style=none] (15) at (0.25, -0) {$=$};
		\node [style=rn] (16) at (3.5, -0) {};
		\node [style=none] (17) at (4.5, -0) {$=$};
		\node [style=gn] (18) at (-3, -0.25) {};
		\node [style=rn] (19) at (-0.5, -0.25) {};
		\node [style=rn] (20) at (1, -0.25) {};
		\node [style=none] (21) at (5.25, 0.25) {};
		\node [style=none] (22) at (5.75, 0.25) {};
		\node [style=rn] (23) at (-6, -0.5) {};
		\node [style=rn] (24) at (-5, -0.5) {};
		\node [style=none] (25) at (3.25, 0.5) {};
		\node [style=none] (26) at (3.75, 0.5) {};
		\node [style=none] (27) at (-6, -0.75) {};
		\node [style=none] (28) at (-5, -0.75) {};
		\node [style=none] (29) at (-3.25, -0.75) {};
		\node [style=none] (30) at (-2.75, -0.75) {};
		\node [style=none] (31) at (-0.5, -0.5) {};
		\node [style=none] (32) at (1, -0.5) {};
	\end{pgfonlayer}
	\begin{pgfonlayer}{edgelayer}
		\draw (23) to (27.center);
		\draw [bend left=15] (16) to (25.center);
		\draw [bend left=15] (3.center) to (9);
		\draw [bend left] (7) to (24);
		\draw [bend left=60] (10) to (19);
		\draw (24) to (28.center);
		\draw (5.center) to (11);
		\draw (9) to (18);
		\draw (13) to (22.center);
		\draw (7) to (23);
		\draw [bend right=15] (16) to (26.center);
		\draw [bend right=60] (10) to (19);
		\draw (1.center) to (7);
		\draw (19) to (31.center);
		\draw (12) to (21.center);
		\draw (4.center) to (10);
		\draw [bend right] (6) to (23);
		\draw (20) to (32.center);
		\draw (6) to (24);
		\draw (0.center) to (6);
		\draw (8) to (16);
		\draw [bend right=15] (18) to (29.center);
		\draw [bend right=15] (2.center) to (9);
		\draw [bend left=15] (18) to (30.center);
	\end{pgfonlayer}
 \end{tikzpicture}
\end{center}

\textbf{$\pi$-copy rule, $\pi$-commutation and colour change}: A $\pi$ phase operator is copied by a node of the other colour. It can also be moved past any phase operator of the other colour, flipping the sign of that phase in the process. The Hadamard gate changes the colour of nodes when it is applied to each input and output. From this and the identity rule we can deduce that the Hadamard gate is self-inverse.
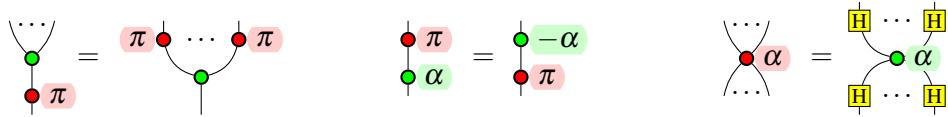
\begin{center}
 \begin{tikzpicture}
	\begin{pgfonlayer}{nodelayer}
		\node [style=gn] (0) at (-6, 0) {};
		\node [style=rn,label={[rphase]right:$\pi$}] (1) at (-6, -0.5) {};
		\node [style=none] (2) at (-6.3, 0.5) {};
		\node [style=none] (3) at (-5.7, 0.5) {};
		\node [style=none] (4) at (-6, -0.75) {};
		\node [style=none] (5) at (-5.25, 0) {$=$};
		\node [style=gn] (6) at (-3.75, -0.25) {};
		\node [style=rn,label={[rphase]left:$\pi$}] (7) at (-4.25, 0.25) {};
		\node [style=rn,label={[rphase]right:$\pi$}] (8) at (-3.25, 0.25) {};
		\node [style=none] (9) at (-3.75, -0.75) {};
		\node [style=none] (10) at (-4.25, 0.5) {};
		\node [style=none] (11) at (-3.25, 0.5) {};
		\node [style=none] (12) at (-6, 0.45) {$\ldots $};
		\node [style=none] (13) at (-3.75, 0.25) {$\ldots $};
		\node [style=none] (14) at (-1, 0.5) {};
		\node [style=rn,label={[rphase]right:$\pi$}] (15) at (-1, 0.25) {};
		\node [style=gn,label={[gphase]right:$\alpha$}] (16) at (-1, -0.25) {};
		\node [style=none] (17) at (-1, -0.5) {};
		\node [style=none] (18) at (0, -0) {$=$};
		\node [style=none] (19) at (0.5, 0.5) {};
		\node [style=none] (20) at (0.5, -0.5) {};
		\node [style=gn,label={[gphase]right:$-\alpha$}] (21) at (0.5, 0.25) {};
		\node [style=rn,label={[rphase]right:$\pi$}] (22) at (0.5, -0.25) {};
		\node [style=rn,label={[rphase]right:$\alpha$}] (23) at (3.5, -0) {};
		\node [style=none] (24) at (3.2, 0.5) {};
		\node [style=none] (25) at (3.8, 0.5) {};
		\node [style=none] (26) at (3.2, -0.5) {};
		\node [style=none] (27) at (3.8, -0.5) {};
		\node [style=none] (28) at (4.5, -0) {$=$};
		\node [style=gn,label={[gphase]right:$\alpha$}] (29) at (5.5, -0) {};
		\node [style=Hadamard] (30) at (5, 0.5) {};
		\node [style=Hadamard] (31) at (6, 0.5) {};
		\node [style=Hadamard] (32) at (5, -0.5) {};
		\node [style=Hadamard] (33) at (6, -0.5) {};
		\node [style=none] (34) at (5, 0.75) {};
		\node [style=none] (35) at (6, 0.75) {};
		\node [style=none] (36) at (5, -0.75) {};
		\node [style=none] (37) at (6, -0.75) {};
		\node [style=none] (38) at (5.5, 0.5) {$\ldots $};
		\node [style=none] (39) at (5.5, -0.5) {$\ldots $};
		\node [style=none] (40) at (3.5, 0.45) {$\ldots $};
		\node [style=none] (41) at (3.5, -0.45) {$\ldots $};
	\end{pgfonlayer}
	\begin{pgfonlayer}{edgelayer}
		\draw (4.center) to (1);
		\draw (1) to (0);
		\draw [bend left=15] (0) to (2.center);
		\draw [bend right=15] (0) to (3.center);
		\draw (9.center) to (6);
		\draw [bend left] (6) to (7);
		\draw [bend right] (6) to (8);
		\draw (8) to (11.center);
		\draw (7) to (10.center);
		\draw (14.center) to (15);
		\draw (15) to (16);
		\draw (16) to (17.center);
		\draw (19.center) to (21);
		\draw (21) to (22);
		\draw (22) to (20.center);
		\draw [bend right=15] (24.center) to (23);
		\draw [bend right=15] (23) to (26.center);
		\draw [bend left=15] (25.center) to (23);
		\draw [bend left=15] (23) to (27.center);
		\draw (34.center) to (30);
		\draw [bend right] (30) to (29);
		\draw [bend right] (29) to (31);
		\draw (31) to (35.center);
		\draw [bend right] (29) to (32);
		\draw (32) to (36.center);
		\draw (33) to (37.center);
		\draw [bend left] (29) to (33);
	\end{pgfonlayer}
 \end{tikzpicture}
\end{center}

\textbf{Euler decomposition of the Hadamard operator}: This rule is special in that it does not have a category-theoretical meaning, but follows from the fact that any unitary single qubit operator can be decomposed as a sequence of three rotations around two orthogonal axes. The Euler decomposition rule cannot be derived from any combination of the other \ZX-calculus rules \cite{duncan_graph_2009}.
\begin{center}
 \begin{tikzpicture}
	\begin{pgfonlayer}{nodelayer}
		\node [style=Hadamard] (0) at (-0.75, -0) {};
		\node [style=rn,label={[rphase]right:$\pi/2$}] (1) at (0.75, -0) {};
		\node [style=gn,label={[gphase]right:$\pi/2$}] (2) at (0.75, -0.6) {};
		\node [style=gn,label={[gphase]right:$\pi/2$}] (3) at (0.75, 0.6) {};
		\node [style=none] (4) at (-0.75, 0.5) {};
		\node [style=none] (5) at (-0.75, -0.5) {};
		\node [style=none] (6) at (0.75, 0.85) {};
		\node [style=none] (7) at (0.75, -0.85) {};
		\node [style=none] (8) at (0, -0) {$=$};
	\end{pgfonlayer}
	\begin{pgfonlayer}{edgelayer}
		\draw (4.center) to (0);
		\draw (0) to (5.center);
		\draw (6.center) to (3);
		\draw (3) to (1);
		\draw (1) to (2);
		\draw (2) to (7.center);
	\end{pgfonlayer}
 \end{tikzpicture}
\end{center}

\subsection{The \ZX-calculus as a formal system}
\label{s:properties}

The elements of the \ZX-calculus, together with their interpretations in terms of matrix mechanics, allow quantum mechanical states and operators to be expressed in diagrammatic form. The rules given in section \ref{s:rules} turn this graphical notation into a formal system in its own right, justifying the name ``calculus''. As such, there are a number of properties of the \ZX-calculus that interest us:
\begin{mitem}
 \item \emph{Universality}: Is any stabilizer state or stabilizer operation expressible as a \ZX-calculus diagram?
 \item \emph{Soundness}: Does any equation derived in the \ZX-calculus hold true when translated back into matrix mechanics?
 \item \emph{Completeness}: Is any equation between two \ZX-calculus diagrams which is true when translated into matrix mechanics derivable using the rules of the \ZX-calculus?
\end{mitem}
Of these properties, soundness is clearly the most important, as a new formalism is of little use if it disagrees with the original. Fortunately, it can easily be checked that the rules of the \ZX-calculus are sound, from which it follows that any equality derived in the \ZX-calculus is true.

Also, the \ZX-calculus is indeed universal for stabilizer quantum mechanics. To see this, note that any Clifford operator can be represented by a quantum circuit consisting of controlled-X, Hadamard and phase gates \cite{nielsen_quantum_2010}. Any pure $n$-qubit stabilizer state can be represented by a Clifford operator applied to the state $\ket{0}\t{n}$. Now both the Hadamard gate and the phase gate can be expressed in the \ZX-calculus. Furthermore, it is easy to see that
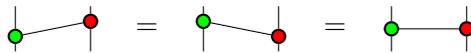
\begin{center}
 \begin{tikzpicture}
	\begin{pgfonlayer}{nodelayer}
		\node [style=gn] (0) at (-3, -0.1) {};
		\node [style=rn] (1) at (-2, 0.1) {};
		\node [style=rn] (2) at (0.5, -0.1) {};
		\node [style=gn] (3) at (-0.5, 0.1) {};
		\node [style=gn] (4) at (2, -0) {};
		\node [style=rn] (5) at (3, -0) {};
		\node [style=none] (6) at (-3, 0.3) {};
		\node [style=none] (7) at (-2, 0.3) {};
		\node [style=none] (8) at (-3, -0.3) {};
		\node [style=none] (9) at (-2, -0.3) {};
		\node [style=none] (10) at (-0.5, 0.3) {};
		\node [style=none] (11) at (0.5, 0.3) {};
		\node [style=none] (12) at (-0.5, -0.3) {};
		\node [style=none] (13) at (0.5, -0.3) {};
		\node [style=none] (14) at (-1.25, -0) {$=$};
		\node [style=none] (15) at (1.25, -0) {$=$};
		\node [style=none] (16) at (2, 0.3) {};
		\node [style=none] (17) at (3, 0.3) {};
		\node [style=none] (18) at (2, -0.3) {};
		\node [style=none] (19) at (3, -0.3) {};
	\end{pgfonlayer}
	\begin{pgfonlayer}{edgelayer}
		\draw (6.center) to (0);
		\draw (0) to (8.center);
		\draw (0) to (1);
		\draw (7.center) to (1);
		\draw (1) to (9.center);
		\draw (10.center) to (3);
		\draw (3) to (12.center);
		\draw (11.center) to (2);
		\draw (2) to (13.center);
		\draw (3) to (2);
		\draw (16.center) to (4);
		\draw (4) to (18.center);
		\draw (4) to (5);
		\draw (17.center) to (5);
		\draw (5) to (19.center);
	\end{pgfonlayer}
 \end{tikzpicture}
\end{center}
is the controlled-X gate. Thus any Clifford circuit can be translated easily into the \ZX-calculus. As
\begin{center}
 \begin{tikzpicture}
	\begin{pgfonlayer}{nodelayer}
		\node [style=rn] (0) at (-0.5, -0.1) {};
		\node [style=rn] (1) at (0.5, -0.1) {};
		\node [style=none] (2) at (-0.5, 0.15) {};
		\node [style=none] (3) at (0.5, 0.15) {};
		\node [style=none] (4) at (0, -0) {$\ldots$};
		\node [style=none] (5) at (-1.5, -0) {$\ket{0}\t{n}=$};
	\end{pgfonlayer}
	\begin{pgfonlayer}{edgelayer}
		\draw (2.center) to (0);
		\draw (3.center) to (1);
	\end{pgfonlayer}
 \end{tikzpicture}
\end{center}
we can also represent any pure stabilizer state. Postselected Z-basis measurements are given by \effect{rn} and \effect{rn,label={[rphase]right:$\pi$}}.

It is easy to show that the \ZX-calculus is complete for $C_1$, the single qubit Clifford group; in fact:
\begin{lem}\label{lem:local_Clifford_normal_form}
 Any single qubit Clifford operator can be written uniquely in one of the forms
 \begin{equation}\label{eq:local_Clifford_normal_form}
  \begin{tikzpicture}[baseline=0cm]
	\begin{pgfonlayer}{nodelayer}
		\node [style=gn,label={[gphase]right:$\alpha$}] (0) at (-1.75, 0.3) {};
		\node [style=rn,label={[rphase]right:$\beta$}] (1) at (-1.75, -0.3) {};
		\node [style=rn,label={[rphase]right:$\pi/2$}] (2) at (1, 0.6) {};
		\node [style=gn,label={[gphase]right:$\pm\pi/2$}] (3) at (1, -0) {};
		\node [style=rn,label={[rphase]right:$\gamma$}] (4) at (1, -0.6) {};
		\node [style=none] (5) at (0, -0) {or};
		\node [style=none] (6) at (-1.75, 0.55) {};
		\node [style=none] (7) at (-1.75, -0.55) {};
		\node [style=none] (8) at (1, -0.9) {};
		\node [style=none] (9) at (1, 0.9) {};
	\end{pgfonlayer}
	\begin{pgfonlayer}{edgelayer}
		\draw (6.center) to (7.center);
		\draw (9.center) to (8.center);
	\end{pgfonlayer}
  \end{tikzpicture}
 \end{equation}
 where $\alpha,\beta,\gamma\in\{0,\pi/2,\pi,-\pi/2\}$.
\end{lem}
\begin{proof}
 Rewritability follows from straightforward application of the spider rule, the $\pi$-commutation rule and the Euler decomposition rule, noting that $R_Z^{\pm 1} = R_Z^{\mp 1} Z$. Uniqueness is due to the fact that the number of distinct normal forms in \eqref{eq:local_Clifford_normal_form}, namely 24, is exactly equal to $\abs{C_1}$.
\end{proof}
In the following sections we will show that the \ZX-calculus is complete for all stabilizer quantum mechanics.

\section{Graph states in the \ZX-calculus}\label{s:ZX_graph_states}

\subsection{Graph states and local Clifford operators}

Graph states can be represented in the graphical calculus in an especially elegant way.

\begin{prop}\label{prop:graph_state}
 A graph state $\ket{G}$, where $G=(E,V)$ is an $n$-vertex graph, is represented in the graphical calculus as follows:
 \begin{mitem}
  \item for each vertex $v\in V$, a green node with one output, and
  \item for each edge $\{u,v\}\in E$, a Hadamard node connected to the green nodes representing vertices $u$ and $v$.
 \end{mitem}
\end{prop}
\begin{proof}[Proof (Sketch)]
 By definition \ref{dfn:graph_state}, the graph state determined by a graph $G=(V,E)$ with adjacency matrix $\theta$ must be an eigenstate of the operators
 \[
  X_v\otimes\bigotimes_{u\in V} Z_u^{\theta_{uv}} \quad\text{for all } v\in V.
 \]
 Now in the \ZX-calculus, $Z$ is represented by \phase{gn,label={[gphase]right:$\pi$}} and $X$ is \phase{rn,label={[rphase]right:$\pi$}}. By the spider law, green phase operators can be moved past any green nodes. Thus, by the $\pi$-copy and the colour change laws, the state defined above is indeed an eigenstate of the given operators.
\end{proof}

We will occasionally use a white ellipse labelled with the name of the graph as short-hand notation for a graph state in diagrams, i.e.
\begin{tikzpicture}[baseline=-0.1 cm]
	\begin{pgfonlayer}{nodelayer}
		\node [style=none] (0) at (-0.75, 0.3) {};
		\node [style=none] (1) at (0.75, 0.3) {};
		\node [style=none] (2) at (-0.75, 0) {};
		\node [ellipse, fill=White, draw=Black, minimum width=2 cm, style=none] (3) at (0, -0) {$G$};
		\node [style=none] (4) at (0.75, 0) {};
		\node [style=none] (5) at (0, 0.27) {$\ldots$};
	\end{pgfonlayer}
	\begin{pgfonlayer}{edgelayer}
		\draw (0.center) to (2.center);
		\draw (1.center) to (4.center);
	\end{pgfonlayer}
\end{tikzpicture}
denotes the state $\ket{G}$.

\begin{ex}
 Let $G$ be the graph
 
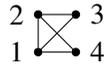
\begin{center}
  \begin{tikzpicture}
	\begin{pgfonlayer}{nodelayer}
		\node [style=bn,label={right:$4$}] (0) at (0.25, -0.25) {};
		\node [style=bn,label={right:$3$}] (1) at (0.25, 0.25) {};
		\node [style=bn,label={left:$2$}] (2) at (-0.25, 0.25) {};
		\node [style=bn,label={left:$1$}] (3) at (-0.25, -0.25) {};
	\end{pgfonlayer}
	\begin{pgfonlayer}{edgelayer}
		\draw (3) to (2);
		\draw (2) to (1);
		\draw (0) to (2);
		\draw (3) to (1);
		\draw (3) to (0);
	\end{pgfonlayer}
  \end{tikzpicture}
 \end{center}
 The corresponding graph state is the 4-qubit state whose stabilizer group is generated by the operators
\begin{align*}
 & X\otimes Z\otimes Z\otimes Z, \\
 & Z\otimes X\otimes Z\otimes Z, \\
 & Z\otimes Z\otimes X\otimes I, \text{ and} \\
 & Z\otimes Z\otimes I\otimes X.
\end{align*}
 By proposition \ref{prop:graph_state}, the corresponding diagram in the \ZX-calculus is
 
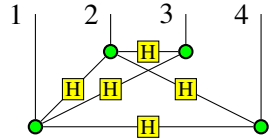
\begin{center}
  \begin{tikzpicture}
	\begin{pgfonlayer}{nodelayer}
		\node [style=gn] (0) at (1.5, -1) {};
		\node [style=gn] (1) at (0.5, -0) {};
		\node [style=gn] (2) at (-0.5, -0) {};
		\node [style=gn] (3) at (-1.5, -1) {};
		\node [style=none,label={left:$1$}] (4) at (-1.5, 0.5) {};
		\node [style=none,label={left:$2$}] (5) at (-0.5, 0.5) {};
		\node [style=none,label={left:$3$}] (6) at (0.5, 0.5) {};
		\node [style=none,label={left:$4$}] (7) at (1.5, 0.5) {};
		\node [style=Hadamard] (8) at (0, -0) {};
		\node [style=Hadamard] (9) at (-0.5, -0.5) {};
		\node [style=Hadamard] (10) at (-1, -0.5) {};
		\node [style=Hadamard] (11) at (0, -1) {};
		\node [style=Hadamard] (12) at (0.5, -0.5) {};
	\end{pgfonlayer}
	\begin{pgfonlayer}{edgelayer}
		\draw (3) to (2);
		\draw (2) to (1);
		\draw (0) to (2);
		\draw (3) to (1);
		\draw (3) to (0);
		\draw (4.center) to (3);
		\draw (5.center) to (2);
		\draw (6.center) to (1);
		\draw (7.center) to (0);
	\end{pgfonlayer}
  \end{tikzpicture}
 \end{center}
 where the vertices are rearranged so that the qubits are in the correct order. We check whether this is an eigenstate of the operator $X\otimes Z\otimes Z\otimes Z$. Indeed,
 
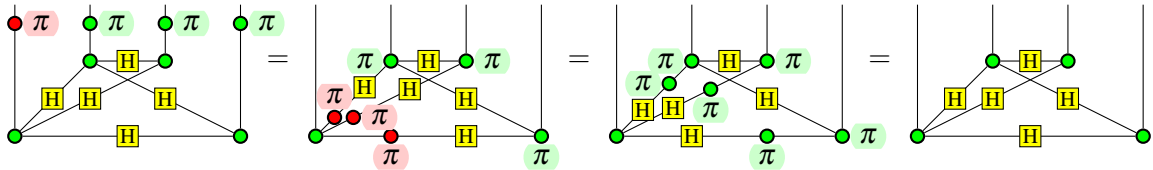
\begin{center}
  \begin{tikzpicture}
	\begin{pgfonlayer}{nodelayer}
		\node [style=gn] (0) at (-4.5, -1) {};
		\node [style=gn] (1) at (-5.5, -0) {};
		\node [style=gn] (2) at (-6.5, -0) {};
		\node [style=gn] (3) at (-7.5, -1) {};
		\node [style=none] (4) at (-7.5, 0.75) {};
		\node [style=none] (5) at (-6.5, 0.75) {};
		\node [style=none] (6) at (-5.5, 0.75) {};
		\node [style=none] (7) at (-4.5, 0.75) {};
		\node [style=Hadamard] (8) at (-6, -0) {};
		\node [style=Hadamard] (9) at (-6.5, -0.5) {};
		\node [style=Hadamard] (10) at (-7, -0.5) {};
		\node [style=Hadamard] (11) at (-6, -1) {};
		\node [style=Hadamard] (12) at (-5.5, -0.5) {};
		\node [style=rn,label={[rphase]right:$\pi$}] (13) at (-7.5, 0.5) {};
		\node [style=gn,label={[gphase]right:$\pi$}] (14) at (-6.5, 0.5) {};
		\node [style=gn,label={[gphase]right:$\pi$}] (15) at (-5.5, 0.5) {};
		\node [style=gn,label={[gphase]right:$\pi$}] (16) at (-4.5, 0.5) {};
		\node [style=none] (17) at (-4, -0) {$=$};
		\node [style=Hadamard] (18) at (-2, -0) {};
		\node [style=none] (19) at (-2.5, 0.75) {};
		\node [style=gn] (20) at (-3.5, -1) {};
		\node [style=none] (21) at (-3.5, 0.75) {};
		\node [style=gn,label={[gphase]below:$\pi$}] (22) at (-0.5, -1) {};
		\node [style=rn,label={[rphase]above:$\pi$}] (23) at (-3.25, -0.75) {};
		\node [style=Hadamard] (24) at (-2.85, -0.35) {};
		\node [style=none] (25) at (-0.5, 0.75) {};
		\node [style=gn,label={[gphase]left:$\pi$}] (26) at (-2.5, -0) {};
		\node [style=Hadamard] (27) at (-2.25, -0.375) {};
		\node [style=none] (28) at (-1.5, 0.75) {};
		\node [style=gn,label={[gphase]right:$\pi$}] (29) at (-1.5, -0) {};
		\node [style=Hadamard] (30) at (-1.5, -1) {};
		\node [style=Hadamard] (31) at (-1.5, -0.5) {};
		\node [style=rn,label={[rphase]below:$\pi$}] (32) at (-2.5, -1) {};
		\node [style=rn,label={[rphase]right:$\pi$}] (33) at (-3, -0.75) {};
		\node [style=none] (34) at (2.5, 0.75) {};
		\node [style=none] (35) at (0.5, 0.75) {};
		\node [style=Hadamard] (36) at (1.5, -1) {};
		\node [style=Hadamard] (37) at (1.25, -0.625) {};
		\node [style=Hadamard] (38) at (2, -0) {};
		\node [style=gn,label={[gphase]right:$\pi$}] (39) at (2.5, -0) {};
		\node [style=gn] (40) at (0.5, -1) {};
		\node [style=gn,label={[gphase]below:$\pi$}] (41) at (2.5, -1) {};
		\node [style=none] (42) at (3.5, 0.75) {};
		\node [style=Hadamard] (43) at (0.85, -0.65) {};
		\node [style=gn,label={[gphase]below:$\pi$}] (44) at (1.75, -0.375) {};
		\node [style=none] (45) at (0, -0) {$=$};
		\node [style=none] (46) at (1.5, 0.75) {};
		\node [style=gn,label={[gphase]left:$\pi$}] (47) at (1.5, -0) {};
		\node [style=Hadamard] (48) at (2.5, -0.5) {};
		\node [style=gn,label={[gphase]right:$\pi$}] (49) at (3.5, -1) {};
		\node [style=gn,label={[gphase]left:$\pi$}] (50) at (1.2, -0.3) {};
		\node [style=none] (51) at (4.5, 0.75) {};
		\node [style=gn] (52) at (5.5, -0) {};
		\node [style=none] (53) at (4, -0) {$=$};
		\node [style=Hadamard] (54) at (6, -0) {};
		\node [style=none] (55) at (7.5, 0.75) {};
		\node [style=gn] (56) at (6.5, -0) {};
		\node [style=Hadamard] (57) at (5.5, -0.5) {};
		\node [style=Hadamard] (58) at (5, -0.5) {};
		\node [style=gn] (59) at (7.5, -1) {};
		\node [style=none] (60) at (5.5, 0.75) {};
		\node [style=none] (61) at (6.5, 0.75) {};
		\node [style=Hadamard] (62) at (6.5, -0.5) {};
		\node [style=gn] (63) at (4.5, -1) {};
		\node [style=Hadamard] (64) at (6, -1) {};
	\end{pgfonlayer}
	\begin{pgfonlayer}{edgelayer}
		\draw (3) to (2);
		\draw (2) to (1);
		\draw (0) to (2);
		\draw (3) to (1);
		\draw (3) to (0);
		\draw (4.center) to (3);
		\draw (5.center) to (2);
		\draw (6.center) to (1);
		\draw (7.center) to (0);
		\draw (20) to (26);
		\draw (26) to (29);
		\draw (22) to (26);
		\draw (20) to (29);
		\draw (20) to (22);
		\draw (21.center) to (20);
		\draw (19.center) to (26);
		\draw (28.center) to (29);
		\draw (25.center) to (22);
		\draw (40) to (47);
		\draw (47) to (39);
		\draw (49) to (47);
		\draw (40) to (39);
		\draw (40) to (49);
		\draw (35.center) to (40);
		\draw (46.center) to (47);
		\draw (34.center) to (39);
		\draw (42.center) to (49);
		\draw (63) to (52);
		\draw (52) to (56);
		\draw (59) to (52);
		\draw (63) to (56);
		\draw (63) to (59);
		\draw (51.center) to (63);
		\draw (60.center) to (52);
		\draw (61.center) to (56);
		\draw (55.center) to (59);
	\end{pgfonlayer}
  \end{tikzpicture}
 \end{center}
 using the $\pi$-copy law and the spider rule in the first step, the colour change law in the second step, and the spider rule again in the third step.

 The same process can be applied to the other Pauli operators given above.
\end{ex}

\begin{dfn}\label{dfn:GS-LC}
 A diagram in the stabilizer \ZX-calculus is called a \emph{GS-LC diagram} if it consists of a graph state (cf. proposition \ref{prop:graph_state}) with arbitrary single qubit Clifford operators applied to each output. Following \cite{anders_fast_2005}, we call the Clifford operator associated with one of the qubits in the graph state its \emph{vertex operator}.
\end{dfn}

\begin{thm}[\cite{duncan_graph_2009}]\label{thm:Van_den_Nest}
 Let $G=(V,E)$ be a graph with adjacency matrix $\theta$ and let $G\star v$ be the graph that results from applying a local complementation about some $v\in V$. Then the equality
 \[
  \ket{G\star v} = R_{X,v}\otimes\bigotimes_{u\in V}R_{Z,u}^{-\theta_{uv}}\ket{G}
 \]
 holds in the \ZX-calculus, i.e.\ we have
 
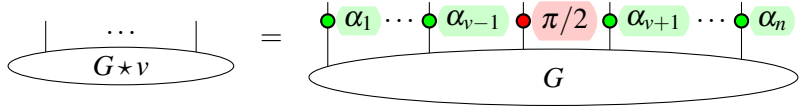
\begin{center}
  \begin{tikzpicture}
	\begin{pgfonlayer}{nodelayer}
		\node [style=none] (0) at (0.75, 0.25) {};
		\node [style=gn,label={[gphase]right:$\alpha_1$}] (1) at (0.75, -0) {};
		\node [style=none] (2) at (1.7, -0) {$\ldots$};
		\node [style=gn,label={[gphase]right:$\alpha_{v-1}$}] (3) at (2.1, -0) {};
		\node [style=none] (4) at (2.1, 0.25) {};
		\node [style=none] (5) at (0.75, -0.75) {};
		\node [ellipse, fill=White, draw=Black, minimum height=0.75 cm, minimum width=6.5 cm, style=none] (6) at (3.75, -0.75) {$G$};
		\node [style=none] (7) at (2.1, -0.75) {};
		\node [style=none] (8) at (0, -0.25) {$=$};
		\node [style=gn,label={[gphase]right:$\alpha_n$}] (9) at (6.25, -0) {};
		\node [style=gn,label={[gphase]right:$\alpha_{v+1}$}] (10) at (4.5, -0) {};
		\node [style=none] (11) at (4.5, 0.25) {};
		\node [style=none] (12) at (4.5, -0.75) {};
		\node [style=none] (13) at (6.25, -0.75) {};
		\node [style=none] (14) at (6.25, 0.25) {};
		\node [style=none] (15) at (5.85, -0) {$\ldots$};
		\node [style=none] (16) at (3.35, 0.25) {};
		\node [style=none] (17) at (3.35, -0.75) {};
		\node [style=rn,label={[rphase]right:$\pi/2$}] (18) at (3.35, -0) {};
		\node [ellipse, fill=White, draw=Black, minimum height=0.5 cm, minimum width=3 cm, style=none] (19) at (-2, -0.6) {$G\star v$};
		\node [style=none] (20) at (-3, 0) {};
		\node [style=none] (21) at (-3, -0.75) {};
		\node [style=none] (22) at (-1, 0) {};
		\node [style=none] (23) at (-1, -0.75) {};
		\node [style=none] (15) at (-2, -0.15) {$\ldots$};
	\end{pgfonlayer}
	\begin{pgfonlayer}{edgelayer}
		\draw (0.center) to (5.center);
		\draw (4.center) to (7.center);
		\draw (11.center) to (12.center);
		\draw (14.center) to (13.center);
		\draw (16.center) to (17.center);
		\draw (20.center) to (21.center);
		\draw (22.center) to (23.center);
	\end{pgfonlayer}
  \end{tikzpicture}
 \end{center}
 where $\alpha_k = -\theta_{kv}\pi/2$ for $k\in V\setminus\{v\}$.
\end{thm}

\subsection{Equivalence transformations of GS-LC diagrams}
\label{s:equivalence_GS-LC}

Consider the $n$-qubit GS-LC diagram
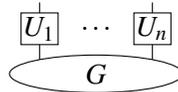
\begin{center}
 \begin{tikzpicture}[baseline=-0.1 cm]
	\begin{pgfonlayer}{nodelayer}
		\node [style=none] (0) at (-0.75, 0.6) {};
		\node [style=none] (1) at (0.75, 0.6) {};
		\node [style=none] (2) at (-0.75, -0.25) {};
		\node [ellipse, fill=White, draw=Black, minimum height=0.5 cm, minimum width=2.3 cm, style=none] (3) at (0, -0.35) {$G$};
		\node [style=none] (4) at (0.75, -0.25) {};
		\node [style=none] (5) at (0, 0.25) {$\ldots$};
		\node [rectangle,fill=white,draw=black,minimum height=12pt,minimum width=14pt,inner sep=0pt] (6) at (-0.75, 0.25) {$U_1$};
		\node [rectangle,fill=white,draw=black,minimum height=12pt,minimum width=14pt,inner sep=0pt] (7) at (0.75, 0.25) {$U_n$};
	\end{pgfonlayer}
	\begin{pgfonlayer}{edgelayer}
		\draw (0.center) to (2.center);
		\draw (1.center) to (4.center);
	\end{pgfonlayer}
 \end{tikzpicture}
\end{center}
where $G=(V,E)$ is a graph with adjacency matrix $\theta$, and $U_v\in C_1$ for $v\in V$. It is useful to set out explicitly three equivalence transformations of GS-LC diagrams, i.e.\ operations that take a GS-LC diagram to an equal but generally not identical GS-LC diagram:

\textbf{Local complementation about a qubit $v$}: Let $G\star v$ denote the graph that results from $G$ through application of the graph-theoretical local complementation about some vertex $v\in V$. Then by theorem \ref{thm:Van_den_Nest},
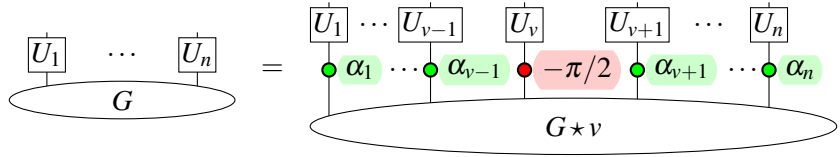
\begin{center}
 \begin{tikzpicture}
	\begin{pgfonlayer}{nodelayer}
		\node [style=none] (0) at (0.75, 0.9) {};
		\node [style=gn,label={[gphase]right:$\alpha_1$}] (1) at (0.75, -0) {};
		\node [style=none] (2) at (1.75, -0) {$\ldots$};
		\node [style=gn,label={[gphase]right:$\alpha_{v-1}$}] (3) at (2.1, -0) {};
		\node [style=none] (4) at (2.1, 0.9) {};
		\node [style=none] (5) at (0.75, -0.75) {};
		\node [ellipse, fill=White, draw=Black, minimum height=0.75 cm, minimum width=7 cm, style=none] (6) at (4, -0.75) {$G\star v$};
		\node [style=none] (7) at (2.1, -0.75) {};
		\node [style=none] (8) at (0, -0) {$=$};
		\node [style=gn,label={[gphase]right:$\alpha_n$}] (9) at (6.6, -0) {};
		\node [style=gn,label={[gphase]right:$\alpha_{v+1}$}] (10) at (4.85, -0) {};
		\node [style=none] (11) at (4.85, 0.9) {};
		\node [style=none] (12) at (4.85, -0.75) {};
		\node [style=none] (13) at (6.6, -0.75) {};
		\node [style=none] (14) at (6.6, 0.9) {};
		\node [style=none] (15) at (6.25, -0) {$\ldots$};
		\node [style=none] (16) at (3.35, 0.9) {};
		\node [style=none] (17) at (3.35, -0.75) {};
		\node [style=rn,label={[rphase]right:$-\pi/2$}] (18) at (3.35, -0) {};
		\node [ellipse, fill=White, draw=Black, minimum height=0.5 cm, minimum width=3 cm, style=none] (19) at (-2, -0.4) {$G$};
		\node [style=none] (20) at (-3, 0.5) {};
		\node [style=none] (21) at (-3, -0.5) {};
		\node [style=none] (22) at (-1, 0.5) {};
		\node [style=none] (23) at (-1, -0.5) {};
		\node [style=none] (24) at (-2, 0.2) {$\ldots$};
		\node [rectangle,fill=white,draw=black,minimum height=13pt,minimum width=14pt,inner sep=0pt] (25) at (-3, 0.2) {$U_1$};
		\node [rectangle,fill=white,draw=black,minimum height=13pt,minimum width=14pt,inner sep=0pt] (26) at (-1, 0.2) {$U_n$};
		\node [rectangle,fill=white,draw=black,minimum height=13pt,minimum width=14pt,inner sep=0pt] (27) at (0.75, 0.6) {$U_1$};
		\node [rectangle,fill=white,draw=black,minimum height=13pt,minimum width=24pt,inner sep=0pt] (28) at (2.1, 0.6) {$U_{v-1}$};
		\node [rectangle,fill=white,draw=black,minimum height=13pt,minimum width=14pt,inner sep=0pt] (29) at (3.35, 0.6) {$U_v$};
		\node [rectangle,fill=white,draw=black,minimum height=13pt,minimum width=24pt,inner sep=0pt] (30) at (4.85, 0.6) {$U_{v+1}$};
		\node [rectangle,fill=white,draw=black,minimum height=13pt,minimum width=14pt,inner sep=0pt] (31) at (6.6, 0.6) {$U_n$};
		\node [style=none] (32) at (1.35, 0.6) {$\ldots$};
		\node [style=none] (33) at (5.8, 0.6) {$\ldots$};
	\end{pgfonlayer}
	\begin{pgfonlayer}{edgelayer}
		\draw (0.center) to (5.center);
		\draw (4.center) to (7.center);
		\draw (11.center) to (12.center);
		\draw (14.center) to (13.center);
		\draw (16.center) to (17.center);
		\draw (20.center) to (21.center);
		\draw (22.center) to (23.center);
	\end{pgfonlayer}
 \end{tikzpicture}
\end{center}
where $\alpha_u = \theta_{uv}\pi/2$ for $u\in V\setminus\{v\}$. In the following, when we say ``local complementation'', we usually mean this transformation, which consists of a graph operation and a change to the local Clifford operators.

\textbf{Fixpoint operation on a qubit $v$}: Let $v\in V$, then
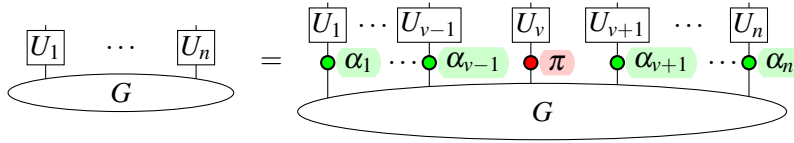
\begin{center}
 \begin{tikzpicture}
	\begin{pgfonlayer}{nodelayer}
		\node [style=none] (0) at (0.75, 0.8) {};
		\node [style=gn,label={[gphase]right:$\alpha_1$}] (1) at (0.75, -0) {};
		\node [style=none] (2) at (1.75, -0) {$\ldots$};
		\node [style=gn,label={[gphase]right:$\alpha_{v-1}$}] (3) at (2.1, -0) {};
		\node [style=none] (4) at (2.1, 0.8) {};
		\node [style=none] (5) at (0.75, -0.75) {};
		\node [ellipse, fill=White, draw=Black, minimum height=0.75 cm, minimum width=6.5 cm, style=none] (6) at (3.6, -0.65) {$G$};
		\node [style=none] (7) at (2.1, -0.75) {};
		\node [style=none] (8) at (0, -0) {$=$};
		\node [style=gn,label={[gphase]right:$\alpha_n$}] (9) at (6.35, -0) {};
		\node [style=gn,label={[gphase]right:$\alpha_{v+1}$}] (10) at (4.6, -0) {};
		\node [style=none] (11) at (4.6, 0.8) {};
		\node [style=none] (12) at (4.6, -0.75) {};
		\node [style=none] (13) at (6.35, -0.75) {};
		\node [style=none] (14) at (6.35, 0.8) {};
		\node [style=none] (15) at (6, -0) {$\ldots$};
		\node [style=none] (16) at (3.45, 0.8) {};
		\node [style=none] (17) at (3.45, -0.75) {};
		\node [style=rn,label={[rphase]right:$\pi$}] (18) at (3.45, -0) {};
		\node [ellipse, fill=White, draw=Black, minimum height=0.5 cm, minimum width=3 cm, style=none] (19) at (-2, -0.4) {$G$};
		\node [style=none] (20) at (-3, 0.5) {};
		\node [style=none] (21) at (-3, -0.5) {};
		\node [style=none] (22) at (-1, 0.5) {};
		\node [style=none] (23) at (-1, -0.5) {};
		\node [style=none] (24) at (-2, 0.2) {$\ldots$};
		\node [rectangle,fill=white,draw=black,minimum height=13pt,minimum width=14pt,inner sep=0pt] (25) at (-3, 0.2) {$U_1$};
		\node [rectangle,fill=white,draw=black,minimum height=13pt,minimum width=14pt,inner sep=0pt] (26) at (-1, 0.2) {$U_n$};
		\node [rectangle,fill=white,draw=black,minimum height=13pt,minimum width=14pt,inner sep=0pt] (27) at (0.75, 0.5) {$U_1$};
		\node [rectangle,fill=white,draw=black,minimum height=13pt,minimum width=24pt,inner sep=0pt] (28) at (2.1, 0.5) {$U_{v-1}$};
		\node [rectangle,fill=white,draw=black,minimum height=13pt,minimum width=14pt,inner sep=0pt] (29) at (3.45, 0.5) {$U_v$};
		\node [rectangle,fill=white,draw=black,minimum height=13pt,minimum width=24pt,inner sep=0pt] (30) at (4.6, 0.5) {$U_{v+1}$};
		\node [rectangle,fill=white,draw=black,minimum height=13pt,minimum width=14pt,inner sep=0pt] (31) at (6.35, 0.5) {$U_n$};
		\node [style=none] (32) at (1.35, 0.5) {$\ldots$};
		\node [style=none] (33) at (5.55, 0.5) {$\ldots$};
	\end{pgfonlayer}
	\begin{pgfonlayer}{edgelayer}
		\draw (0.center) to (5.center);
		\draw (4.center) to (7.center);
		\draw (11.center) to (12.center);
		\draw (14.center) to (13.center);
		\draw (16.center) to (17.center);
		\draw (20.center) to (21.center);
		\draw (22.center) to (23.center);
	\end{pgfonlayer}
 \end{tikzpicture}
\end{center}
where $\alpha_u = \theta_{uv}\pi$ for $u\in V\setminus\{v\}$. This equality holds by the definition of graph states, or, alternatively, by a double local complementation about $v$.

\textbf{Local complementation along an edge $\{v,w\}$}: Let $v,w\in V$ such that $\{v,w\}\in E$. Then
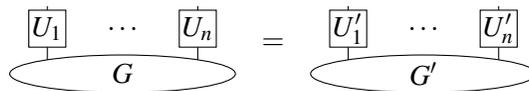
\begin{center}
 \begin{tikzpicture}
	\begin{pgfonlayer}{nodelayer}
		\node [style=none] (0) at (1, 0.5) {};
		\node [style=none] (1) at (1, -0.5) {};
		\node [ellipse, fill=White, draw=Black, minimum height=0.5 cm, minimum width=3 cm, style=none] (2) at (2, -0.4) {$G'$};
		\node [style=none] (3) at (0, -0) {$=$};
		\node [style=none] (4) at (3, -0.5) {};
		\node [style=none] (5) at (3, 0.5) {};
		\node [ellipse, fill=White, draw=Black, minimum height=0.5 cm, minimum width=3 cm, style=none] (6) at (-2, -0.4) {$G$};
		\node [style=none] (7) at (-3, 0.5) {};
		\node [style=none] (8) at (-3, -0.5) {};
		\node [style=none] (9) at (-1, 0.5) {};
		\node [style=none] (10) at (-1, -0.5) {};
		\node [style=none] (11) at (-2, 0.2) {$\ldots$};
		\node [rectangle,fill=white,draw=black,minimum height=14pt,minimum width=14pt,inner sep=0pt] (12) at (-3, 0.2) {$U_1$};
		\node [rectangle,fill=white,draw=black,minimum height=14pt,minimum width=14pt,inner sep=0pt] (13) at (-1, 0.2) {$U_n$};
		\node [rectangle,fill=white,draw=black,minimum height=14pt,minimum width=14pt,inner sep=0pt] (14) at (1, 0.2) {$U_1'$};
		\node [rectangle,fill=white,draw=black,minimum height=14pt,minimum width=14pt,inner sep=0pt] (15) at (3, 0.2) {$U_n'$};
		\node [style=none] (16) at (2, 0.2) {$\ldots$};
	\end{pgfonlayer}
	\begin{pgfonlayer}{edgelayer}
		\draw (0.center) to (1.center);
		\draw (5.center) to (4.center);
		\draw (7.center) to (8.center);
		\draw (9.center) to (10.center);
	\end{pgfonlayer}
 \end{tikzpicture}
\end{center}
where
\[
 U_j' = \begin{cases} U_j\circ R_Z\circ R_X^{-1}\circ R_Z & \text{if }j\in\{v,w\} \\ U_j\circ Z & \text{if } j\in V\setminus\{v,w\}\wedge(\{j,v\}\in E \vee \{j,w\}\in E) \\ U_j & \text{otherwise} \end{cases}
\]
and $G' = (V,E')$ satisfies the following properties
\begin{mitem}
 \item $G' = ((G\star v)\star w)\star v = ((G\star w)\star v)\star w$;
 \item $\{v,w\}\in E'$;
 \item for $j\in V\setminus\{v,w\}$, $\{j,v\}\in E'\Leftrightarrow \{j,w\}\in E$ and $\{j,w\}\in E'\Leftrightarrow \{j,v\}\in E$, i.e.\ a vertex $j$ is adjacent to $v$ in $G'$ if and only if $j$ was adjacent to $w$ in $G$ and correspondingly with $v$ and $w$ exchanged;
 \item for $p,q\in V\setminus\{v,w\}$, let $P$ be the intersection of $p$'s neighbourhood with $\{v,w\}$, i.e.\ $v\in P$ if $\{p,v\}\in E$ and $w\in P$ if $\{p,w\}\in E$, and define $Q$ correspondingly. Then the edge $\{p,q\}$ is toggled if and only if $P,Q$ and $\emptyset$ are pairwise distinct.
\end{mitem}
This is an equivalence transformation because it consists of three subsequent local complementations on qubits, but it is worth a separate mention because of non-obvious properties like the symmetry under interchange of $v$ and $w$.

\subsection{Any stabilizer state diagram is equal to some GS-LC diagram}

From standard stabilizer quantum mechanics, we know that any stabilizer state is local Clifford-e\-quiv\-a\-lent to some graph state (cf. theorem \ref{thm:stabilizer_graph_state}). In the following, we will show that a corresponding statement holds in the \ZX-calculus: any stabilizer state diagram is equal to some GS-LC diagram. The proof of this result is strongly inspired by Anders and Briegel's proof that stabilizer quantum mechanics can be simulated efficiently on classical computers using a representation based on graph states and local Clifford operators \cite{anders_fast_2005}.

First, note that \ZX-calculus diagrams can be decomposed into five types of basic elements.

\begin{lem}\label{lem:basic_elements}
 Any \ZX-calculus diagram can be written as a combination of four basic spiders
 \begin{equation}\label{eq:basic_spiders}
  \begin{tikzpicture}[baseline=0 cm]
	\begin{pgfonlayer}{nodelayer}
		\node [style=gn] (0) at (-3, -0) {};
		\node [style=gn] (1) at (1, -0) {};
		\node [style=gn] (2) at (-1, -0) {};
		\node [style=gn] (3) at (3, -0) {};
		\node [style=none] (4) at (-3.25, 0.25) {};
		\node [style=none] (5) at (-2.75, 0.25) {};
		\node [style=none] (6) at (-3, -0.25) {};
		\node [style=none] (7) at (1, 0.25) {};
		\node [style=none] (8) at (0.75, -0.25) {};
		\node [style=none] (9) at (1.25, -0.25) {};
		\node [style=none] (10) at (-1, -0.25) {};
		\node [style=none] (11) at (3, 0.25) {};
	\end{pgfonlayer}
	\begin{pgfonlayer}{edgelayer}
		\draw [bend right=15] (4.center) to (0);
		\draw [bend left=15] (5.center) to (0);
		\draw (0) to (6.center);
		\draw (7.center) to (1);
		\draw [bend right=15] (1) to (8.center);
		\draw [bend left=15] (1) to (9.center);
		\draw (2) to (10.center);
		\draw (11.center) to (3);
	\end{pgfonlayer}
  \end{tikzpicture}
 \end{equation}
 and the 24 single qubit Clifford unitaries.
\end{lem}
\begin{proof}
 Using the spider law, any green spider with phase 0 can be ``pulled apart'' into a diagram composed of the four elements given above. By the identity law, cups
 \begin{tikzpicture}
	\begin{pgfonlayer}{nodelayer}
		\node [style=none] (0) at (-0.25, 0.25) {};
		\node [style=none] (1) at (0, 0) {};
		\node [style=none] (2) at (0.25, 0.25) {};
	\end{pgfonlayer}
	\begin{pgfonlayer}{edgelayer}
		\draw [bend right=45] (0.center) to (1.center);
		\draw [bend left=45] (2.center) to (1.center);
	\end{pgfonlayer}
 \end{tikzpicture}
 and caps 
 \begin{tikzpicture}
	\begin{pgfonlayer}{nodelayer}
		\node [style=none] (0) at (0.25, -0) {};
		\node [style=none] (1) at (0, 0.25) {};
		\node [style=none] (2) at (-0.25, -0) {};
	\end{pgfonlayer}
	\begin{pgfonlayer}{edgelayer}
		\draw [bend right=45] (0.center) to (1.center);
		\draw [bend left=45] (2.center) to (1.center);
	\end{pgfonlayer}
 \end{tikzpicture}
 can be replaced with green spiders.

 Any red spider of phase 0 can be turned into a green spider using the colour change law, introducing a Hadamard operator on each leg. Thus any red spider can be written as a combination of Hadamard operators and the basic green spiders.

 If a red or green spider has a non-zero phase, it can be decomposed into a phase 0 spider and a single qubit phase operator. Therefore, any diagram in the \ZX-calculus for stabilizer quantum mechanics can be written as a combination of the four spiders given in \eqref{eq:basic_spiders} and the 24 single qubit Clifford unitaries.
\end{proof}

\begin{thm}\label{thm:ZX_GS-LC}
 Any stabilizer state diagram is equal to some GS-LC diagram within the \ZX-calculus.
\end{thm}
\begin{proof}
 For clarity, the proof has been split into various lemmas, which are stated and proved after the theorem.

 By lemma \ref{lem:basic_elements}, any \ZX-calculus diagram can be written in terms of five basic elements. Recall that a \ZX-calculus diagram represents a quantum state if and only if it has no inputs. Of the basic elements given in lemma \ref{lem:basic_elements}, \state{gn} is the only one with no inputs. Thus any diagram representing a state must contain at least one such component (or a cup, which can be replaced by spiders). Clearly \state{gn} is a GS-LC diagram. We can now proceed by induction: If, for each of the basic components, applying it to a GS-LC diagram yields a diagram that is equal to some GS-LC diagram, then any stabilizer state diagram is equal to some GS-LC diagram. Lemmas \ref{lem:basic_state}, \ref{lem:LC}, \ref{lem:measurement}, \ref{lem:split} and \ref{lem:join} show these inductive steps. Therefore any stabilizer state diagram can be decomposed into the basic elements and then converted, step by step, into a GS-LC diagram.
\end{proof}

\begin{lem}\label{lem:basic_state}
 A stabilizer state diagram which consists of a GS-LC diagram and \state{gn} is equal to some GS-LC diagram within the ZX-calculus.
\end{lem}
\begin{proof}
 Adding a vertex to a graph yields another graph, so adding \state{gn} to a graph state diagram yields another graph state diagram. The same holds for GS-LC diagrams.
\end{proof}

\begin{lem}\label{lem:LC}
 A stabilizer state diagram which consists of a single qubit Clifford unitary applied to some GS-LC diagram is equal to a GS-LC diagram within the ZX-calculus.
\end{lem}
\begin{proof}
 This follows directly from definition \ref{dfn:GS-LC}, the definition of GS-LC diagrams.
\end{proof}

\begin{lem}\label{lem:measurement}
 A stabilizer state diagram which consists of \effect{gn} applied to some GS-LC diagram is equal to a GS-LC diagram or to the zero diagram within the ZX-calculus.
\end{lem}
\begin{proof}
 Call the vertex of the GS-LC diagram to which the post-selected measurement \effect{gn} is applied the \emph{operand vertex}. There are two cases.

 \textit{The operand vertex has no neighbours}: There are six single qubit pure stabilizer states. If the operand vertex is in state \state{gn,label={[gphase]right:$\pi$}}, the result of the measurement is zero. Otherwise the measurement operator \effect{gn} combines with the single qubit state to a non-zero global factor, which can be ignored.

 \textit{The operand vertex has at least one neighbour}: It is well known that Z-basis measurements on graph states are easy.

 In the ZX-calculus, the Z-basis states are denoted (somewhat counter-intuitively) by red effects: \effect{rn} denotes $\ket{0}$ and \effect{rn,label={[rphase]right:$\pi$}} represents $\bra{1}$. By the copy rule,
  
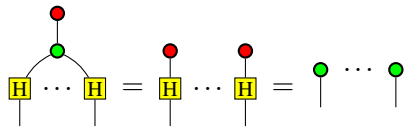
\begin{center}
   \begin{tikzpicture}
	\begin{pgfonlayer}{nodelayer}
		\node [style=rn] (0) at (-2, 1) {};
		\node [style=gn] (1) at (-2, 0.5) {};
		\node [style=Hadamard] (2) at (-2.5, -0) {};
		\node [style=Hadamard] (3) at (-1.5, -0) {};
		\node [style=none] (4) at (-2, -0) {$\ldots$};
		\node [style=none] (5) at (-2.5, -0.5) {};
		\node [style=none] (6) at (-1.5, -0.5) {};
		\node [style=none] (7) at (-1, -0) {$=$};
		\node [style=rn] (8) at (-0.5, 0.5) {};
		\node [style=rn] (9) at (0.5, 0.5) {};
		\node [style=Hadamard] (10) at (-0.5, -0) {};
		\node [style=Hadamard] (11) at (0.5, -0) {};
		\node [style=none] (12) at (0, -0) {$\ldots$};
		\node [style=none] (13) at (-0.5, -0.5) {};
		\node [style=none] (14) at (0.5, -0.5) {};
		\node [style=none] (15) at (1, -0) {$=$};
		\node [style=gn] (16) at (1.5, 0.25) {};
		\node [style=gn] (17) at (2.5, 0.25) {};
		\node [style=none] (18) at (2, 0.25) {$\ldots$};
		\node [style=none] (19) at (1.5, -0.25) {};
		\node [style=none] (20) at (2.5, -0.25) {};
	\end{pgfonlayer}
	\begin{pgfonlayer}{edgelayer}
		\draw (0) to (1);
		\draw [bend right=15] (1) to (2);
		\draw (2) to (5.center);
		\draw [bend left=15] (1) to (3);
		\draw (3) to (6.center);
		\draw (8) to (10);
		\draw (10) to (13.center);
		\draw (9) to (11);
		\draw (11) to (14.center);
		\draw (16) to (19.center);
		\draw (17) to (20.center);
	\end{pgfonlayer}
   \end{tikzpicture}
  \end{center}
  and by the $\pi$-copy rule, the same holds for \effect{rn,label={[rphase]right:$\pi$}}. Thus if the vertex operator of the operand vertex is
  
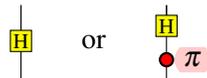
\begin{center}
   \begin{tikzpicture}[baseline=-0.1cm]
	\begin{pgfonlayer}{nodelayer}
		\node [style=Hadamard] (0) at (0, 0) {};
		\node [style=none] (1) at (0, 0.5) {};
		\node [style=none] (2) at (0, -0.5) {};
	\end{pgfonlayer}
	\begin{pgfonlayer}{edgelayer}
		\draw (0.center) to (1);
		\draw (2.center) to (1);
	\end{pgfonlayer}
   \end{tikzpicture}
   $\quad$ or $\quad$
   \begin{tikzpicture}[baseline=-0.1cm]
	\begin{pgfonlayer}{nodelayer}
		\node [style=Hadamard] (0) at (0, 0.2) {};
		\node [style=none] (1) at (0, 0.5) {};
		\node [style=none] (2) at (0, -0.5) {};
		\node [style=rn,label={[rphase]right:$\pi$}] (3) at (0, -0.25) {};
	\end{pgfonlayer}
	\begin{pgfonlayer}{edgelayer}
		\draw (0.center) to (1);
		\draw (2.center) to (1);
	\end{pgfonlayer}
   \end{tikzpicture}
  \end{center}
  the measured vertex is simply removed from the graph state.

  Otherwise, we can pick one neighbour of the operand vertex; following \cite{anders_fast_2005}, this neighbour will be called the \emph{swapping partner}. A local complementation about the operand vertex adds \phase{rn,label={[rphase]right:$\pi/2$}} to its vertex operator. A local complementation about the swapping partner adds \phase{gn,label={[gphase]right:$-\pi/2$}} to the vertex operator on the operand vertex. Now these two single qubit operators together generate all of $C_1$. Note that local complementations about the operand vertex or its swapping partner do not remove the edge between these two vertices. Therefore, after each local complementation, the operand vertex still has a neighbour, enabling further local complementations. Hence it is always possible to change the vertex operator on the operand vertex to \Hadamard using local complementations. Then, the measurement is easy.
\end{proof}

\begin{lem}\label{lem:split}
 A stabilizer state diagram which consists of \splitnode applied to some GS-LC diagram is equal to a GS-LC diagram within the ZX-calculus.
\end{lem}
\begin{proof}
 As before, call the vertex we are acting upon the operand vertex. Again, there are two cases.

 \textit{The operand vertex has no neighbours}: In this case, the part of the diagram representing the non-operand qubits does not change, hence if it is in GS-LC form originally, it will remain that way. The overall diagram will be in GS-LC form if and only if \splitnode applied to the operand vertex can be transformed into a GS-LC diagram. Now, the six single qubit stabilizer states can be written as
  \begin{equation}\label{eq:single_qubit_states}
   \begin{tikzpicture}[baseline=0.1cm]
	\begin{pgfonlayer}{nodelayer}
		\node [style=gn] (0) at (-4.75, -0) {};
		\node [style=gn,label={[gphase]right:$\pi/2$}] (1) at (-3.25, -0) {};
		\node [style=gn,label={[gphase]right:$\pi$}] (2) at (-1.25, -0) {};
		\node [style=gn,label={[gphase]right:$-\pi/2$}] (3) at (0.5, -0) {};
		\node [style=rn] (4) at (2.75, -0) {};
		\node [style=rn,label={[rphase]right:$\pi$}] (5) at (4.25, -0) {};
		\node [style=none] (6) at (-4.75, 0.5) {};
		\node [style=none] (7) at (-3.25, 0.5) {};
		\node [style=none] (8) at (-1.25, 0.5) {};
		\node [style=none] (9) at (0.5, 0.5) {};
		\node [style=none] (10) at (2.75, 0.5) {};
		\node [style=none] (11) at (4.25, 0.5) {};
	\end{pgfonlayer}
	\begin{pgfonlayer}{edgelayer}
		\draw (6.center) to (0);
		\draw (7.center) to (1);
		\draw (8.center) to (2);
		\draw (9.center) to (3);
		\draw (10.center) to (4);
		\draw (11.center) to (5);
	\end{pgfonlayer}
   \end{tikzpicture}
  \end{equation}
  By the spider law, the identity law, and the self-inverse property of the Hadamard operator,
  \begin{center}
   \begin{tikzpicture}
	\begin{pgfonlayer}{nodelayer}
		\node [style=gn] (0) at (-2, -0) {};
		\node [style=gn,label={[gphase]right:$\alpha$}] (1) at (-2, -0.5) {};
		\node [style=gn] (2) at (0, -0.5) {};
		\node [style=gn] (3) at (1, -0.5) {};
		\node [style=gn,label={[gphase]right:$\alpha$}] (4) at (1, -0) {};
		\node [style=none] (5) at (-2.5, 0.5) {};
		\node [style=none] (6) at (-1.5, 0.5) {};
		\node [style=none] (7) at (0, 0.5) {};
		\node [style=none] (8) at (1, 0.5) {};
		\node [style=Hadamard] (9) at (0.5, -0.5) {};
		\node [style=Hadamard] (10) at (0, -0) {};
		\node [style=none] (11) at (-0.75, -0) {$=$};
	\end{pgfonlayer}
	\begin{pgfonlayer}{edgelayer}
		\draw [bend right=15] (5.center) to (0);
		\draw [bend left=15] (6.center) to (0);
		\draw (0) to (1);
		\draw (7.center) to (10);
		\draw (10) to (2);
		\draw (2) to (9);
		\draw (9) to (3);
		\draw (8.center) to (4);
		\draw (4) to (3);
	\end{pgfonlayer}
   \end{tikzpicture}
  \end{center}
  for $\alpha\in\{0,\pi/2,\pi,-\pi/2\}$. Using the copy law and the $\pi$-copy law, for $\beta\in\{0,\pi\}$,
  \begin{center}
   \begin{tikzpicture}
	\begin{pgfonlayer}{nodelayer}
		\node [style=gn] (0) at (-3, -0) {};
		\node [style=rn,label={[rphase]right:$\beta$}] (1) at (-3, -0.5) {};
		\node [style=none] (2) at (-3.5, 0.5) {};
		\node [style=none] (3) at (-2.5, 0.5) {};
		\node [style=none] (4) at (-1.75, -0) {$=$};
		\node [style=rn,label={[rphase]right:$\beta$}] (5) at (-1, -0.25) {};
		\node [style=rn,label={[rphase]right:$\beta$}] (6) at (0, -0.25) {};
		\node [style=none] (7) at (-1, 0.25) {};
		\node [style=none] (8) at (0, 0.25) {};
		\node [style=none] (9) at (1.25, -0) {$=$};
		\node [style=gn] (10) at (2, -0.75) {};
		\node [style=Hadamard] (11) at (2, -0.25) {};
		\node [style=rn,label={[rphase]right:$\beta$}] (12) at (2, 0.25) {};
		\node [style=gn] (13) at (3, -0.75) {};
		\node [style=Hadamard] (14) at (3, -0.25) {};
		\node [style=rn,label={[rphase]right:$\beta$}] (15) at (3, 0.25) {};
		\node [style=none] (16) at (2, 0.6) {};
		\node [style=none] (17) at (3, 0.6) {};
	\end{pgfonlayer}
	\begin{pgfonlayer}{edgelayer}
		\draw [bend right=15] (2.center) to (0);
		\draw [bend left=15] (3.center) to (0);
		\draw (0) to (1);
		\draw (7.center) to (5);
		\draw (8.center) to (6);
		\draw (16.center) to (10);
		\draw (17.center) to (13);
	\end{pgfonlayer}
   \end{tikzpicture}
  \end{center}
  In each of these cases, the right hand side of the equation can easily be seen to be a GS-LC diagram.

 \textit{The operand vertex has at least one neighbour}: Note that
  \begin{center}
   \begin{tikzpicture}
	\begin{pgfonlayer}{nodelayer}
		\node [style=gn] (0) at (-1.25, 0.25) {};
		\node [style=none] (1) at (-1.25, -0.25) {};
		\node [style=none] (2) at (-1.75, 0.75) {};
		\node [style=none] (3) at (-0.75, 0.75) {};
		\node [style=none] (4) at (-0.25, 0.25) {$=$};
		\node [style=gn] (5) at (0.5, -0) {};
		\node [style=gn] (6) at (1.5, -0) {};
		\node [style=Hadamard] (7) at (1, -0) {};
		\node [style=Hadamard] (8) at (0.5, 0.5) {};
		\node [style=none] (9) at (0.5, 1) {};
		\node [style=none] (10) at (1.5, 1) {};
		\node [style=none] (11) at (1.5, -0.5) {};
	\end{pgfonlayer}
	\begin{pgfonlayer}{edgelayer}
		\draw [bend right=15] (2.center) to (0);
		\draw (0) to (1.center);
		\draw [bend right=15] (0) to (3.center);
		\draw (9.center) to (8);
		\draw (8) to (5);
		\draw (5) to (7);
		\draw (7) to (6);
		\draw (10.center) to (6);
		\draw (6) to (11.center);
	\end{pgfonlayer}
   \end{tikzpicture}
  \end{center}
  so if the vertex operator on the operand vertex is trivial, we just add a new vertex and edge to the graph. Now, as described in the proof for lemma \ref{lem:measurement}, we can use local complementations about the operand vertex and a swapping partner to change the vertex operator on the operand vertex to the identity. Thus whenever we apply \splitnode to a GS-LC diagram, the result is equal to some GS-LC diagram.
\end{proof}

\begin{lem}\label{lem:join}
 A stabilizer state diagram which consists of \joinnode applied to some GS-LC diagram is equal to a GS-LC diagram or the zero diagram within the ZX-calculus.
\end{lem}
\begin{proof}
 As usual, call the qubits to which \joinnode is applied the operand qubits. This time there are two of them, and there are four cases to consider.

 \textit{Operand vertices are connected only to each other}: Since neither operand vertex is connected to any other vertex, we can neglect all non-operand vertices. Now, for any $U,V\in C_1$,
 \begin{center}
  \begin{tikzpicture}
	\begin{pgfonlayer}{nodelayer}
		\node [style=gn] (0) at (-4, -0.5) {};
		\node [style=gn] (1) at (-3, -0.5) {};
		\node [style=gn] (2) at (-3.5, 0.5) {};
		\node [style=Hadamard] (3) at (-3.5, -0.5) {};
		\node [style=normalrect] (4) at (-4, -0) {$U$};
		\node [style=normalrect] (5) at (-3, -0) {$V$};
		\node [style=none] (6) at (-3.5, 1) {};
		\node [style=none] (7) at (-2.25, -0) {$=$};
		\node [style=gn] (8) at (-1.25, 0.25) {};
		\node [style=normalrect] (9) at (-1.25, -0.5) {$W$};
		\node [style=none] (10) at (-1.25, 0.75) {};
	\end{pgfonlayer}
	\begin{pgfonlayer}{edgelayer}
		\draw (6.center) to (2);
		\draw [bend right=15] (2) to (4);
		\draw (4) to (0);
		\draw [bend left=15] (2) to (5);
		\draw (5) to (1);
		\draw (0) to (3);
		\draw (3) to (1);
		\draw (10.center) to (8);
		\draw [bend right=60, looseness=1.50] (8) to (9);
		\draw [bend left=60, looseness=1.50] (8) to (9);
	\end{pgfonlayer}
  \end{tikzpicture}
 \end{center}
 where $W$ is again in $C_1$. Using the $\pi$-commutation rule, the colour change law and the Euler decomposition of the Hadamard operator, it is easy to show that any single qubit Clifford unitary can be written as
 \begin{equation}\label{eq:LC}
  \begin{tikzpicture}[baseline=0cm]
	\begin{pgfonlayer}{nodelayer}
		\node [style=normalrect] (0) at (-0.75, -0) {$W$};
		\node [style=gn,label={[gphase]right:$\alpha$}] (1) at (0.75, 0.6) {};
		\node [style=rn,label={[rphase]right:$\beta$}] (2) at (0.75, -0) {};
		\node [style=gn,label={[gphase]right:$\gamma$}] (3) at (0.75, -0.6) {};
		\node [style=none] (4) at (-0.75, 0.3) {};
		\node [style=none] (5) at (-0.75, -0.3) {};
		\node [style=none] (6) at (0.75, 0.8) {};
		\node [style=none] (7) at (0.75, -0.8) {};
		\node [style=none] (8) at (0, -0) {$=$};
	\end{pgfonlayer}
	\begin{pgfonlayer}{edgelayer}
		\draw (4.center) to (0);
		\draw (0) to (5.center);
		\draw (6.center) to (7.center);
	\end{pgfonlayer}
  \end{tikzpicture}
 \end{equation}
 for some $\alpha,\beta,\gamma\in\{0,\pi/2,\pi,-\pi/2\}$. Thus, using the spider law and the Hopf law,
 
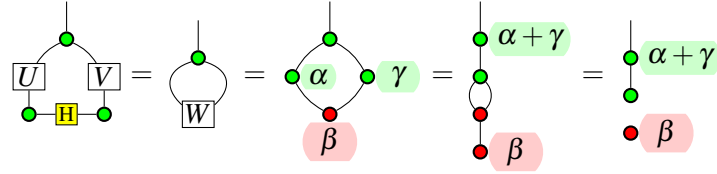
\begin{center}
  \begin{tikzpicture}
	\begin{pgfonlayer}{nodelayer}
		\node [style=gn] (0) at (-4, -0.5) {};
		\node [style=gn] (1) at (-3, -0.5) {};
		\node [style=gn] (2) at (-3.5, 0.5) {};
		\node [style=Hadamard] (3) at (-3.5, -0.5) {};
		\node [style=normalrect] (4) at (-4, -0) {$U$};
		\node [style=normalrect] (5) at (-3, -0) {$V$};
		\node [style=none] (6) at (-3.5, 1) {};
		\node [style=none] (7) at (-2.5, -0) {$=$};
		\node [style=gn] (8) at (-1.75, 0.25) {};
		\node [style=normalrect] (9) at (-1.75, -0.5) {$W$};
		\node [style=none] (10) at (-1.75, 0.75) {};
		\node [style=none] (11) at (-1, -0) {$=$};
		\node [style=gn,label={[gphase]right:$\alpha$}] (12) at (-0.5, -0) {};
		\node [style=gn,label={[gphase]right:$\gamma$}] (13) at (0.5, -0) {};
		\node [style=gn] (14) at (0, 0.5) {};
		\node [style=rn,label={[rphase]below:$\beta$}] (15) at (0, -0.5) {};
		\node [style=none] (16) at (0, 1) {};
		\node [style=none] (17) at (1.5, -0) {$=$};
		\node [style=gn] (18) at (2, -0) {};
		\node [style=rn] (19) at (2, -0.5) {};
		\node [style=gn,label={[gphase]right:$\alpha+\gamma$}] (20) at (2, 0.5) {};
		\node [style=rn,label={[rphase]right:$\beta$}] (21) at (2, -1) {};
		\node [style=none] (22) at (2, 1) {};
		\node [style=none] (23) at (3.5, -0) {$=$};
		\node [style=none] (24) at (4, 0.75) {};
		\node [style=gn] (25) at (4, -0.25) {};
		\node [style=rn,label={[rphase]right:$\beta$}] (26) at (4, -0.75) {};
		\node [style=gn,label={[gphase]right:$\alpha+\gamma$}] (27) at (4, 0.25) {};
	\end{pgfonlayer}
	\begin{pgfonlayer}{edgelayer}
		\draw (6.center) to (2);
		\draw [bend right=15] (2) to (4);
		\draw (4) to (0);
		\draw [bend left=15] (2) to (5);
		\draw (5) to (1);
		\draw (0) to (3);
		\draw (3) to (1);
		\draw (10.center) to (8);
		\draw [bend right=60, looseness=1.50] (8) to (9);
		\draw [bend left=60, looseness=1.50] (8) to (9);
		\draw [bend right=45] (18) to (19);
		\draw [bend left=45] (18) to (19);
		\draw (16.center) to (14);
		\draw [bend right=15] (14) to (12);
		\draw [bend right=15] (12) to (15);
		\draw [bend right=15] (15) to (13);
		\draw [bend right=15] (13) to (14);
		\draw (22.center) to (20);
		\draw (20) to (18);
		\draw (19) to (21);
		\draw (24.center) to (25);
	\end{pgfonlayer}
  \end{tikzpicture}
 \end{center}
 Hence if $\beta=\pi$, the resulting diagram is zero, otherwise it is a GS-LC diagram.

 \textit{One operand vertex has no neighbours}: If one of the operand vertices has no neighbours, it must be in one of the six single qubit states given in \eqref{eq:single_qubit_states}. Now for $\alpha\in\{0,\pi/2,\pi,-\pi/2\}$ and $\beta\in\{0,\pi\}$,
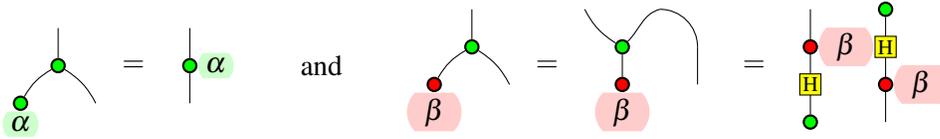
 \begin{center}
  \begin{tikzpicture}
	\begin{pgfonlayer}{nodelayer}
		\node [style=gn,label={[gphase]below:$\alpha$}] (0) at (-6, -0.5) {};
		\node [style=gn] (1) at (-5.5, -0) {};
		\node [style=none] (2) at (-5.5, 0.5) {};
		\node [style=none] (3) at (-5, -0.5) {};
		\node [style=none] (4) at (-4.5, -0) {$=$};
		\node [style=gn,label={[gphase]right:$\alpha$}] (5) at (-3.75, -0) {};
		\node [style=none] (6) at (-3.75, 0.5) {};
		\node [style=none] (7) at (-3.75, -0.5) {};
		\node [style=none] (8) at (0.5, -0.25) {};
		\node [style=rn,label={[rphase]below:$\beta$}] (9) at (-0.5, -0.25) {};
		\node [style=none] (10) at (0, 0.75) {};
		\node [style=none] (11) at (1, -0) {$=$};
		\node [style=gn] (12) at (0, 0.25) {};
		\node [style=gn] (13) at (2, 0.25) {};
		\node [style=rn,label={[rphase]below:$\beta$}] (14) at (2, -0.25) {};
		\node [style=none] (15) at (1.5, 0.75) {};
		\node [style=none] (16) at (2.5, 0.75) {};
		\node [style=none] (17) at (3, 0.25) {};
		\node [style=none] (18) at (3, -0.25) {};
		\node [style=none] (19) at (3.75, -0) {$=$};
		\node [style=gn] (20) at (4.5, -0.75) {};
		\node [style=gn] (21) at (5.5, 0.75) {};
		\node [style=Hadamard] (22) at (5.5, 0.25) {};
		\node [style=rn,label={[rphase]right:$\beta$}] (23) at (5.5, -0.25) {};
		\node [style=none] (24) at (4.5, 0.75) {};
		\node [style=none] (25) at (5.5, -0.75) {};
		\node [style=none] (26) at (-2, -0) {and};
		\node [style=rn,label={[rphase]right:$\beta$}] (27) at (4.5, 0.25) {};
		\node [style=Hadamard] (28) at (4.5, -0.25) {};
	\end{pgfonlayer}
	\begin{pgfonlayer}{edgelayer}
		\draw (2.center) to (1);
		\draw [bend right=15] (1) to (0);
		\draw [bend left=15] (1) to (3.center);
		\draw (6.center) to (5);
		\draw (5) to (7.center);
		\draw (10.center) to (12);
		\draw [bend right=15] (12) to (9);
		\draw [bend left=15] (12) to (8.center);
		\draw [bend right=15] (15.center) to (13);
		\draw [in=-165, out=45] (13) to (16.center);
		\draw [bend left=45] (16.center) to (17.center);
		\draw (17.center) to (18.center);
		\draw (13) to (14);
		\draw (24.center) to (20);
		\draw (21) to (25.center);
	\end{pgfonlayer}
  \end{tikzpicture}
 \end{center}
 Thus by lemma \ref{lem:LC} and lemma \ref{lem:measurement}, no matter what the properties of the other operand vertex are, the resulting state is always equal to a GS-LC diagram or the zero diagram.

 Having resolved the case where the two operand vertices are connected only to each other and the case where one of the operand vertices has no neighbours, we are left with the cases where both operand vertices have neighbours and at least one of the operand vertices has a non-operand neighbour.

 \textit{Both operand vertices have non-operand neighbours}: Denote the two operand vertices by $a$ and $b$. Pick one of $a$'s non-operand neighbours to be a swapping partner. As laid out in the proof of lemma \ref{lem:measurement}, we can use local complementations about $a$ and its swapping partner to change the vertex operator of $a$ to the identity. We can then do the same for $b$, picking a new swapping partner from among its neighbours. If $a$ is connected to $b$ or to $b$'s swapping partner, these operations may result in adding some phase operators of the form \phase{gn,label={[gphase]right:$-\pi/2$}} to $a$'s vertex operator. This is not a problem, as green phase operators commute with \joinnode. Once the vertex operators for both operand vertices are identities or green phase operators, we can move the green phases through the spider and then merge the vertices. Note that
 
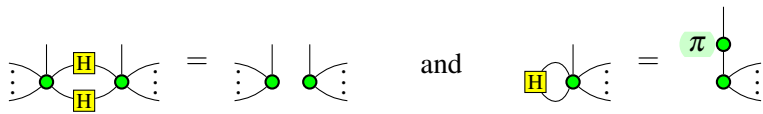
\begin{center}
  \begin{tikzpicture}
	\begin{pgfonlayer}{nodelayer}
		\node [style=gn] (0) at (-4.75, -0) {};
		\node [style=gn] (1) at (-3.75, -0) {};
		\node [style=Hadamard] (2) at (-4.25, 0.25) {};
		\node [style=Hadamard] (3) at (-4.25, -0.25) {};
		\node [style=none] (4) at (-5.25, 0.25) {};
		\node [style=none] (5) at (-5.2, 0.1) {$\vdots$};
		\node [style=none] (6) at (-5.25, -0.25) {};
		\node [style=none] (7) at (-3.25, 0.25) {};
		\node [style=none] (8) at (-3.3, 0.1) {$\vdots$};
		\node [style=none] (9) at (-3.25, -0.25) {};
		\node [style=none] (10) at (-4.75, 0.5) {};
		\node [style=none] (11) at (-3.75, 0.5) {};
		\node [style=none] (12) at (-2.75, 0.25) {$=$};
		\node [style=none] (13) at (-2.25, 0.25) {};
		\node [style=none] (14) at (-2.2, 0.1) {$\vdots$};
		\node [style=none] (15) at (-2.25, -0.25) {};
		\node [style=none] (16) at (-0.75, 0.25) {};
		\node [style=none] (17) at (-0.8, 0.1) {$\vdots$};
		\node [style=none] (18) at (-0.75, -0.25) {};
		\node [style=none] (19) at (-1.75, 0.5) {};
		\node [style=none] (20) at (-1.25, 0.5) {};
		\node [style=gn] (21) at (-1.75, -0) {};
		\node [style=gn] (22) at (-1.25, -0) {};
		\node [style=none] (23) at (2.7, 0.1) {$\vdots$};
		\node [style=none] (24) at (2.75, -0.25) {};
		\node [style=none] (25) at (2.25, 0.5) {};
		\node [style=gn] (26) at (2.25, -0) {};
		\node [style=none] (27) at (2.75, 0.25) {};
		\node [style=none] (28) at (4.7, 0.1) {$\vdots$};
		\node [style=none] (29) at (4.75, -0.25) {};
		\node [style=none] (30) at (4.25, 1) {};
		\node [style=gn] (31) at (4.25, -0) {};
		\node [style=none] (32) at (4.75, 0.25) {};
		\node [style=Hadamard] (33) at (1.75, -0) {};
		\node [style=none] (34) at (3.25, 0.25) {$=$};
		\node [style=gn,label={[gphase]left:$\pi$}] (35) at (4.25, 0.5) {};
		\node [style=none] (36) at (0.5, 0.25) {and};
	\end{pgfonlayer}
	\begin{pgfonlayer}{edgelayer}
		\draw (10.center) to (0);
		\draw [bend left=15] (4.center) to (0);
		\draw [bend left=15] (0) to (6.center);
		\draw [bend left=15] (0) to (2);
		\draw [bend left=15] (2) to (1);
		\draw [bend right=15] (0) to (3);
		\draw [bend right=15] (3) to (1);
		\draw (11.center) to (1);
		\draw [bend left=15] (1) to (7.center);
		\draw [bend right=15] (1) to (9.center);
		\draw [bend left=15] (13.center) to (21);
		\draw [bend right=15] (15.center) to (21);
		\draw (19.center) to (21);
		\draw (20.center) to (22);
		\draw [bend left=15] (22) to (16.center);
		\draw [bend right=15] (22) to (18.center);
		\draw (25.center) to (26);
		\draw [bend left=15] (26) to (27.center);
		\draw [bend right=15] (26) to (24.center);
		\draw (30.center) to (31);
		\draw [bend left=15] (31) to (32.center);
		\draw [bend right=15] (31) to (29.center);
		\draw [bend left=60, looseness=1.50] (26) to (33);
		\draw [bend right=60, looseness=1.50] (26) to (33);
	\end{pgfonlayer}
  \end{tikzpicture}
 \end{center}
 so we can remove any double edges or self-loops resulting from the merging. The result is a GS-LC diagram on $n-1$ qubits, where $n$ is the number of qubits in the original GS-LC diagram.

 \textit{One operand vertex is connected only to the other, but the latter has a non-operand neighbour}: We can change the vertex operator of the second operand vertex to the identity as in the previous case. In the process, the first operand vertex may aquire one or more non-operand neighbours; in that case we proceed as above. Else, by \eqref{eq:LC}, for any vertex operator $V$ on the first operand qubit,
 
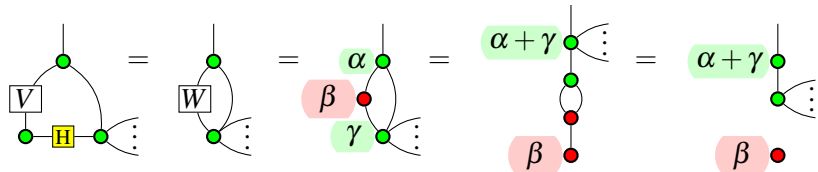
\begin{center}
  \begin{tikzpicture}
	\begin{pgfonlayer}{nodelayer}
		\node [style=gn] (0) at (-4.5, 0.25) {};
		\node [style=normalrect] (1) at (-5, -0.25) {$V$};
		\node [style=gn] (2) at (-5, -0.75) {};
		\node [style=gn] (3) at (-4, -0.75) {};
		\node [style=Hadamard] (4) at (-4.5, -0.75) {};
		\node [style=none] (5) at (-4.5, 0.75) {};
		\node [style=none] (6) at (-3.5, -0.5) {};
		\node [style=none] (7) at (-3.55, -0.65) {$\vdots$};
		\node [style=none] (8) at (-3.5, -1) {};
		\node [style=none] (9) at (-3.5, 0.25) {$=$};
		\node [style=gn] (10) at (-2.5, 0.25) {};
		\node [style=gn] (11) at (-2.5, -0.75) {};
		\node [style=normalrect] (12) at (-2.75, -0.25) {$W$};
		\node [style=none] (13) at (-2.5, 0.75) {};
		\node [style=none] (14) at (-2, -1) {};
		\node [style=none] (15) at (-2.05, -0.65) {$\vdots$};
		\node [style=none] (16) at (-2, -0.5) {};
		\node [style=none] (17) at (-1.5, 0.25) {$=$};
		\node [style=gn,label={[gphase]left:$\alpha$}] (18) at (-0.25, 0.25) {};
		\node [style=rn,label={[rphase]left:$\beta$}] (19) at (-0.5, -0.25) {};
		\node [style=gn,label={[gphase]left:$\gamma$}] (20) at (-0.25, -0.75) {};
		\node [style=none] (21) at (-0.25, 0.75) {};
		\node [style=none] (22) at (0.25, -1) {};
		\node [style=none] (23) at (0.2, -0.65) {$\vdots$};
		\node [style=none] (24) at (0.25, -0.5) {};
		\node [style=none] (25) at (0.5, 0.25) {$=$};
		\node [style=gn] (26) at (2.25, -0) {};
		\node [style=gn,label={[gphase]left:$\alpha+\gamma$}] (27) at (2.25, 0.5) {};
		\node [style=rn] (28) at (2.25, -0.5) {};
		\node [style=rn,label={[rphase]left:$\beta$}] (29) at (2.25, -1) {};
		\node [style=none] (30) at (2.25, 1) {};
		\node [style=none] (31) at (2.75, 0.75) {};
		\node [style=none] (32) at (2.7, 0.6) {$\vdots$};
		\node [style=none] (33) at (2.75, 0.25) {};
		\node [style=none] (34) at (3.25, 0.25) {$=$};
		\node [style=gn] (35) at (5, -0.25) {};
		\node [style=gn,label={[gphase]left:$\alpha+\gamma$}] (36) at (5, 0.25) {};
		\node [style=rn,label={[rphase]left:$\beta$}] (37) at (5, -1) {};
		\node [style=none] (38) at (5.5, -0) {};
		\node [style=none] (39) at (5.45, -0.15) {$\vdots$};
		\node [style=none] (40) at (5.5, -0.5) {};
		\node [style=none] (41) at (5, 0.75) {};
	\end{pgfonlayer}
	\begin{pgfonlayer}{edgelayer}
		\draw (5.center) to (0);
		\draw [bend right=15] (0) to (1);
		\draw (1) to (2);
		\draw (2) to (4);
		\draw (4) to (3);
		\draw [bend left=15] (3) to (6.center);
		\draw [bend right=15] (3) to (8.center);
		\draw (13.center) to (10);
		\draw [bend left=15] (11) to (16.center);
		\draw [bend right=15] (11) to (14.center);
		\draw [bend left=45] (10) to (11);
		\draw [bend right=15] (10) to (12);
		\draw [bend right=15] (12) to (11);
		\draw (21.center) to (18);
		\draw [bend right=15] (18) to (19);
		\draw [bend right=15] (19) to (20);
		\draw [bend left] (18) to (20);
		\draw [bend left=15] (20) to (24.center);
		\draw [bend right=15] (20) to (22.center);
		\draw (30.center) to (27);
		\draw (27) to (26);
		\draw (28) to (29);
		\draw [bend right=45] (26) to (28);
		\draw [bend left=45] (26) to (28);
		\draw [bend left=15] (27) to (31.center);
		\draw [bend right=15] (27) to (33.center);
		\draw (41.center) to (36);
		\draw (36) to (35);
		\draw [bend left=15] (35) to (38.center);
		\draw [bend right=15] (35) to (40.center);
		\draw [bend left] (0) to (3);
	\end{pgfonlayer}
  \end{tikzpicture}
 \end{center}
 where $W=V\circ H$ and we have used the spider law and the Hopf law. Again, if $\beta=\pi$ the resulting diagram is the zero diagram, otherwise it is equal to some GS-LC diagram.

 The four cases we have considered cover all the possible configurations of the graph underlying the original GS-LC diagram, hence the proof is complete.
\end{proof}

\subsection{Reduced GS-LC states}

Following \cite{elliott_graphical_2008}, we now define a more normalized form of GS-LC diagrams. The resulting diagrams are still not unique, but the number of equivalent diagrams is significantly smaller.

\begin{dfn}\label{dfn:rGS-LC}
 A \emph{stabilizer state diagram in reduced GS-LC (or rGS-LC) form} is a diagram in GS-LC form satisfying the following additional constraints:
 \begin{menum}
  \item\label{it:rGS-LC:VO} All vertex operators belong to the set
   \begin{equation}\label{eq:reduced_vertex_operators}
    R = 
    \begin{tikzpicture}[baseline=-0.1cm]
	\begin{pgfonlayer}{nodelayer}
		\node [style=gn,label={[gphase]right:$\pi/2$}] (0) at (-4, -0) {};
		\node [style=gn,label={[gphase]right:$-\pi/2$}] (1) at (2.25, -0.3) {};
		\node [style=rn,label={[rphase]right:$\pi/2$}] (2) at (2.25, 0.3) {};
		\node [style=none] (3) at (-4, 0.6) {};
		\node [style=none] (4) at (2.25, 0.6) {};
		\node [style=none] (5) at (-4, -0.6) {};
		\node [style=none] (6) at (2.25, -0.6) {};
		\node [style=none] (7) at (-5.25, -0) {$\Bigg\{$};
		\node [style=none] (8) at (4, -0) {$\Bigg\}.$};
		\node [style=gn,label={[gphase]right:$\pi$}] (9) at (-2.5, -0) {};
		\node [style=none] (10) at (-2.5, 0.6) {};
		\node [style=none] (11) at (-2.5, -0.6) {};
		\node [style=none] (12) at (-1.4, 0.6) {};
		\node [style=gn,label={[gphase]right:$-\pi/2$}] (13) at (-1.4, -0) {};
		\node [style=none] (14) at (-1.4, -0.6) {};
		\node [style=none] (15) at (-4.75, 0.6) {};
		\node [style=none] (16) at (-4.75, -0.6) {};
		\node [style=none] (17) at (0.5, -0.6) {};
		\node [style=none] (18) at (0.5, 0.6) {};
		\node [style=rn,label={[rphase]right:$\pi/2$}] (19) at (0.5, 0.3) {};
		\node [style=gn,label={[gphase]right:$\pi/2$}] (20) at (0.5, -0.3) {};
	\end{pgfonlayer}
	\begin{pgfonlayer}{edgelayer}
		\draw (3.center) to (5.center);
		\draw (4.center) to (6.center);
		\draw (10.center) to (11.center);
		\draw (12.center) to (14.center);
		\draw (15.center) to (16.center);
		\draw (18.center) to (17.center);
	\end{pgfonlayer}
    \end{tikzpicture}
   \end{equation}
  \item\label{it:rGS-LC:H} Two adjacent vertices must not both have vertex operators that include red nodes.
 \end{menum}
\end{dfn}

\begin{thm}\label{thm:ZX_rGS-LC}
 Any stabilizer state diagram is equal to some rGS-LC diagram within the \ZX-calculus.
\end{thm}
\begin{proof}
 By theorem \ref{thm:ZX_GS-LC}, any stabilizer state diagram is equal to some GS-LC diagram within the \ZX-calculus. Lemma \ref{lem:local_Clifford_normal_form} shows that each vertex operator in the GS-LC diagram can be brought into one of the forms
 \begin{center}
  \begin{tikzpicture}[baseline=0cm]
	\begin{pgfonlayer}{nodelayer}
		\node [style=gn,label={[gphase]right:$\alpha$}] (0) at (-1.75, 0.3) {};
		\node [style=rn,label={[rphase]right:$\beta$}] (1) at (-1.75, -0.3) {};
		\node [style=rn,label={[rphase]right:$\pi/2$}] (2) at (1, 0.6) {};
		\node [style=gn,label={[gphase]right:$\pm\pi/2$}] (3) at (1, -0) {};
		\node [style=rn,label={[rphase]right:$\gamma$}] (4) at (1, -0.6) {};
		\node [style=none] (5) at (0, -0) {or};
		\node [style=none] (6) at (-1.75, 0.55) {};
		\node [style=none] (7) at (-1.75, -0.55) {};
		\node [style=none] (8) at (1, -0.85) {};
		\node [style=none] (9) at (1, 0.85) {};
	\end{pgfonlayer}
	\begin{pgfonlayer}{edgelayer}
		\draw (6.center) to (7.center);
		\draw (9.center) to (8.center);
	\end{pgfonlayer}
  \end{tikzpicture}
 \end{center}
 Note that the cases $\beta=0$ and $\gamma=0$ of the above normal forms correspond exactly to the elements of  $R$, the restricted set of vertex operators given in \eqref{eq:reduced_vertex_operators}. A local complementation about a vertex $v$ pre-multiplies the vertex operator of $v$ with \phase{rn,label={[rphase]right:$-\pi/2$}}, so any vertex operator can be brought into one of the forms \eqref{eq:reduced_vertex_operators} by some number of local complementations about the corresponding vertex. The other effects of local complementations are to toggle some of the edges in the graph state and to pre-multiply the vertex operators of neighbouring vertices by \phase{gn,label={[gphase]right:$\pi/2$}}. The set $R$ is not mapped to itself under repeated pre-multiplication with green phase operators: this operation sends the set $\{$\phase{gn,label={[gphase]right:$\alpha$}}$\}$ to itself, but it maps
 \begin{center}
   \begin{tikzpicture}[baseline=-0.1]
	\begin{pgfonlayer}{nodelayer}
		\node [style=gn,label={[gphase]right:$\pm\pi/2$}] (0) at (-2.25, -0.3) {};
		\node [style=rn,label={[rphase]right:$\pi/2$}] (1) at (-2.25, 0.3) {};
		\node [style=none] (2) at (-2.25, 0.6) {};
		\node [style=none] (3) at (-2.25, -0.6) {};
		\node [style=none] (4) at (-3.25, -0) {$\Bigg\}\mapsto\Bigg\{$};
		\node [style=none] (5) at (-5.25, -0.6) {};
		\node [style=none] (6) at (-5.25, 0.6) {};
		\node [style=rn,label={[rphase]right:$\pi/2$}] (7) at (-5.25, 0.3) {};
		\node [style=gn,label={[gphase]right:$\pm\pi/2$}] (8) at (-5.25, -0.3) {};
		\node [style=rn,label={[rphase]right:$\pi/2$}] (9) at (-0.25, -0) {};
		\node [style=rn,label={[rphase]right:$-\pi/2$}] (10) at (1.5, -0.3) {};
		\node [style=gn,label={[gphase]right:$\pi$}] (11) at (1.5, 0.3) {};
		\node [style=none] (12) at (-0.25, 0.6) {};
		\node [style=none] (13) at (1.5, 0.6) {};
		\node [style=none] (14) at (-0.25, -0.6) {};
		\node [style=none] (15) at (1.5, -0.6) {};
		\node [style=none] (16) at (-5.75, -0) {$\Bigg\{$};
		\node [style=none] (17) at (3.25, -0) {$\Bigg\}.$};
	\end{pgfonlayer}
	\begin{pgfonlayer}{edgelayer}
		\draw (2.center) to (3.center);
		\draw (6.center) to (5.center);
		\draw (12.center) to (14.center);
		\draw (13.center) to (15.center);
	\end{pgfonlayer}
   \end{tikzpicture}
 \end{center}
 The normal form of a vertex operator contains at most two red nodes. Once a vertex operator is in one of the forms in $R$, pre-multiplication by green phase operators does not change the number of red nodes it contains when expressed in normal form. Thus the process of removing red nodes from the vertex operators by applying local complementations must terminate after at most $2n$ steps for an $n$-qubit diagram, at which point all vertex operators are elements of the set $R$.

 With all vertex operators in $R$, suppose there are two adjacent qubits $u$ and $v$ which both have red nodes in their vertex operators, i.e.\ there is a subdiagram of the form
  \begin{equation}\label{eq:neighbouring_red_nodes}
   \begin{tikzpicture}[baseline=0.6cm]
	\begin{pgfonlayer}{nodelayer}
		\node [style=gn] (0) at (-0.5, -0) {};
		\node [style=gn] (1) at (0.5, -0) {};
		\node [style=Hadamard] (2) at (0, -0) {};
		\node [style=gn,label={[gphase]left:$a\pi/2$}] (3) at (-0.5, 0.6) {};
		\node [style=gn,label={[gphase]right:$b\pi/2$}] (4) at (0.5, 0.6) {};
		\node [style=rn,label={[rphase]left:$\pi/2$}] (5) at (-0.5, 1.2) {};
		\node [style=rn,label={[rphase]right:$\pi/2$}] (6) at (0.5, 1.2) {};
		\node [style=none,label={above:$u$}] (7) at (-0.5, 1.45) {};
		\node [style=none,label={above:$v$}] (8) at (0.5, 1.45) {};
		\node [style=none] (9) at (-0.8, -0.5) {};
		\node [style=none] (10) at (-0.5, -0.45) {$\ldots$};
		\node [style=none] (11) at (-0.2, -0.5) {};
		\node [style=none] (12) at (0.2, -0.5) {};
		\node [style=none] (13) at (0.5, -0.45) {$\ldots$};
		\node [style=none] (14) at (0.8, -0.5) {};
	\end{pgfonlayer}
	\begin{pgfonlayer}{edgelayer}
		\draw (7.center) to (0);
		\draw (0) to (1);
		\draw (8.center) to (1);
		\draw [bend right=15] (0) to (9.center);
		\draw [bend left=15] (0) to (11.center);
		\draw [bend right=15] (1) to (12.center);
		\draw [bend left=15] (1) to (14.center);
	\end{pgfonlayer}
   \end{tikzpicture}
  \end{equation}
 for $a,b\in\{\pm 1\}$. A local complementation along the edge $\{u,v\}$ maps the vertex operator of $u$ to
 \begin{equation}\label{eq:pi_2+api_2+lc-edge}
  \begin{tikzpicture}[baseline=0cm]
	\begin{pgfonlayer}{nodelayer}
		\node [style=rn,label={[rphase]right:$\pi/2$}] (0) at (-3, 0.9) {};
		\node [style=rn,label={[rphase]right:$-\pi/2$}] (1) at (-3, -0.3) {};
		\node [style=gn,label={[gphase]right:$(a+1)\pi/2$}] (3) at (-3, 0.3) {};
		\node [style=gn,label={[gphase]right:$\pi/2$}] (4) at (-3, -0.9) {};
		\node [style=none] (5) at (-3, 1.15) {};
		\node [style=none] (6) at (-3, -1.15) {};
		\node [style=none] (7) at (-0.5, -0) {$=$};
		\node [style=rn,label={[rphase]right:$(a+1)\pi/2$}] (8) at (0, 0.3) {};
		\node [style=gn,label={[gphase]right:$-a\pi/2$}] (9) at (0, -0.3) {};
		\node [style=none] (10) at (0, 0.55) {};
		\node [style=none] (11) at (0, -0.55) {};
		\node [style=none] (12) at (2.5, -0) {$=$};
		\node [style=gn,label={[gphase]right:$\pi/2$}] (13) at (3, 0.3) {};
		\node [style=rn,label={[rphase]right:$(a+1)\pi/2$}] (14) at (3, -0.3) {};
		\node [style=none] (15) at (3, 0.55) {};
		\node [style=none] (16) at (3, -0.55) {};
		\node [style=rn,label={[rphase]right:$\pi/2$}] (17) at (-5, 1.2) {};
		\node [style=rn,label={[rphase]right:$-\pi/2$}] (18) at (-5, -0.6) {};
		\node [style=gn,label={[gphase]right:$a\pi/2$}] (19) at (-5, 0.6) {};
		\node [style=gn,label={[gphase]right:$\pi/2$}] (20) at (-5, -0) {};
		\node [style=gn,label={[gphase]right:$\pi/2$}] (21) at (-5, -1.2) {};
		\node [style=none] (22) at (-5, 1.45) {};
		\node [style=none] (23) at (-5, -1.45) {};
		\node [style=none] (7) at (-3.5, -0) {$=$};
	\end{pgfonlayer}
	\begin{pgfonlayer}{edgelayer}
		\draw (5.center) to (6.center);
		\draw (10.center) to (11.center);
		\draw (15.center) to (16.center);
		\draw (22.center) to (23.center);
	\end{pgfonlayer}
  \end{tikzpicture}
 \end{equation}
 and similarly for $v$. If $a=1$, we apply a fixpoint operation to $u$ and if $b=1$, we apply one to $v$. After this, the vertex operators on both $u$ and $v$ are green phase operators. Vertex operators of qubits adjacent to $u$ or $v$ are pre-multiplied with some power of \phase{gn,label={[gphase]right:$\pi$}}. Thus each such operation removes the red nodes from a pair of adjacent qubits and leaves all vertex operators in the set $R$. Hence after at most $n/2$ such operations, it will be impossible to find a subdiagram as in \eqref{eq:neighbouring_red_nodes}. Thus, the diagram is in reduced GS-LC form.
\end{proof}

\subsection{Equivalence transformations of rGS-LC diagrams}\label{s:equivalence_rGS-LC}

It is obvious that local complementations and applications of the fixpoint rule do not in general take rGS-LC diagrams to rGS-LC diagrams. Still, rGS-LC forms are not necessarily unique, as the following two propositions show. The propositions are adapted from similar results in \cite{elliott_graphical_2008}.

\begin{prop}\label{prop:rGS-LC_transformation1}
 Suppose a rGS-LC diagram contains a pair of neighbouring qubits $p$ and $q$ in the following configuration
 \begin{center}
  \begin{tikzpicture}
	\begin{pgfonlayer}{nodelayer}
		\node [style=gn] (0) at (-0.5, -0) {};
		\node [style=gn] (1) at (0.5, -0) {};
		\node [style=Hadamard] (2) at (0, -0) {};
		\node [style=gn,label={[gphase]left:$a\pi/2$}] (3) at (-0.5, 0.6) {};
		\node [style=gn,label={[gphase]right:$b\pi$}] (4) at (0.5, 0.6) {};
		\node [style=rn,label={[rphase]left:$\pi/2$}] (5) at (-0.5, 1.2) {};
		\node [style=none,label={right:$p$}] (7) at (-0.5, 1.45) {};
		\node [style=none,label={right:$q$}] (8) at (0.5, 1.45) {};
		\node [style=none] (9) at (-0.8, -0.5) {};
		\node [style=none] (10) at (-0.5, -0.45) {$\ldots$};
		\node [style=none] (11) at (-0.2, -0.5) {};
		\node [style=none] (12) at (0.2, -0.5) {};
		\node [style=none] (13) at (0.5, -0.45) {$\ldots$};
		\node [style=none] (14) at (0.8, -0.5) {};
	\end{pgfonlayer}
	\begin{pgfonlayer}{edgelayer}
		\draw (7.center) to (0);
		\draw (0) to (1);
		\draw (8.center) to (1);
		\draw [bend right=15] (0) to (9.center);
		\draw [bend left=15] (0) to (11.center);
		\draw [bend right=15] (1) to (12.center);
		\draw [bend left=15] (1) to (14.center);
	\end{pgfonlayer}
  \end{tikzpicture}
 \end{center}
 where $a\in\{\pm 1\}$ and $b\in\{0,1\}$. Then a local complementation about $q$, followed by a local complementation about $p$, yields a diagram which can be brought into rGS-LC form by at most two applications of the fixpoint rule.
\end{prop}

\begin{proof}
 Consider first the effect of the two local complementations on the vertex operators of $p$ and $q$. We have
 \begin{center}
  \begin{tikzpicture}
	\begin{pgfonlayer}{nodelayer}
		\node [style=rn,label={[rphase]right:$\pi/2$}] (0) at (-7.5, 0.9) {};
		\node [style=rn,label={[rphase]right:$-\pi/2$}] (1) at (-7.5, -0.9) {};
		\node [style=gn,label={[gphase]right:$a\pi/2$}] (2) at (-7.5, 0.3) {};
		\node [style=gn,label={[gphase]right:$\pi/2$}] (3) at (-7.5, -0.3) {};
		\node [style=none] (4) at (-7.5, 1.15) {};
		\node [style=none] (5) at (-7.5, -1.15) {};
		\node [style=none] (6) at (-6, -0) {$=$};
		\node [style=rn,label={[rphase]right:$\pi/2$}] (7) at (-5.5, 0.6) {};
		\node [style=rn,label={[rphase]right:$-\pi/2$}] (8) at (-5.5, -0.6) {};
		\node [style=gn,label={[gphase]right:$(a+1)\pi/2$}] (9) at (-5.5, -0) {};
		\node [style=none] (10) at (-5.5, 0.85) {};
		\node [style=none] (11) at (-5.5, -0.85) {};
		\node [style=none] (12) at (-3.05, -0) {$=$};
		\node [style=gn,label={[gphase]right:$(a+1)\pi/2$}] (13) at (-2.6, 0.3) {};
		\node [style=rn,label={[rphase]right:$(a+1)\pi/2$}] (14) at (-2.6, -0.3) {};
		\node [style=none] (15) at (-2.6, 0.55) {};
		\node [style=none] (16) at (-2.6, -0.55) {};
		\node [style=gn,label={[gphase]right:$b\pi$}] (17) at (1.25, 0.6) {};
		\node [style=rn,label={[rphase]right:$-\pi/2$}] (18) at (1.25, -0) {};
		\node [style=gn,label={[gphase]right:$\pi/2$}] (19) at (1.25, -0.6) {};
		\node [style=none] (20) at (1.25, 0.85) {};
		\node [style=none] (21) at (1.25, -0.85) {};
		\node [style=none] (22) at (2.9, -0) {$=$};
		\node [style=rn,label={[rphase]right:$(-1)^{b+1}\pi/2$}] (23) at (3.4, 0.32) {};
		\node [style=gn,label={[gphase]right:$(-1)^b\pi/2$}] (24) at (3.4, -0.32) {};
		\node [style=none] (25) at (3.4, 0.6) {};
		\node [style=none] (26) at (3.4, -0.6) {};
		\node [style=none] (27) at (6, -0) {$=$};
		\node [style=rn,label={[rphase]right:$\pi/2$}] (28) at (6.5, 0.6) {};
		\node [style=gn,label={[gphase]right:$-\pi/2$}] (29) at (6.5, -0) {};
		\node [style=rn,label={[rphase]right:$(1+b)\pi$}] (30) at (6.5, -0.6) {};
		\node [style=none] (31) at (6.5, 0.85) {};
		\node [style=none] (32) at (6.5, -0.85) {};
		\node [style=none] (27) at (0.25, -0) {and};
	\end{pgfonlayer}
	\begin{pgfonlayer}{edgelayer}
		\draw (4.center) to (5.center);
		\draw (10.center) to (11.center);
		\draw (15.center) to (16.center);
		\draw (20.center) to (21.center);
		\draw (25.center) to (26.center);
		\draw (31.center) to (32.center);
	\end{pgfonlayer}
  \end{tikzpicture}
 \end{center}
 If $a=+1$, we apply a fixpoint operation to $p$ and if $b=0$, we apply a fixpoint operation to $q$; then the vertex operators of $p$ and $q$ are in $R$. The fixpoint operations add \phase{gn,label={[gphase]right:$\pi$}} to neighbouring qubits, which maps the set $R$ to itself. As fixpoint operations do not change any edges, we do not have to worry about them when considering whether the rest of the diagram satisfies definition \ref{dfn:rGS-LC}.

 We first need to check that the two local complementations map all vertex operators to allowed ones. Vertices not adjacent to $p$ or $q$ can safely be ignored because their vertex operators remain unchanged. As the local complementations and fixpoint operations add only green phase operators to vertices other than $p$ and $q$, any vertex operator on another qubit that started out as a green phase will remain a green phase. It remains to check the effect of the transformation on qubits whose vertex operator contains a red node and which are adjacent to $p$ or $q$. By definition \ref{dfn:rGS-LC}, such qubits cannot be adjacent to $p$. So suppose $w$ is a qubit in the original graph state with a red node in its vertex operator and suppose the initial diagram contains an edge $\{w,q\}$. Then the local complementation about $q$ adds a phase factor \phase{gn,label={[gphase]right:$\pi/2$}} to the vertex operator of $w$ and it creates an edge between $w$ and $p$. The complementation about $p$ adds another \phase{gn,label={[gphase]right:$\pi/2$}} to $w$ and removes the edge between $w$ and $q$. Thus the vertex operator of $w$ remains in the set $R$, i.e. the transformation preserves property 1 of the definition of rGS-LC diagrams.

 Suppose there are two qubits $v,w$ in the original graph state, both of which have red nodes in their vertex operators and are adjacent to $q$. Since the original diagram is in rGS-LC form, there is no edge between $v$ and $w$. Now the local complementation about $q$ adds an edge between $v$ and $w$ and creates edges $\{p,v\}$ and $\{p,w\}$. The local complementation about $p$ removes the edge $\{v,w\}$, so once again $v$ and $w$ are not adjacent. Edges involving any qubits that are not adjacent to $p$ or $q$ remain unchanged. Thus the transformation preserves property 2 of definition \ref{dfn:rGS-LC}. Hence, the resulting diagram is in rGS-LC form.
\end{proof}

\begin{prop}\label{prop:rGS-LC_transformation2}
 Suppose a rGS-LC diagram contains a pair of neighbouring qubits $p$ and $q$ in the following configuration
 
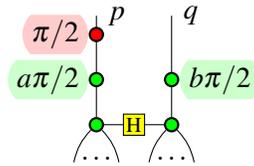
\begin{center}
  \begin{tikzpicture}
	\begin{pgfonlayer}{nodelayer}
		\node [style=gn] (0) at (-0.5, -0) {};
		\node [style=gn] (1) at (0.5, -0) {};
		\node [style=Hadamard] (2) at (0, -0) {};
		\node [style=gn,label={[gphase]left:$a\pi/2$}] (3) at (-0.5, 0.6) {};
		\node [style=gn,label={[gphase]right:$b\pi/2$}] (4) at (0.5, 0.6) {};
		\node [style=rn,label={[rphase]left:$\pi/2$}] (5) at (-0.5, 1.2) {};
		\node [style=none,label={right:$p$}] (7) at (-0.5, 1.45) {};
		\node [style=none,label={right:$q$}] (8) at (0.5, 1.45) {};
		\node [style=none] (9) at (-0.8, -0.5) {};
		\node [style=none] (10) at (-0.5, -0.45) {$\ldots$};
		\node [style=none] (11) at (-0.2, -0.5) {};
		\node [style=none] (12) at (0.2, -0.5) {};
		\node [style=none] (13) at (0.5, -0.45) {$\ldots$};
		\node [style=none] (14) at (0.8, -0.5) {};
	\end{pgfonlayer}
	\begin{pgfonlayer}{edgelayer}
		\draw (7.center) to (0);
		\draw (0) to (1);
		\draw (8.center) to (1);
		\draw [bend right=15] (0) to (9.center);
		\draw [bend left=15] (0) to (11.center);
		\draw [bend right=15] (1) to (12.center);
		\draw [bend left=15] (1) to (14.center);
	\end{pgfonlayer}
  \end{tikzpicture}
 \end{center}
 where $a,b\in\{\pm 1\}$. Then a local complementation along the edge $\{p,q\}$ yields a diagram which can be brought into rGS-LC form by at most two applications of the fixpoint rule.
\end{prop}

\begin{proof}
 After the local complementation along the edge, the vertex operator of $p$ is given by \eqref{eq:pi_2+api_2+lc-edge}. For the vertex operator of $q$, we have
 \begin{center}
  \begin{tikzpicture}
	\begin{pgfonlayer}{nodelayer}
		\node [style=gn,label={[gphase]right:$(b+1)\pi/2$}] (0) at (-4.75, 0.6) {};
		\node [style=rn,label={[rphase]right:$-\pi/2$}] (2) at (-4.75, -0) {};
		\node [style=gn,label={[gphase]right:$\pi/2$}] (3) at (-4.75, -0.6) {};
		\node [style=none] (4) at (-4.75, 0.85) {};
		\node [style=none] (5) at (-4.75, -0.85) {};
		\node [style=none] (6) at (-2.5, -0) {$=$};
		\node [style=rn,label={[rphase]right:$b\pi/2$}] (7) at (-2, 0.3) {};
		\node [style=gn,label={[gphase]right:$-b\pi/2$}] (8) at (-2, -0.3) {};
		\node [style=none] (9) at (-2, 0.55) {};
		\node [style=none] (10) at (-2, -0.55) {};
		\node [style=none] (11) at (-0, -0) {$=$};
		\node [style=rn,label={[rphase]right:$\pi/2$}] (12) at (0.5, 0.6) {};
		\node [style=gn,label={[gphase]right:$-\pi/2$}] (13) at (0.5, -0) {};
		\node [style=rn,label={[rphase]right:$(b-1)\pi/2$}] (14) at (0.5, -0.6) {};
		\node [style=none] (15) at (0.5, 0.85) {};
		\node [style=none] (16) at (0.5, -0.85) {};
		\node [style=gn,label={[gphase]right:$b\pi/2$}] (17) at (-7, 0.9) {};
		\node [style=gn,label={[gphase]right:$\pi/2$}] (18) at (-7, 0.3) {};
		\node [style=rn,label={[rphase]right:$-\pi/2$}] (19) at (-7, -0.3) {};
		\node [style=gn,label={[gphase]right:$\pi/2$}] (20) at (-7, -0.9) {};
		\node [style=none] (21) at (-7, 1.15) {};
		\node [style=none] (22) at (-7, -1.15) {};
		\node [style=none] (6) at (-5.25, -0) {$=$};
	\end{pgfonlayer}
	\begin{pgfonlayer}{edgelayer}
		\draw (4.center) to (5.center);
		\draw (9.center) to (10.center);
		\draw (15.center) to (16.center);
		\draw (21.center) to (22.center);
	\end{pgfonlayer}
  \end{tikzpicture}
 \end{center}
 Thus if $a=1$, we apply a fixpoint operation to $p$ and if $b=-1$, we apply a fixpoint operation to $q$. From the properties of local complementations along edges (see section \ref{s:equivalence_GS-LC}) it follows that this transformation preserves the two properties of rGS-LC states. 
\end{proof}
\section{Completeness}
\label{s:completeness}

\subsection{Comparing rGS-LC diagrams}
\label{s:equality_testing}

Theorem \ref{thm:ZX_rGS-LC} shows that any stabilizer state diagram is equal to some rGS-LC diagram. Thus, the \ZX-calculus is complete for stabilizer states if, given two rGS-LC diagrams representing the same state, we can show that they are equal using the rules of the \ZX-calculus. Again, we follow \cite{elliott_graphical_2008}.

\begin{dfn}
 A pair of rGS-LC diagrams on the same number of qubit is called \emph{simplified} if there are no pairs of qubits $p,q$ such that $p$ has a red node in its vertex operator in the first diagram but not in the second, $q$ has a red node in the second diagram but not in the first, and $p$ and $q$ are adjacent in at least one of the diagrams.
\end{dfn}

\begin{prop}
 Any pair of rGS-LC diagrams on $n$ qubits can be simplified.
\end{prop}
\begin{proof}
 Suppose there exists a pair of qubits $p,q$ such that $p$ has a red node in its vertex operator in the first diagram but not in the second, $q$ has a red node in the second diagram but not in the first, and $p$ and $q$ are adjacent in at least one of the diagrams. Then in the diagram in which they are adjacent, we can apply the apropriate one of the equivalence transformations given in section \ref{s:equivalence_rGS-LC}. The equivalence rules do not change the total number of red nodes among the vertex operators. Each such application pairs up red nodes between the two diagrams. Paired up qubits do not participate further in these transformations, therefore in less than $n$ steps the pair of diagrams is simplified.
\end{proof}

\begin{lem}\label{lem:unpaired_red_node}
 Consider a simplified pair of rGS-LC diagrams and suppose there exists an unpaired red node, i.e.\ there is a qubit $p$ which has a red node in its vertex operator in one of the diagrams, but not in the other. Then the two diagrams are not equal.
\end{lem}
\begin{proof}
 Let $D_1$ be the diagram in which $p$ has the red node, $D_2$ the other diagram. There are multiple cases:

 \emph{In either diagram, $p$ has no neighbours}: In this case, the overall quantum state factorises and the two diagrams are equal only if the two states of $p$ are the same. But
 \begin{center}
  \begin{tikzpicture}
	\begin{pgfonlayer}{nodelayer}
		\node [style=gn,label={[gphase]right:$\alpha$}] (0) at (-3.4, -0) {};
		\node [style=rn,label={[rphase]right:$-b\pi/2$}] (1) at (0.9, -0.3) {};
		\node [style=rn,label={[rphase]right:$\pi/2$}] (2) at (0.9, 0.3) {};
		\node [style=none] (3) at (0.9, 0.55) {};
		\node [style=none] (4) at (-3.4, 0.25) {};
		\node [style=none] (5) at (-2.5, -0) {$\neq$};
		\node [style=none] (6) at (-4.8, 0.5) {};
		\node [style=gn,label={[gphase]right:$\alpha$}] (7) at (-4.8, 0.3) {};
		\node [style=gn] (8) at (-4.8, -0.3) {};
		\node [style=none] (9) at (-3.9, -0) {$=$};
		\node [style=none] (10) at (0.35, -0) {$=$};
		\node [style=rn,label={[rphase]right:$(1-b)\pi/2$}] (11) at (-2, -0) {};
		\node [style=none] (12) at (-2, 0.25) {};
		\node [style=rn,label={[rphase]right:$\pi/2$}] (13) at (3.1, 0.3) {};
		\node [style=none] (14) at (2.6, -0) {$=$};
		\node [style=gn,label={[gphase]right:$b\pi/2$}] (15) at (3.1, -0.3) {};
		\node [style=none] (16) at (3.1, 0.55) {};
		\node [style=rn,label={[rphase]right:$\pi/2$}] (17) at (5, 0.6) {};
		\node [style=gn] (18) at (5, -0.6) {};
		\node [style=gn,label={[gphase]right:$b\pi/2$}] (19) at (5, -0) {};
		\node [style=none] (20) at (4.5, -0) {$=$};
		\node [style=none] (21) at (5, 0.85) {};
	\end{pgfonlayer}
	\begin{pgfonlayer}{edgelayer}
		\draw (4.center) to (0);
		\draw (3.center) to (1);
		\draw (6.center) to (8);
		\draw (12.center) to (11);
		\draw (16.center) to (15);
		\draw (21.center) to (18);
	\end{pgfonlayer}
  \end{tikzpicture}
 \end{center}
 for $\alpha\in\{0,\pi/2,\pi,-\pi/2\}$ and $b\in\{\pm 1\}$, so the diagrams must be unequal.

 \emph{$p$ is isolated in one of the diagrams but not in the other}: We know that two graph states are equal only if one can be transformed into the other via a sequence of local complementations. A local complementation never turns a vertex with neighbours into a vertex without neighbours, or conversely. Thus the two diagrams cannot be equal.

 \emph{$p$ has neighbours in both diagrams}: Let $N_1$ be the set of all qubits that are adjacent to $p$ in $D_1$, and define $N_2$ similarly. The vertex operators of any qubit in $N_1$ must be green phases in both diagrams. In $D_1$, this is because of the definition of rGS-LC states, in $D_2$ it is because the pair of diagrams is simplified. To both diagrams apply the operation
 \[
  U = \left(\bigotimes_{v\in N_1} \wedge X_{v\to p}\right)\circ R_Z^{p},
 \]
 where $R_Z^{p}$ denotes \phase{gn,label={[gphase]right:$\pi/2$}} on $p$, and $\wedge X_{v\to p}$ is a controlled-X operation with control $v$ and target $p$. The controlled-X gates with different controls and the same target commute, so this is well-defined. Furthermore, $U$ is invertible, so (in a slight abuse of notation) $U\circ D_1 = U\circ D_2 \Leftrightarrow D_1 = D_2$. We will show that, no matter what the properties of $D_1$ and $D_2$ are (beyond the existence of an unpaired red node on $p$),
 \begin{mitem}
  \item in $U\circ D_1$, qubit $p$ is in state \state{rn} or \state{rn,label={[rphase]right:$\pi$}};
  \item in $U\circ D_2$, $p$ is either entangled with other qubits, or in one of the states \state{gn,label={[gphase]right:$\phi$}}, where $\phi\in\{0,\pi/2,\pi,-\pi/2\}$.
 \end{mitem}
 By the arguments used in the first two cases, this implies that $U\circ D_1\neq U\circ D_2$ and therefore $D_1\neq D_2$.

 Let $n=\abs{N_1}$, $m=\abs{N_1\cap N_2}$, and suppose the qubits are arranged in such a way that the first $m$ elements of $N_1$ are those which are also elements of $N_2$, if there are any. Consider first the effect on diagram $D_1$. The local Clifford operator on $p$ combines with the $R_Z$ from $U$ to
 \[
  R_Z\circ R_X\circ R_Z^{\pm 1} = H\circ Z^{a},
 \]
 where $a = (1\mp 1)/2$. Thus $U\circ D_1$ is equal to
 
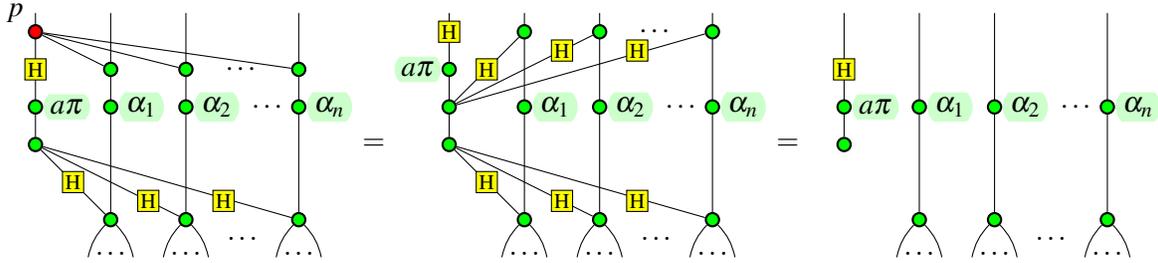
\begin{center}
  \begin{tikzpicture}
	\begin{pgfonlayer}{nodelayer}
		\node [style=gn] (0) at (-7.5, -0) {};
		\node [style=gn] (1) at (-6.5, -1) {};
		\node [style=gn] (2) at (-5.5, -1) {};
		\node [style=gn] (3) at (-4, -1) {};
		\node [style=gn,label={[gphase]right:$a\pi$}] (4) at (-7.5, 0.5) {};
		\node [style=gn,label={[gphase]right:$\alpha_1$}] (5) at (-6.5, 0.5) {};
		\node [style=gn,label={[gphase]right:$\alpha_2$}] (6) at (-5.5, 0.5) {};
		\node [style=gn,label={[gphase]right:$\alpha_n$}] (7) at (-4, 0.5) {};
		\node [style=Hadamard] (8) at (-7.5, 1) {};
		\node [style=none,label={left:$p$}] (9) at (-7.5, 1.75) {};
		\node [style=none] (10) at (-6.5, 1.75) {};
		\node [style=none] (11) at (-5.5, 1.75) {};
		\node [style=none] (12) at (-4, 1.75) {};
		\node [style=Hadamard] (13) at (-7, -0.5) {};
		\node [style=Hadamard] (14) at (-6, -0.75) {};
		\node [style=Hadamard] (15) at (-5, -0.75) {};
		\node [style=none] (16) at (-6.8, -1.5) {};
		\node [style=none] (17) at (-6.5, -1.45) {$\ldots$};
		\node [style=none] (18) at (-6.2, -1.5) {};
		\node [style=none] (19) at (-5.8, -1.5) {};
		\node [style=none] (20) at (-5.5, -1.45) {$\ldots$};
		\node [style=none] (21) at (-5.2, -1.5) {};
		\node [style=none] (22) at (-4.3, -1.5) {};
		\node [style=none] (23) at (-4, -1.45) {$\ldots$};
		\node [style=none] (24) at (-3.7, -1.5) {};
		\node [style=rn] (25) at (-7.5, 1.5) {};
		\node [style=gn] (26) at (-6.5, 1) {};
		\node [style=gn] (27) at (-5.5, 1) {};
		\node [style=gn] (28) at (-4, 1) {};
		\node [style=none] (29) at (-4.75, -1.25) {$\ldots$};
		\node [style=none] (30) at (-4.75, 1) {$\cdots$};
		\node [style=none] (31) at (-3, -0) {$=$};
		\node [style=Hadamard] (32) at (-0.5, -0.75) {};
		\node [style=none] (33) at (1.5, 1.75) {};
		\node [style=gn] (34) at (0, 1.5) {};
		\node [style=none] (35) at (0.3, -1.5) {};
		\node [style=gn,label={[gphase]right:$\alpha_n$}] (36) at (1.5, 0.5) {};
		\node [style=none] (37) at (0.75, 1.5) {$\cdots$};
		\node [style=gn] (38) at (1.5, 1.5) {};
		\node [style=none] (39) at (-0.3, -1.5) {};
		\node [style=none] (40) at (1.8, -1.5) {};
		\node [style=none] (41) at (-1, -1.45) {$\ldots$};
		\node [style=none] (42) at (0, -1.45) {$\ldots$};
		\node [style=Hadamard] (43) at (0.5, -0.75) {};
		\node [style=none] (44) at (1.5, -1.45) {$\ldots$};
		\node [style=none] (45) at (-0.7, -1.5) {};
		\node [style=none] (46) at (0.75, -1.25) {$\ldots$};
		\node [style=gn,label={[gphase]right:$\alpha_1$}] (47) at (-1, 0.5) {};
		\node [style=Hadamard] (48) at (-1.5, -0.5) {};
		\node [style=gn,label={[gphase]left:$a\pi$}] (49) at (-2, 1) {};
		\node [style=none] (50) at (-1, 1.75) {};
		\node [style=gn,label={[gphase]right:$\alpha_2$}] (51) at (0, 0.5) {};
		\node [style=Hadamard] (52) at (-2, 1.5) {};
		\node [style=gn] (53) at (1.5, -1) {};
		\node [style=gn] (54) at (-1, -1) {};
		\node [style=none] (55) at (1.2, -1.5) {};
		\node [style=gn] (56) at (0, -1) {};
		\node [style=none] (57) at (-1.3, -1.5) {};
		\node [style=gn] (58) at (-1, 1.5) {};
		\node [style=none] (59) at (0, 1.75) {};
		\node [style=gn] (60) at (-2, 0.5) {};
		\node [style=gn] (61) at (-2, -0) {};
		\node [style=none] (62) at (-2, 1.75) {};
		\node [style=Hadamard] (63) at (-1.5, 1) {};
		\node [style=Hadamard] (64) at (-0.5, 1.25) {};
		\node [style=Hadamard] (65) at (0.5, 1.25) {};
		\node [style=none] (66) at (5.25, 1.75) {};
		\node [style=gn,label={[gphase]right:$\alpha_1$}] (67) at (4.25, 0.5) {};
		\node [style=none] (68) at (4.25, 1.75) {};
		\node [style=none] (69) at (4.55, -1.5) {};
		\node [style=gn] (70) at (5.25, -1) {};
		\node [style=none] (71) at (3.95, -1.5) {};
		\node [style=gn,label={[gphase]right:$\alpha_2$}] (72) at (5.25, 0.5) {};
		\node [style=none] (73) at (3.25, 1.75) {};
		\node [style=Hadamard] (74) at (3.25, 1) {};
		\node [style=none] (75) at (6.45, -1.5) {};
		\node [style=none] (76) at (6, -1.25) {$\ldots$};
		\node [style=none] (77) at (6.35, 0.5) {$\cdots$};
		\node [style=none] (78) at (4.95, -1.5) {};
		\node [style=none] (79) at (6.75, -1.45) {$\ldots$};
		\node [style=none] (80) at (2.5, -0) {$=$};
		\node [style=none] (81) at (6.75, 1.75) {};
		\node [style=none] (82) at (5.55, -1.5) {};
		\node [style=gn] (83) at (6.75, -1) {};
		\node [style=gn,label={[gphase]right:$a\pi$}] (84) at (3.25, 0.5) {};
		\node [style=gn,label={[gphase]right:$\alpha_n$}] (85) at (6.75, 0.5) {};
		\node [style=none] (86) at (5.25, -1.45) {$\ldots$};
		\node [style=gn] (87) at (4.25, -1) {};
		\node [style=none] (88) at (7.05, -1.5) {};
		\node [style=none] (89) at (4.25, -1.45) {$\ldots$};
		\node [style=gn] (90) at (3.25, -0) {};
		\node [style=none] (91) at (-4.4, 0.5) {$\cdots$};
		\node [style=none] (92) at (1.1, 0.5) {$\cdots$};
	\end{pgfonlayer}
	\begin{pgfonlayer}{edgelayer}
		\draw (0) to (1);
		\draw (0) to (2);
		\draw (0) to (3);
		\draw (9.center) to (0);
		\draw (10.center) to (1);
		\draw (11.center) to (2);
		\draw (12.center) to (3);
		\draw [bend right=15] (1) to (16.center);
		\draw [bend left=15] (1) to (18.center);
		\draw [bend right=15] (2) to (19.center);
		\draw [bend left=15] (2) to (21.center);
		\draw [bend right=15] (3) to (22.center);
		\draw [bend left=15] (3) to (24.center);
		\draw (25) to (26);
		\draw (25) to (27);
		\draw (25) to (28);
		\draw (61) to (54);
		\draw (61) to (56);
		\draw (61) to (53);
		\draw (62.center) to (61);
		\draw (50.center) to (54);
		\draw (59.center) to (56);
		\draw (33.center) to (53);
		\draw [bend right=15] (54) to (57.center);
		\draw [bend left=15] (54) to (45.center);
		\draw [bend right=15] (56) to (39.center);
		\draw [bend left=15] (56) to (35.center);
		\draw [bend right=15] (53) to (55.center);
		\draw [bend left=15] (53) to (40.center);
		\draw (60) to (58);
		\draw (60) to (34);
		\draw (60) to (38);
		\draw (73.center) to (90);
		\draw (68.center) to (87);
		\draw (66.center) to (70);
		\draw (81.center) to (83);
		\draw [bend right=15] (87) to (71.center);
		\draw [bend left=15] (87) to (69.center);
		\draw [bend right=15] (70) to (78.center);
		\draw [bend left=15] (70) to (82.center);
		\draw [bend right=15] (83) to (75.center);
		\draw [bend left=15] (83) to (88.center);
	\end{pgfonlayer}
  \end{tikzpicture}
 \end{center}
 Here, $\alpha_k\in\{0,\pi/2,\pi,-\pi/2\}$ for $k=1,\ldots,n$ and we have used the fact that green nodes can be moved past each other. Note that at the end, qubit $p$ is isolated and in the state \state{rn,label={[rphase]right:$a\pi$}}.

 Next consider $D_2$. As $N_2$ is not in general equal to $N_1$, there may be qubits adjacent to $p$ which do not have controlled-X gates applied to them, qubits which have controlled-X gates applied to them but are not adjacent to $p$, and qubits which are adjacent to $p$ and have controlled-X gates applied to them. In the following diagram, $\beta$ and $\gamma_1,\ldots,\gamma_n$ are multiples of $\pi/2$ as usual. The phase $\beta$ is the combination of the original local Clifford operator on $p$ and the $R_Z$ part of $U$. Similar to before, we do not care about edges that do not involve $p$. This time we also neglect edges between $p$ and vertices not in $N_1$:
 
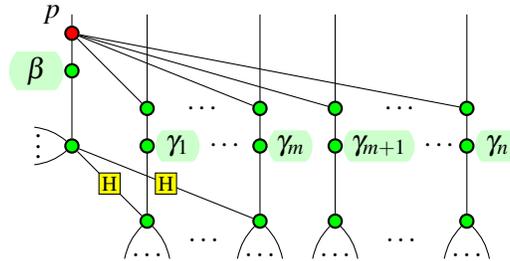
\begin{center}
  \begin{tikzpicture}
	\begin{pgfonlayer}{nodelayer}
		\node [style=gn] (0) at (-2.5, -0) {};
		\node [style=gn] (1) at (-1.5, -1) {};
		\node [style=gn] (2) at (0, -1) {};
		\node [style=gn,label={[gphase]left:$\beta$}] (3) at (-2.5, 1) {};
		\node [style=gn,label={[gphase]right:$\gamma_1$}] (4) at (-1.5, -0) {};
		\node [style=gn,label={[gphase]right:$\gamma_m$}] (5) at (0, -0) {};
		\node [style=none,label={left:$p$}] (6) at (-2.5, 1.75) {};
		\node [style=none] (7) at (-1.5, 1.75) {};
		\node [style=none] (8) at (0, 1.75) {};
		\node [style=Hadamard] (9) at (-2, -0.5) {};
		\node [style=Hadamard] (10) at (-1.25, -0.5) {};
		\node [style=none] (11) at (-1.8, -1.5) {};
		\node [style=none] (12) at (-1.5, -1.45) {$\ldots$};
		\node [style=none] (13) at (-1.2, -1.5) {};
		\node [style=none] (14) at (-0.3, -1.5) {};
		\node [style=none] (15) at (0, -1.45) {$\ldots$};
		\node [style=none] (16) at (0.3, -1.5) {};
		\node [style=rn] (17) at (-2.5, 1.5) {};
		\node [style=gn] (18) at (-1.5, 0.5) {};
		\node [style=gn] (19) at (0, 0.5) {};
		\node [style=none] (20) at (-0.75, -1.25) {$\ldots$};
		\node [style=none] (21) at (-0.75, 0.5) {$\cdots$};
		\node [style=none] (22) at (-3, 0.25) {};
		\node [style=none] (23) at (-2.95, 0.1) {$\vdots$};
		\node [style=none] (24) at (-3, -0.25) {};
		\node [style=none] (25) at (1, -1.45) {$\ldots$};
		\node [style=gn] (26) at (2.75, 0.5) {};
		\node [style=none] (27) at (1.875, -1.25) {$\ldots$};
		\node [style=none] (28) at (2.45, -1.5) {};
		\node [style=none] (29) at (3.05, -1.5) {};
		\node [style=gn] (30) at (1, -1) {};
		\node [style=gn,label={[gphase]right:$\gamma_{m+1}$}] (31) at (1, -0) {};
		\node [style=none] (32) at (1.3, -1.5) {};
		\node [style=none] (33) at (0.7, -1.5) {};
		\node [style=none] (34) at (1.875, 0.5) {$\cdots$};
		\node [style=gn,label={[gphase]right:$\gamma_n$}] (35) at (2.75, -0) {};
		\node [style=none] (36) at (2.75, 1.75) {};
		\node [style=gn] (37) at (2.75, -1) {};
		\node [style=none] (38) at (2.75, -1.45) {$\ldots$};
		\node [style=none] (39) at (1, 1.75) {};
		\node [style=gn] (40) at (1, 0.5) {};
		\node [style=none] (41) at (-0.45, -0) {$\cdots$};
		\node [style=none] (42) at (2.4, -0) {$\cdots$};
	\end{pgfonlayer}
	\begin{pgfonlayer}{edgelayer}
		\draw (0) to (1);
		\draw (0) to (2);
		\draw (6.center) to (0);
		\draw (7.center) to (1);
		\draw (8.center) to (2);
		\draw [bend right=15] (1) to (11.center);
		\draw [bend left=15] (1) to (13.center);
		\draw [bend right=15] (2) to (14.center);
		\draw [bend left=15] (2) to (16.center);
		\draw (17) to (18);
		\draw (17) to (19);
		\draw (39.center) to (30);
		\draw (36.center) to (37);
		\draw [bend right=15] (30) to (33.center);
		\draw [bend left=15] (30) to (32.center);
		\draw [bend right=15] (37) to (28.center);
		\draw [bend left=15] (37) to (29.center);
		\draw (17) to (40);
		\draw (17) to (26);
		\draw [bend right=15] (0) to (22.center);
		\draw [bend left=15] (0) to (24.center);
	\end{pgfonlayer}
  \end{tikzpicture}
 \end{center}
 We will distinguish different cases, depending on the value of $\beta$.

 If $\beta=\pi/2$, apply a local complementation and a fixpoint operation about $p$. This does not change the edges incident on $p$:
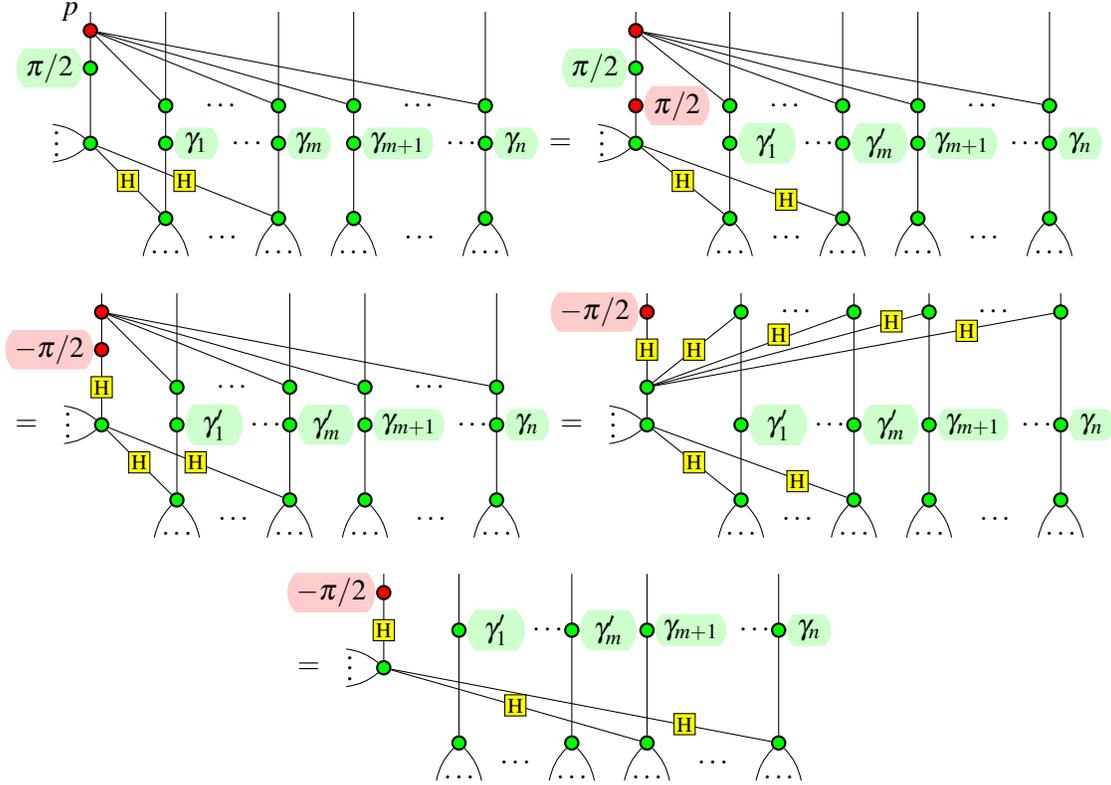
 \begin{center}
  \begin{tikzpicture}
	\begin{pgfonlayer}{nodelayer}
		\node [style=gn] (0) at (-6.25, -0) {};
		\node [style=gn] (1) at (-5.25, -1) {};
		\node [style=gn] (2) at (-3.75, -1) {};
		\node [style=gn,label={[gphase]left:$\pi/2$}] (3) at (-6.25, 1) {};
		\node [style=gn,label={[gphase]right:$\gamma_1$}] (4) at (-5.25, -0) {};
		\node [style=gn,label={[gphase]right:$\gamma_m$}] (5) at (-3.75, -0) {};
		\node [style=none,label={left:$p$}] (6) at (-6.25, 1.75) {};
		\node [style=none] (7) at (-5.25, 1.75) {};
		\node [style=none] (8) at (-3.75, 1.75) {};
		\node [style=Hadamard] (9) at (-5.75, -0.5) {};
		\node [style=Hadamard] (10) at (-5, -0.5) {};
		\node [style=none] (11) at (-5.55, -1.5) {};
		\node [style=none] (12) at (-5.25, -1.45) {$\ldots$};
		\node [style=none] (13) at (-4.95, -1.5) {};
		\node [style=none] (14) at (-4.05, -1.5) {};
		\node [style=none] (15) at (-3.75, -1.45) {$\ldots$};
		\node [style=none] (16) at (-3.45, -1.5) {};
		\node [style=rn] (17) at (-6.25, 1.5) {};
		\node [style=gn] (18) at (-5.25, 0.5) {};
		\node [style=gn] (19) at (-3.75, 0.5) {};
		\node [style=none] (20) at (-4.5, -1.25) {$\ldots$};
		\node [style=none] (21) at (-4.5, 0.5) {$\cdots$};
		\node [style=none] (22) at (-6.75, 0.25) {};
		\node [style=none] (23) at (-6.7, 0.1) {$\vdots$};
		\node [style=none] (24) at (-6.75, -0.25) {};
		\node [style=none] (25) at (-2.75, -1.45) {$\ldots$};
		\node [style=gn] (26) at (-1, 0.5) {};
		\node [style=none] (27) at (-1.875, -1.25) {$\ldots$};
		\node [style=none] (28) at (-1.3, -1.5) {};
		\node [style=none] (29) at (-0.7, -1.5) {};
		\node [style=gn] (30) at (-2.75, -1) {};
		\node [style=gn,label={[gphase]right:$\gamma_{m+1}$}] (31) at (-2.75, -0) {};
		\node [style=none] (32) at (-2.45, -1.5) {};
		\node [style=none] (33) at (-3.05, -1.5) {};
		\node [style=none] (34) at (-1.875, 0.5) {$\cdots$};
		\node [style=gn,label={[gphase]right:$\gamma_n$}] (35) at (-1, -0) {};
		\node [style=none] (36) at (-1, 1.75) {};
		\node [style=gn] (37) at (-1, -1) {};
		\node [style=none] (38) at (-1, -1.45) {$\ldots$};
		\node [style=none] (39) at (-2.75, 1.75) {};
		\node [style=gn] (40) at (-2.75, 0.5) {};
		\node [style=none] (41) at (0, -0) {$=$};
		\node [style=gn,label={[gphase]left:$\pi/2$}] (42) at (1, 1) {};
		\node [style=none] (43) at (6.8, -1.5) {};
		\node [style=gn,label={[gphase]right:$\gamma_1'$}] (44) at (2.25, -0) {};
		\node [style=Hadamard] (45) at (3, -0.75) {};
		\node [style=none] (46) at (1, 1.75) {};
		\node [style=gn,label={[gphase]right:$\gamma_{m+1}$}] (47) at (4.75, -0) {};
		\node [style=Hadamard] (48) at (1.625, -0.5) {};
		\node [style=none] (49) at (0.5, 0.25) {};
		\node [style=none] (50) at (2.25, -1.45) {$\ldots$};
		\node [style=none] (51) at (0.55, 0.1) {$\vdots$};
		\node [style=gn] (52) at (6.5, -1) {};
		\node [style=none] (53) at (0.5, -0.25) {};
		\node [style=gn] (54) at (1, -0) {};
		\node [style=none] (55) at (3.45, -1.5) {};
		\node [style=gn,label={[gphase]right:$\gamma_n$}] (56) at (6.5, -0) {};
		\node [style=none] (57) at (5.625, 0.5) {$\cdots$};
		\node [style=gn] (58) at (2.25, 0.5) {};
		\node [style=gn] (59) at (3.75, 0.5) {};
		\node [style=none] (60) at (3.75, -1.45) {$\ldots$};
		\node [style=none] (61) at (6.2, -1.5) {};
		\node [style=none] (62) at (2.25, 1.75) {};
		\node [style=none] (63) at (2.55, -1.5) {};
		\node [style=none] (64) at (4.05, -1.5) {};
		\node [style=none] (65) at (1.95, -1.5) {};
		\node [style=none] (66) at (6.5, -1.45) {$\ldots$};
		\node [style=none] (67) at (3, 0.5) {$\cdots$};
		\node [style=none] (68) at (6.5, 1.75) {};
		\node [style=gn,label={[gphase]right:$\gamma_m'$}] (69) at (3.75, -0) {};
		\node [style=gn] (70) at (3.75, -1) {};
		\node [style=none] (71) at (4.75, 1.75) {};
		\node [style=none] (72) at (5.625, -1.25) {$\ldots$};
		\node [style=gn] (73) at (4.75, 0.5) {};
		\node [style=none] (74) at (5.05, -1.5) {};
		\node [style=none] (75) at (3, -1.25) {$\ldots$};
		\node [style=none] (76) at (4.75, -1.45) {$\ldots$};
		\node [style=none] (77) at (3.75, 1.75) {};
		\node [style=rn] (78) at (1, 1.5) {};
		\node [style=gn] (79) at (6.5, 0.5) {};
		\node [style=gn] (80) at (4.75, -1) {};
		\node [style=gn] (81) at (2.25, -1) {};
		\node [style=none] (82) at (4.45, -1.5) {};
		\node [style=rn,label={[rphase]right:$\pi/2$}] (83) at (1, 0.5) {};
		\node [style=none] (84) at (-4.15, -0) {$\cdots$};
		\node [style=none] (85) at (-1.3, -0) {$\cdots$};
		\node [style=none] (86) at (6.2, -0) {$\cdots$};
		\node [style=none] (87) at (3.45, -0) {$\cdots$};
	\end{pgfonlayer}
	\begin{pgfonlayer}{edgelayer}
		\draw (0) to (1);
		\draw (0) to (2);
		\draw (6.center) to (0);
		\draw (7.center) to (1);
		\draw (8.center) to (2);
		\draw [bend right=15] (1) to (11.center);
		\draw [bend left=15] (1) to (13.center);
		\draw [bend right=15] (2) to (14.center);
		\draw [bend left=15] (2) to (16.center);
		\draw (17) to (18);
		\draw (17) to (19);
		\draw (39.center) to (30);
		\draw (36.center) to (37);
		\draw [bend right=15] (30) to (33.center);
		\draw [bend left=15] (30) to (32.center);
		\draw [bend right=15] (37) to (28.center);
		\draw [bend left=15] (37) to (29.center);
		\draw (17) to (40);
		\draw (17) to (26);
		\draw [bend right=15] (0) to (22.center);
		\draw [bend left=15] (0) to (24.center);
		\draw (54) to (81);
		\draw (54) to (70);
		\draw (46.center) to (54);
		\draw (62.center) to (81);
		\draw (77.center) to (70);
		\draw [bend right=15] (81) to (65.center);
		\draw [bend left=15] (81) to (63.center);
		\draw [bend right=15] (70) to (55.center);
		\draw [bend left=15] (70) to (64.center);
		\draw (78) to (58);
		\draw (78) to (59);
		\draw (71.center) to (80);
		\draw (68.center) to (52);
		\draw [bend right=15] (80) to (82.center);
		\draw [bend left=15] (80) to (74.center);
		\draw [bend right=15] (52) to (61.center);
		\draw [bend left=15] (52) to (43.center);
		\draw (78) to (73);
		\draw (78) to (79);
		\draw [bend right=15] (54) to (49.center);
		\draw [bend left=15] (54) to (53.center);
	\end{pgfonlayer}
  \end{tikzpicture}
 \end{center}
 \begin{center}
  \begin{tikzpicture}
	\begin{pgfonlayer}{nodelayer}
		\node [style=gn] (0) at (-6.25, -0) {};
		\node [style=gn] (1) at (-5.25, -1) {};
		\node [style=gn] (2) at (-3.75, -1) {};
		\node [style=Hadamard] (3) at (-6.25, 0.5) {};
		\node [style=gn,label={[gphase]right:$\gamma_1'$}] (4) at (-5.25, -0) {};
		\node [style=gn,label={[gphase]right:$\gamma_m'$}] (5) at (-3.75, -0) {};
		\node [style=none] (6) at (-6.25, 1.75) {};
		\node [style=none] (7) at (-5.25, 1.75) {};
		\node [style=none] (8) at (-3.75, 1.75) {};
		\node [style=Hadamard] (9) at (-5.75, -0.5) {};
		\node [style=Hadamard] (10) at (-5, -0.5) {};
		\node [style=none] (11) at (-5.55, -1.5) {};
		\node [style=none] (12) at (-5.25, -1.45) {$\ldots$};
		\node [style=none] (13) at (-4.95, -1.5) {};
		\node [style=none] (14) at (-4.05, -1.5) {};
		\node [style=none] (15) at (-3.75, -1.45) {$\ldots$};
		\node [style=none] (16) at (-3.45, -1.5) {};
		\node [style=rn] (17) at (-6.25, 1.5) {};
		\node [style=gn] (18) at (-5.25, 0.5) {};
		\node [style=gn] (19) at (-3.75, 0.5) {};
		\node [style=none] (20) at (-4.5, -1.25) {$\ldots$};
		\node [style=none] (21) at (-4.5, 0.5) {$\cdots$};
		\node [style=none] (22) at (-6.75, 0.25) {};
		\node [style=none] (23) at (-6.7, 0.1) {$\vdots$};
		\node [style=none] (24) at (-6.75, -0.25) {};
		\node [style=none] (25) at (-2.75, -1.45) {$\ldots$};
		\node [style=gn] (26) at (-1, 0.5) {};
		\node [style=none] (27) at (-1.875, -1.25) {$\ldots$};
		\node [style=none] (28) at (-1.3, -1.5) {};
		\node [style=none] (29) at (-0.7, -1.5) {};
		\node [style=gn] (30) at (-2.75, -1) {};
		\node [style=gn,label={[gphase]right:$\gamma_{m+1}$}] (31) at (-2.75, -0) {};
		\node [style=none] (32) at (-2.45, -1.5) {};
		\node [style=none] (33) at (-3.05, -1.5) {};
		\node [style=none] (34) at (-1.875, 0.5) {$\cdots$};
		\node [style=gn,label={[gphase]right:$\gamma_n$}] (35) at (-1, -0) {};
		\node [style=none] (36) at (-1, 1.75) {};
		\node [style=gn] (37) at (-1, -1) {};
		\node [style=none] (38) at (-1, -1.45) {$\ldots$};
		\node [style=none] (39) at (-2.75, 1.75) {};
		\node [style=gn] (40) at (-2.75, 0.5) {};
		\node [style=none] (41) at (0, -0) {$=$};
		\node [style=Hadamard] (42) at (1, 1) {};
		\node [style=none] (43) at (6.8, -1.5) {};
		\node [style=gn,label={[gphase]right:$\gamma_1'$}] (44) at (2.25, -0) {};
		\node [style=Hadamard] (45) at (3, -0.75) {};
		\node [style=none] (46) at (1, 1.75) {};
		\node [style=gn,label={[gphase]right:$\gamma_{m+1}$}] (47) at (4.75, -0) {};
		\node [style=Hadamard] (48) at (1.625, -0.5) {};
		\node [style=none] (49) at (0.5, 0.25) {};
		\node [style=none] (50) at (2.25, -1.45) {$\ldots$};
		\node [style=none] (51) at (0.55, 0.1) {$\vdots$};
		\node [style=gn] (52) at (6.5, -1) {};
		\node [style=none] (53) at (0.5, -0.25) {};
		\node [style=gn] (54) at (1, -0) {};
		\node [style=none] (55) at (3.45, -1.5) {};
		\node [style=gn,label={[gphase]right:$\gamma_n$}] (56) at (6.5, -0) {};
		\node [style=none] (57) at (5.625, 1.5) {$\cdots$};
		\node [style=gn] (58) at (2.25, 1.5) {};
		\node [style=gn] (59) at (3.75, 1.5) {};
		\node [style=none] (60) at (3.75, -1.45) {$\ldots$};
		\node [style=none] (61) at (6.2, -1.5) {};
		\node [style=none] (62) at (2.25, 1.75) {};
		\node [style=none] (63) at (2.55, -1.5) {};
		\node [style=none] (64) at (4.05, -1.5) {};
		\node [style=none] (65) at (1.95, -1.5) {};
		\node [style=none] (66) at (6.5, -1.45) {$\ldots$};
		\node [style=none] (67) at (3, 1.5) {$\cdots$};
		\node [style=none] (68) at (6.5, 1.75) {};
		\node [style=gn,label={[gphase]right:$\gamma_m'$}] (69) at (3.75, -0) {};
		\node [style=gn] (70) at (3.75, -1) {};
		\node [style=none] (71) at (4.75, 1.75) {};
		\node [style=none] (72) at (5.625, -1.25) {$\ldots$};
		\node [style=gn] (73) at (4.75, 1.5) {};
		\node [style=none] (74) at (5.05, -1.5) {};
		\node [style=none] (75) at (3, -1.25) {$\ldots$};
		\node [style=none] (76) at (4.75, -1.45) {$\ldots$};
		\node [style=none] (77) at (3.75, 1.75) {};
		\node [style=gn] (78) at (1, 0.5) {};
		\node [style=gn] (79) at (6.5, 1.5) {};
		\node [style=gn] (80) at (4.75, -1) {};
		\node [style=gn] (81) at (2.25, -1) {};
		\node [style=none] (82) at (4.45, -1.5) {};
		\node [style=rn,label={[rphase]left:$-\pi/2$}] (83) at (1, 1.5) {};
		\node [style=none] (84) at (-4.05, -0) {$\cdots$};
		\node [style=none] (85) at (-1.3, -0) {$\cdots$};
		\node [style=none] (86) at (6.2, -0) {$\cdots$};
		\node [style=none] (87) at (3.45, -0) {$\cdots$};
		\node [style=rn,label={[rphase]left:$-\pi/2$}] (88) at (-6.25, 1) {};
		\node [style=none] (89) at (-7.25, 0) {$=$};
		\node [style=Hadamard] (90) at (1.625, 1) {};
		\node [style=Hadamard] (91) at (2.75, 1.2) {};
		\node [style=Hadamard] (92) at (4.25, 1.375) {};
		\node [style=Hadamard] (93) at (5.25, 1.25) {};
	\end{pgfonlayer}
	\begin{pgfonlayer}{edgelayer}
		\draw (0) to (1);
		\draw (0) to (2);
		\draw (6.center) to (0);
		\draw (7.center) to (1);
		\draw (8.center) to (2);
		\draw [bend right=15] (1) to (11.center);
		\draw [bend left=15] (1) to (13.center);
		\draw [bend right=15] (2) to (14.center);
		\draw [bend left=15] (2) to (16.center);
		\draw (17) to (18);
		\draw (17) to (19);
		\draw (39.center) to (30);
		\draw (36.center) to (37);
		\draw [bend right=15] (30) to (33.center);
		\draw [bend left=15] (30) to (32.center);
		\draw [bend right=15] (37) to (28.center);
		\draw [bend left=15] (37) to (29.center);
		\draw (17) to (40);
		\draw (17) to (26);
		\draw [bend right=15] (0) to (22.center);
		\draw [bend left=15] (0) to (24.center);
		\draw (54) to (81);
		\draw (54) to (70);
		\draw (46.center) to (54);
		\draw (62.center) to (81);
		\draw (77.center) to (70);
		\draw [bend right=15] (81) to (65.center);
		\draw [bend left=15] (81) to (63.center);
		\draw [bend right=15] (70) to (55.center);
		\draw [bend left=15] (70) to (64.center);
		\draw (78) to (58);
		\draw (78) to (59);
		\draw (71.center) to (80);
		\draw (68.center) to (52);
		\draw [bend right=15] (80) to (82.center);
		\draw [bend left=15] (80) to (74.center);
		\draw [bend right=15] (52) to (61.center);
		\draw [bend left=15] (52) to (43.center);
		\draw (78) to (73);
		\draw (78) to (79);
		\draw [bend right=15] (54) to (49.center);
		\draw [bend left=15] (54) to (53.center);
	\end{pgfonlayer}
  \end{tikzpicture}
 \end{center}
 \begin{center}
  \begin{tikzpicture}
	\begin{pgfonlayer}{nodelayer}
		\node [style=gn] (0) at (-6.25, -0) {};
		\node [style=gn] (1) at (-5.25, -1) {};
		\node [style=gn] (2) at (-3.75, -1) {};
		\node [style=Hadamard] (3) at (-6.25, 0.5) {};
		\node [style=gn,label={[gphase]right:$\gamma_1'$}] (4) at (-5.25, 0.5) {};
		\node [style=gn,label={[gphase]right:$\gamma_m'$}] (5) at (-3.75, 0.5) {};
		\node [style=none] (6) at (-6.25, 1.25) {};
		\node [style=none] (7) at (-5.25, 1.25) {};
		\node [style=none] (8) at (-3.75, 1.25) {};
		\node [style=Hadamard] (9) at (-4.5, -0.5) {};
		\node [style=Hadamard] (10) at (-2.25, -0.75) {};
		\node [style=none] (11) at (-5.55, -1.5) {};
		\node [style=none] (12) at (-5.25, -1.45) {$\ldots$};
		\node [style=none] (13) at (-4.95, -1.5) {};
		\node [style=none] (14) at (-4.05, -1.5) {};
		\node [style=none] (15) at (-3.75, -1.45) {$\ldots$};
		\node [style=none] (16) at (-3.45, -1.5) {};
		\node [style=none] (20) at (-4.5, -1.25) {$\ldots$};
		\node [style=none] (22) at (-6.75, 0.25) {};
		\node [style=none] (23) at (-6.7, 0.1) {$\vdots$};
		\node [style=none] (24) at (-6.75, -0.25) {};
		\node [style=none] (25) at (-2.75, -1.45) {$\ldots$};
		\node [style=none] (27) at (-1.875, -1.25) {$\ldots$};
		\node [style=none] (28) at (-1.3, -1.5) {};
		\node [style=none] (29) at (-0.7, -1.5) {};
		\node [style=gn] (30) at (-2.75, -1) {};
		\node [style=gn,label={[gphase]right:$\gamma_{m+1}$}] (31) at (-2.75, 0.5) {};
		\node [style=none] (32) at (-2.45, -1.5) {};
		\node [style=none] (33) at (-3.05, -1.5) {};
		\node [style=gn,label={[gphase]right:$\gamma_n$}] (35) at (-1, 0.5) {};
		\node [style=none] (36) at (-1, 1.25) {};
		\node [style=gn] (37) at (-1, -1) {};
		\node [style=none] (38) at (-1, -1.45) {$\ldots$};
		\node [style=none] (39) at (-2.75, 1.25) {};
		\node [style=none] (84) at (-4.05, 0.5) {$\cdots$};
		\node [style=none] (85) at (-1.3, 0.5) {$\cdots$};
		\node [style=rn,label={[rphase]left:$-\pi/2$}] (88) at (-6.25, 1) {};
		\node [style=none] (89) at (-7.25, 0) {$=$};
	\end{pgfonlayer}
	\begin{pgfonlayer}{edgelayer}
		\draw (0) to (30);
		\draw (0) to (37);
		\draw (6.center) to (0);
		\draw (7.center) to (1);
		\draw (8.center) to (2);
		\draw [bend right=15] (1) to (11.center);
		\draw [bend left=15] (1) to (13.center);
		\draw [bend right=15] (2) to (14.center);
		\draw [bend left=15] (2) to (16.center);
		\draw (39.center) to (30);
		\draw (36.center) to (37);
		\draw [bend right=15] (30) to (33.center);
		\draw [bend left=15] (30) to (32.center);
		\draw [bend right=15] (37) to (28.center);
		\draw [bend left=15] (37) to (29.center);
		\draw [bend right=15] (0) to (22.center);
		\draw [bend left=15] (0) to (24.center);
	\end{pgfonlayer}
  \end{tikzpicture}
 \end{center}
 where $\gamma_k'=\gamma_k-\pi/2$ for $k=1,\ldots,m$. Now if $N_1=N_2$, $p$ has no more neighbours and is in the state \state{rn,label={[rphase]right:$-\pi/2$}}. This is not the same as the state $p$ has in diagram 1, so the diagrams are not equal. Else, after the application of $U$, $p$ still has some neighbours in diagram 2. Local complementations do not change this fact. Thus the two diagrams cannot be equal. The case $\beta=-\pi/2$ is entirely analogous, except that there is no fixpoint operation at the beginning.

 If $\beta=0$, there are two subcases. First, suppose there exists $v\in N_2$ such that $v\notin N_1$. Apply a local complementation about this $v$. This operation changes the vertex operator on $p$ to \phase{gn,label={[gphase]right:$\pi/2$}}. It also changes the edges incident on $p$, but the important thing is that $p$ will still have at least one neighbour. Thus we can proceed as in the case $\beta=\pi/2$. Secondly, suppose there is no $v\in N_2$ which is not in $N_1$. Since $N_2\neq\emptyset$ ($N_2=\emptyset$ corresponds to the case ``p has no neighbours in $D_2$''), we must then be able to find $v\in N_1\cap N_2$. The diagram looks as follows, where now $m>0$ (again, we are ignoring edges that do not involve $p$):
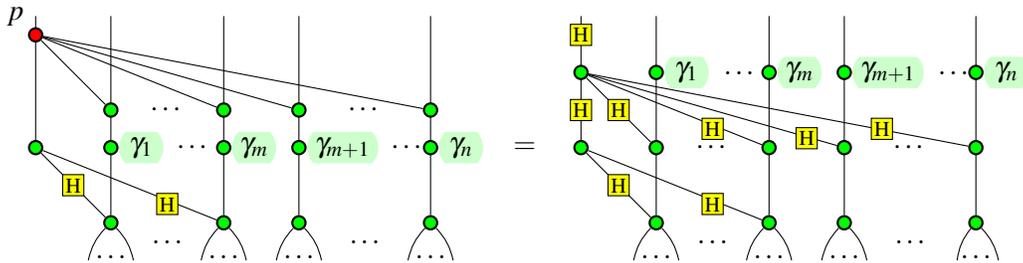
 \begin{center}
  \begin{tikzpicture}
	\begin{pgfonlayer}{nodelayer}
		\node [style=gn] (0) at (-6.25, -0) {};
		\node [style=gn] (1) at (-5.25, -1) {};
		\node [style=gn] (2) at (-3.75, -1) {};
		\node [style=gn,label={[gphase]right:$\gamma_1$}] (3) at (-5.25, -0) {};
		\node [style=gn,label={[gphase]right:$\gamma_m$}] (4) at (-3.75, -0) {};
		\node [style=none,label={left:$p$}] (5) at (-6.25, 1.75) {};
		\node [style=none] (6) at (-5.25, 1.75) {};
		\node [style=none] (7) at (-3.75, 1.75) {};
		\node [style=Hadamard] (8) at (-5.75, -0.5) {};
		\node [style=Hadamard] (9) at (-4.5, -0.75) {};
		\node [style=none] (10) at (-5.55, -1.5) {};
		\node [style=none] (11) at (-5.25, -1.45) {$\ldots$};
		\node [style=none] (12) at (-4.95, -1.5) {};
		\node [style=none] (13) at (-4.05, -1.5) {};
		\node [style=none] (14) at (-3.75, -1.45) {$\ldots$};
		\node [style=none] (15) at (-3.45, -1.5) {};
		\node [style=rn] (16) at (-6.25, 1.5) {};
		\node [style=gn] (17) at (-5.25, 0.5) {};
		\node [style=gn] (18) at (-3.75, 0.5) {};
		\node [style=none] (19) at (-4.5, -1.25) {$\ldots$};
		\node [style=none] (20) at (-4.5, 0.5) {$\cdots$};
		\node [style=none] (21) at (-2.75, -1.45) {$\ldots$};
		\node [style=gn] (22) at (-1, 0.5) {};
		\node [style=none] (23) at (-1.875, -1.25) {$\ldots$};
		\node [style=none] (24) at (-1.3, -1.5) {};
		\node [style=none] (25) at (-0.7, -1.5) {};
		\node [style=gn] (26) at (-2.75, -1) {};
		\node [style=gn,label={[gphase]right:$\gamma_{m+1}$}] (27) at (-2.75, -0) {};
		\node [style=none] (28) at (-2.45, -1.5) {};
		\node [style=none] (29) at (-3.05, -1.5) {};
		\node [style=none] (30) at (-1.875, 0.5) {$\cdots$};
		\node [style=gn,label={[gphase]right:$\gamma_n$}] (31) at (-1, -0) {};
		\node [style=none] (32) at (-1, 1.75) {};
		\node [style=gn] (33) at (-1, -1) {};
		\node [style=none] (34) at (-1, -1.45) {$\ldots$};
		\node [style=none] (35) at (-2.75, 1.75) {};
		\node [style=gn] (36) at (-2.75, 0.5) {};
		\node [style=none] (37) at (0.25, -0) {$=$};
		\node [style=Hadamard] (38) at (1, 1.5) {};
		\node [style=none] (39) at (6.55, -1.5) {};
		\node [style=gn,label={[gphase]right:$\gamma_1$}] (40) at (2, 1) {};
		\node [style=Hadamard] (41) at (2.75, -0.75) {};
		\node [style=none] (42) at (1, 1.75) {};
		\node [style=gn,label={[gphase]right:$\gamma_{m+1}$}] (43) at (4.5, 1) {};
		\node [style=Hadamard] (44) at (1.5, -0.5) {};
		\node [style=none] (45) at (2, -1.45) {$\ldots$};
		\node [style=gn] (46) at (6.25, -1) {};
		\node [style=gn] (47) at (1, -0) {};
		\node [style=none] (48) at (3.2, -1.5) {};
		\node [style=gn,label={[gphase]right:$\gamma_n$}] (49) at (6.25, 1) {};
		\node [style=none] (50) at (5.375, -0) {$\cdots$};
		\node [style=gn] (51) at (2, -0) {};
		\node [style=gn] (52) at (3.5, -0) {};
		\node [style=none] (53) at (3.5, -1.45) {$\ldots$};
		\node [style=none] (54) at (6, -1.5) {};
		\node [style=none] (55) at (2, 1.75) {};
		\node [style=none] (56) at (2.3, -1.5) {};
		\node [style=none] (57) at (3.8, -1.5) {};
		\node [style=none] (58) at (1.7, -1.5) {};
		\node [style=none] (59) at (6.25, -1.45) {$\ldots$};
		\node [style=none] (60) at (2.75, -0) {$\cdots$};
		\node [style=none] (61) at (6.25, 1.75) {};
		\node [style=gn,label={[gphase]right:$\gamma_m$}] (62) at (3.5, 1) {};
		\node [style=gn] (63) at (3.5, -1) {};
		\node [style=none] (64) at (4.5, 1.75) {};
		\node [style=none] (65) at (5.375, -1.25) {$\ldots$};
		\node [style=gn] (66) at (4.5, -0) {};
		\node [style=none] (67) at (4.8, -1.5) {};
		\node [style=none] (68) at (2.75, -1.25) {$\ldots$};
		\node [style=none] (69) at (4.5, -1.45) {$\ldots$};
		\node [style=none] (70) at (3.5, 1.75) {};
		\node [style=gn] (71) at (1, 1) {};
		\node [style=gn] (72) at (6.25, -0) {};
		\node [style=gn] (73) at (4.5, -1) {};
		\node [style=gn] (74) at (2, -1) {};
		\node [style=none] (75) at (4.2, -1.5) {};
		\node [style=none] (76) at (-4.15, -0) {$\cdots$};
		\node [style=none] (77) at (-1.3, -0) {$\cdots$};
		\node [style=none] (78) at (5.95, 1) {$\ldots$};
		\node [style=none] (79) at (3.1, 1) {$\ldots$};
		\node [style=Hadamard] (80) at (1, 0.5) {};
		\node [style=Hadamard] (81) at (1.5, 0.5) {};
		\node [style=Hadamard] (82) at (2.75, 0.25) {};
		\node [style=Hadamard] (83) at (4, 0.125) {};
		\node [style=Hadamard] (84) at (5, 0.25) {};
	\end{pgfonlayer}
	\begin{pgfonlayer}{edgelayer}
		\draw (0) to (1);
		\draw (0) to (2);
		\draw (5.center) to (0);
		\draw (6.center) to (1);
		\draw (7.center) to (2);
		\draw [bend right=15] (1) to (10.center);
		\draw [bend left=15] (1) to (12.center);
		\draw [bend right=15] (2) to (13.center);
		\draw [bend left=15] (2) to (15.center);
		\draw (16) to (17);
		\draw (16) to (18);
		\draw (35.center) to (26);
		\draw (32.center) to (33);
		\draw [bend right=15] (26) to (29.center);
		\draw [bend left=15] (26) to (28.center);
		\draw [bend right=15] (33) to (24.center);
		\draw [bend left=15] (33) to (25.center);
		\draw (16) to (36);
		\draw (16) to (22);
		\draw (47) to (74);
		\draw (47) to (63);
		\draw (42.center) to (47);
		\draw (55.center) to (74);
		\draw (70.center) to (63);
		\draw [bend right=15] (74) to (58.center);
		\draw [bend left=15] (74) to (56.center);
		\draw [bend right=15] (63) to (48.center);
		\draw [bend left=15] (63) to (57.center);
		\draw (71) to (51);
		\draw (71) to (52);
		\draw (64.center) to (73);
		\draw (61.center) to (46);
		\draw [bend right=15] (73) to (75.center);
		\draw [bend left=15] (73) to (67.center);
		\draw [bend right=15] (46) to (54.center);
		\draw [bend left=15] (46) to (39.center);
		\draw (71) to (66);
		\draw (71) to (72);
	\end{pgfonlayer}
  \end{tikzpicture}
 \end{center}
 To show that the two diagrams are unequal it suffices to show that in diagram 2 the state of $p$ either factors out, but is not \state{rn} or \state{rn,label={[rphase]right:$\pi$}}, or that it remains entangled with other qubits. We are thus justified in ignoring large portions of the above diagram to focus only on $p$, $v$ and the controlled-Z between the two. In particular, we will ignore for the moment the controlled-Z gates between $p$ and qubits other than $v$, as well as the last Hadamard gate on $p$. Then
 
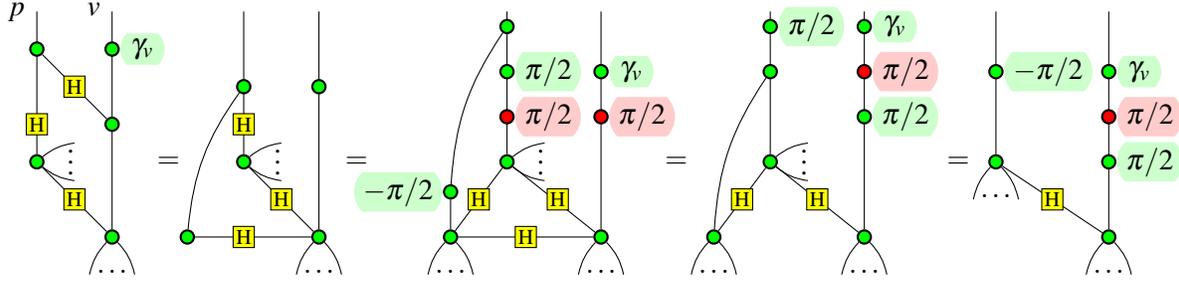
\begin{center}
  \begin{tikzpicture}
	\begin{pgfonlayer}{nodelayer}
		\node [style=gn] (0) at (-5, -0) {};
		\node [style=gn] (1) at (-4, -1) {};
		\node [style=gn] (2) at (-4, 1) {};
		\node [style=none] (3) at (-5, 2) {};
		\node [style=none] (4) at (-4, 2) {};
		\node [style=Hadamard] (5) at (-4.5, -0.5) {};
		\node [style=none] (6) at (-4.3, -1.5) {};
		\node [style=none] (7) at (-4, -1.45) {$\ldots$};
		\node [style=none] (8) at (-3.7, -1.5) {};
		\node [style=gn] (9) at (-5, 1) {};
		\node [style=gn] (10) at (-5.75, -1) {};
		\node [style=none] (11) at (-3.5, -0) {$=$};
		\node [style=gn,label={[gphase]right:$\gamma_v$}] (12) at (-0.25, 1.2) {};
		\node [style=none] (13) at (-1.5, 2) {};
		\node [style=Hadamard] (14) at (-0.825, -0.5) {};
		\node [style=none] (15) at (-0.25, -1.45) {$\ldots$};
		\node [style=gn] (16) at (-1.5, -0) {};
		\node [style=gn] (17) at (-2.25, -1) {};
		\node [style=none] (18) at (-0.25, 2) {};
		\node [style=none] (19) at (0.05, -1.5) {};
		\node [style=none] (20) at (-0.55, -1.5) {};
		\node [style=gn] (21) at (-1.5, 1.8) {};
		\node [style=gn] (22) at (-0.25, -1) {};
		\node [style=Hadamard] (23) at (-1.25, -1) {};
		\node [style=Hadamard] (24) at (-5, 0.5) {};
		\node [style=Hadamard] (25) at (-5, -1) {};
		\node [style=rn,label={[rphase]right:$\pi/2$}] (26) at (-0.25, 0.6) {};
		\node [style=rn,label={[rphase]right:$\pi/2$}] (27) at (-1.5, 0.6) {};
		\node [style=gn,label={[gphase]right:$\pi/2$}] (28) at (-1.5, 1.2) {};
		\node [style=gn,label={[gphase]left:$-\pi/2$}] (29) at (-2.25, -0.4) {};
		\node [style=Hadamard] (30) at (-7.75, 0.5) {};
		\node [style=gn] (31) at (-6.75, -1) {};
		\node [style=none,label={left:$p$}] (32) at (-7.75, 2) {};
		\node [style=Hadamard] (33) at (-7.25, -0.5) {};
		\node [style=none] (34) at (-7.05, -1.5) {};
		\node [style=none,label={left:$\vphantom{p}v$}] (35) at (-6.75, 2) {};
		\node [style=gn] (36) at (-7.75, -0) {};
		\node [style=none] (37) at (-6.45, -1.5) {};
		\node [style=gn,label={[gphase]right:$\gamma_v$}] (38) at (-6.75, 1.5) {};
		\node [style=none] (39) at (-6.75, -1.45) {$\ldots$};
		\node [style=gn] (40) at (-7.75, 1.5) {};
		\node [style=Hadamard] (41) at (-7.25, 1) {};
		\node [style=gn] (42) at (-6.75, 0.5) {};
		\node [style=none] (43) at (-7.25, 0.25) {};
		\node [style=none] (44) at (-7.3, 0.1) {$\vdots$};
		\node [style=none] (45) at (-7.25, -0.25) {};
		\node [style=none] (46) at (-6, -0) {$=$};
		\node [style=none] (47) at (-4.5, 0.25) {};
		\node [style=none] (48) at (-4.55, 0.1) {$\vdots$};
		\node [style=none] (49) at (-4.5, -0.25) {};
		\node [style=none] (50) at (-1, 0.25) {};
		\node [style=none] (51) at (-1.05, 0.1) {$\vdots$};
		\node [style=none] (52) at (-1, -0.25) {};
		\node [style=Hadamard] (53) at (-1.875, -0.5) {};
		\node [style=none] (54) at (0.75, -0) {$=$};
		\node [style=gn] (55) at (2, -0) {};
		\node [style=none] (56) at (2.95, -1.5) {};
		\node [style=none] (57) at (2.45, 0.1) {$\vdots$};
		\node [style=none] (58) at (2.5, 0.25) {};
		\node [style=none] (59) at (3.25, -1.45) {$\ldots$};
		\node [style=rn,label={[rphase]right:$\pi/2$}] (60) at (3.25, 1.2) {};
		\node [style=Hadamard] (61) at (1.625, -0.5) {};
		\node [style=gn] (62) at (2, 1.2) {};
		\node [style=none] (63) at (2.5, -0.25) {};
		\node [style=gn,label={[gphase]right:$\pi/2$}] (64) at (2, 1.8) {};
		\node [style=gn] (65) at (1.25, -1) {};
		\node [style=Hadamard] (66) at (2.625, -0.5) {};
		\node [style=none] (67) at (3.25, 2) {};
		\node [style=none] (68) at (2, 2) {};
		\node [style=gn,label={[gphase]right:$\gamma_v$}] (69) at (3.25, 1.8) {};
		\node [style=gn] (70) at (3.25, -1) {};
		\node [style=none] (71) at (3.55, -1.5) {};
		\node [style=gn,label={[gphase]right:$\pi/2$}] (72) at (3.25, 0.6) {};
		\node [style=none] (73) at (4.5, -0) {$=$};
		\node [style=rn,label={[rphase]right:$\pi/2$}] (74) at (6.5, 0.6) {};
		\node [style=gn,label={[gphase]right:$-\pi/2$}] (75) at (5, 1.2) {};
		\node [style=gn,label={[gphase]right:$\pi/2$}] (76) at (6.5, 0) {};
		\node [style=none] (77) at (5.3, -0.5) {};
		\node [style=Hadamard] (78) at (5.75, -0.5) {};
		\node [style=none] (79) at (6.5, -1.45) {$\ldots$};
		\node [style=gn,label={[gphase]right:$\gamma_v$}] (80) at (6.5, 1.2) {};
		\node [style=none] (81) at (4.7, -0.5) {};
		\node [style=none] (82) at (5, 2) {};
		\node [style=gn] (84) at (6.5, -1) {};
		\node [style=none] (85) at (6.5, 2) {};
		\node [style=none] (86) at (6.2, -1.5) {};
		\node [style=none] (87) at (5, -0.45) {$\ldots$};
		\node [style=none] (88) at (6.8, -1.5) {};
		\node [style=gn] (89) at (5, -0) {};
		\node [style=none] (90) at (-1.95, -1.5) {};
		\node [style=none] (91) at (-2.55, -1.5) {};
		\node [style=none] (92) at (1.55, -1.5) {};
		\node [style=none] (93) at (0.95, -1.5) {};
		\node [style=none] (94) at (-2.25, -1.45) {$\ldots$};
		\node [style=none] (95) at (1.25, -1.45) {$\ldots$};
	\end{pgfonlayer}
	\begin{pgfonlayer}{edgelayer}
		\draw (0) to (1);
		\draw (3.center) to (0);
		\draw (4.center) to (1);
		\draw [bend right=15] (1) to (6.center);
		\draw [bend left=15] (1) to (8.center);
		\draw [bend right=15] (9) to (10);
		\draw (16) to (22);
		\draw (13.center) to (16);
		\draw (18.center) to (22);
		\draw [bend right=15] (22) to (20.center);
		\draw [bend left=15] (22) to (19.center);
		\draw (10) to (1);
		\draw (17) to (22);
		\draw (17) to (29);
		\draw [bend left=15] (29) to (21);
		\draw (36) to (31);
		\draw (32.center) to (36);
		\draw (35.center) to (31);
		\draw [bend right=15] (31) to (34.center);
		\draw [bend left=15] (31) to (37.center);
		\draw (40) to (42);
		\draw [bend left=15] (36) to (43.center);
		\draw [bend right=15] (36) to (45.center);
		\draw [bend left=15] (0) to (47.center);
		\draw [bend right=15] (0) to (49.center);
		\draw [bend left=15] (16) to (50.center);
		\draw [bend right=15] (16) to (52.center);
		\draw (16) to (17);
		\draw (55) to (70);
		\draw (68.center) to (55);
		\draw (67.center) to (70);
		\draw [bend right=15] (70) to (56.center);
		\draw [bend left=15] (70) to (71.center);
		\draw [bend left=15] (55) to (58.center);
		\draw [bend right=15] (55) to (63.center);
		\draw (55) to (65);
		\draw [bend right=15] (62) to (65);
		\draw (89) to (84);
		\draw (82.center) to (89);
		\draw (85.center) to (84);
		\draw [bend right=15] (84) to (86.center);
		\draw [bend left=15] (84) to (88.center);
		\draw [bend left=15] (89) to (77.center);
		\draw [bend right=15] (89) to (81.center);
		\draw [bend right=15] (17) to (91.center);
		\draw [bend left=15] (17) to (90.center);
		\draw [bend right=15] (65) to (93.center);
		\draw [bend left=15] (65) to (92.center);
	\end{pgfonlayer}
  \end{tikzpicture}
 \end{center}
 where for the second equality we have applied a local complementation and a fixpoint operation to $v$ and used the Euler decomposition, the third equality follows by a local complementation on $p$, and the last one comes from the merging of $p$ with the green node in the bottom left. Note that, in the end, $p$ and $v$ are still connected by an edge. None of the operations we ignored in picking out this part of the diagram will change that. Thus, as before, the state of $p$ cannot be the same as in diagram 1. The two diagrams are unequal.

 The case $\beta=\pi$ is analogous to $\beta=0$, except we start with a fixpoint operation on the same qubit as the first local complementation.

 We have shown that a simplified pair of rGS-LC diagrams are not equal if there are any unpaired red nodes.
\end{proof}

\begin{thm}\label{thm:rGS-LC_equality}
 The two diagrams making up a simplified pair of rGS-LC diagram are equal, i.e.\ they correspond to the same quantum state, if and only if they are identical.
\end{thm}
\begin{proof}
 By lemma \ref{lem:unpaired_red_node}, the diagrams are unequal if there are any unpaired red nodes. We can therefore assume that all red nodes are paired up.

 Let the diagrams be $D_1$ and $D_2$. Suppose the graph underlying $D_1$ is $G_1=(V,E_1)$ and that underlying $D_2$ is $G_2=(V,E_2)$. For simplicity, suppose $V=\{1,2,\ldots,n\}$. We can draw the two diagrams as
 
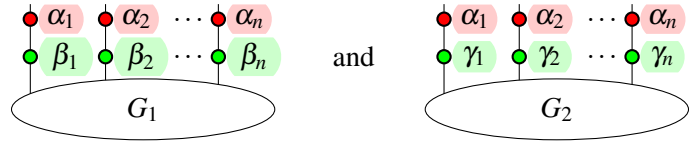
\begin{center}
  \begin{tikzpicture}
	\begin{pgfonlayer}{nodelayer}
		\node [style=none] (0) at (-3, 0.7) {};
		\node [style=gn,label={[gphase]right:$\beta_2$}] (1) at (-3, -0) {};
		\node [style=none] (2) at (-1.9, -0) {$\ldots$};
		\node [style=gn,label={[gphase]right:$\beta_1$}] (3) at (-4, -0) {};
		\node [style=gn,label={[gphase]right:$\beta_n$}] (4) at (-1.5, -0) {};
		\node [style=rn,label={[rphase]right:$\alpha_2$}] (5) at (-3, 0.5) {};
		\node [style=none] (6) at (-4, -0.75) {};
		\node [style=none] (7) at (-1.9, 0.5) {$\ldots$};
		\node [style=rn,label={[rphase]right:$\alpha_1$}] (8) at (-4, 0.5) {};
		\node [style=none] (9) at (-1.5, 0.7) {};
		\node [ellipse,fill=White,draw=Black,minimum width=3.5cm] (10) at (-2.5, -0.7) {$G_1$};
		\node [style=none] (11) at (-3, -0.75) {};
		\node [style=none] (12) at (-1.5, -0.75) {};
		\node [style=rn,label={[rphase]right:$\alpha_n$}] (13) at (-1.5, 0.5) {};
		\node [style=none] (14) at (-4, 0.7) {};
		\node [style=none] (15) at (2.5, 0.7) {};
		\node [style=gn,label={[gphase]right:$\gamma_2$}] (16) at (2.5, -0) {};
		\node [style=gn,label={[gphase]right:$\gamma_1$}] (17) at (1.5, -0) {};
		\node [style=none] (18) at (3.6, -0) {$\ldots$};
		\node [style=gn,label={[gphase]right:$\gamma_n$}] (19) at (4, -0) {};
		\node [style=rn,label={[rphase]right:$\alpha_2$}] (20) at (2.5, 0.5) {};
		\node [style=none] (21) at (1.5, -0.75) {};
		\node [style=none] (22) at (3.6, 0.5) {$\ldots$};
		\node [style=rn,label={[rphase]right:$\alpha_1$}] (23) at (1.5, 0.5) {};
		\node [style=none] (24) at (4, 0.7) {};
		\node [style=none] (25) at (2.5, -0.75) {};
		\node [ellipse,fill=White,draw=Black,minimum width=3.5cm] (26) at (3, -0.7) {$G_2$};
		\node [style=none] (27) at (4, -0.75) {};
		\node [style=rn,label={[rphase]right:$\alpha_n$}] (28) at (4, 0.5) {};
		\node [style=none] (29) at (1.5, 0.7) {};
		\node [style=none] (30) at (0.3, -0) {and};
	\end{pgfonlayer}
	\begin{pgfonlayer}{edgelayer}
		\draw (14.center) to (6.center);
		\draw (0.center) to (11.center);
		\draw (9.center) to (12.center);
		\draw (29.center) to (21.center);
		\draw (15.center) to (25.center);
		\draw (24.center) to (27.center);
	\end{pgfonlayer}
  \end{tikzpicture}
 \end{center}
 where, for all $v\in V$, $\alpha_v\in\{0,\pi/2\}$ and
 \[
  \beta_v,\gamma_v\in \begin{cases} \{\pm\pi/2\} & \text{if }\alpha_v = \pi/2 \\ \{0,\pi/2,\pi,-\pi/2\} & \text{otherwise.} \end{cases}
 \]
 Let $V'=\{v\in V|\alpha_v=\pi/2\}$ and define the operators
 \[
  U = \bigotimes_{v\in V'} R_{X,v}^{-1} \quad\text{and}\quad W = \bigotimes_{\{u,w\}\in E_1} \wedge Z_{uw},
 \]
 where $\wedge Z_{uw}$ denotes a controlled-Z operator applied to qubits $u$ and $w$. Controlled-Z operators commute, therefore both $U$ and $W$ are invertible and we have $(W\circ U)\circ D_1=(W\circ U)\circ D_2$ if and only if $D_1=D_2$. Now in $U\circ D_1$ and $U\circ D_2$, all vertex operators are green nodes, which can be moved past controlled-Z operators. Thus $(W\circ U)\circ D_1$ is equal to
 \begin{tikzpicture}[baseline=-.1cm]
	\begin{pgfonlayer}{nodelayer}
		\node [style=gn,label={[gphase]right:$\beta_1$}] (0) at (0, -0) {};
		\node [style=none] (1) at (0, 0.2) {};
		\node [style=gn,label={[gphase]right:$\beta_n$}] (2) at (1.5, -0) {};
		\node [style=none] (3) at (1.5, 0.2) {};
		\node [style=none] (4) at (1.1, 0) {$\cdots$};
	\end{pgfonlayer}
	\begin{pgfonlayer}{edgelayer}
		\draw (0) to (1.center);
		\draw (2) to (3.center);
	\end{pgfonlayer}
 \end{tikzpicture}.
 Now $(W\circ U)\circ D_2$ can be rewritten as follows
 
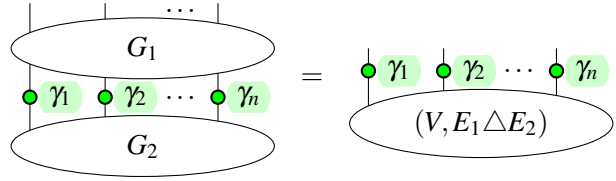
\begin{center}
  \begin{tikzpicture}
	\begin{pgfonlayer}{nodelayer}
		\node [style=none] (0) at (-2.75, 1.25) {};
		\node [style=gn,label={[gphase]right:$\gamma_2$}] (1) at (-2.75, -0) {};
		\node [style=none] (2) at (-1.75, -0) {$\ldots$};
		\node [style=gn,label={[gphase]right:$\gamma_1$}] (3) at (-3.75, -0) {};
		\node [style=gn,label={[gphase]right:$\gamma_n$}] (4) at (-1.25, -0) {};
		\node [style=none] (5) at (-3.75, -0.75) {};
		\node [style=none] (6) at (-1.25, 1.25) {};
		\node [ellipse,fill=White,draw=Black,minimum width=3.5cm] (7) at (-2.25, -0.65) {$G_2$};
		\node [style=none] (8) at (-2.75, -0.75) {};
		\node [style=none] (9) at (-1.25, -0.75) {};
		\node [style=none] (10) at (-3.75, 1.25) {};
		\node [style=none] (11) at (1.75, 0.65) {};
		\node [style=gn,label={[gphase]right:$\gamma_2$}] (12) at (1.75, 0.35) {};
		\node [style=gn,label={[gphase]right:$\gamma_1$}] (13) at (0.75, 0.35) {};
		\node [style=none] (14) at (2.75, 0.35) {$\ldots$};
		\node [style=gn,label={[gphase]right:$\gamma_n$}] (15) at (3.25, 0.35) {};
		\node [style=none] (16) at (0.75, -0.5) {};
		\node [style=none] (17) at (3.25, 0.65) {};
		\node [style=none] (18) at (1.75, -0.5) {};
		\node [ellipse,fill=White,draw=Black,minimum width=3.5cm] (19) at (2.25, -0.35) {$(V,E_1\triangle E_2)$};
		\node [style=none] (20) at (3.25, -0.5) {};
		\node [style=none] (21) at (0.75, 0.65) {};
		\node [style=none] (22) at (0, 0.25) {$=$};
		\node [ellipse,fill=White,draw=Black,minimum width=3.5cm] (23) at (-2.25, 0.65) {$G_1$};
		\node [style=none] (24) at (-1.75, 1.15) {$\ldots$};
	\end{pgfonlayer}
	\begin{pgfonlayer}{edgelayer}
		\draw (10.center) to (5.center);
		\draw (0.center) to (8.center);
		\draw (6.center) to (9.center);
		\draw (21.center) to (16.center);
		\draw (11.center) to (18.center);
		\draw (17.center) to (20.center);
	\end{pgfonlayer}
  \end{tikzpicture}
 \end{center}
 Here, the white ellipse labelled $G_1$ denotes the graph state $G_1$ with an additional input for each vertex. $E_1\triangle E_2$ is the symmetric difference of the two sets $E_1$ and $E_2$, i.e.\ the graph $(V,E_1\triangle E_2)$ contains all edges which are contained in either $G_1$ or $G_2$, but not in both. Clearly this is equal to a product of single qubit states only if $E_1\triangle E_2=\emptyset$. That condition is satisfied if and only if $E_1=E_2$, i.e.\ $G_1=G_2$.

 Assuming that the underlying graphs are equal, we have $(W\circ U)\circ D_1=(W\circ U)\circ D_2$ if and only if $\beta_v=\gamma_v$ for all $v\in V$. Thus $(W\circ U)\circ D_1=(W\circ U)\circ D_2$ if and only if $D_1$ and $D_2$ are identical. By unitarity of $(W\circ U)$, this implies that the diagrams $D_1$ and $D_2$ are equal if and only if they are identical, as required.
\end{proof}

\subsection{Completeness for stabilizer states}

In section \ref{s:ZX_graph_states} we show that any stabilizer state diagram is equal to some rGS-LC diagram. By theorem \ref{thm:rGS-LC_equality}, two rGS-LC diagrams represent the same quantum state if and only if simplifying the pair leads to two identical diagrams. Any rewrite rules used to prove these two results are invertible. Therefore, given two stabilizer state diagrams representing the same state, there exists a sequence of rewrite steps obeying the rules of the \ZX-calculus, which transforms one diagram into the other. Thus:

\begin{thm}\label{thm:completeness_states}
 The \ZX-calculus is complete for stabilizer states.
\end{thm}

There are many stabilizer diagrams which have a non-zero number of inputs, so they are not states and the previous arguments do not apply to them. To extend our results to those diagrams, we make use of the quantum mechanical map-state duality as laid out in the following section.

\subsection{Map-state duality in the \ZX-calculus}
\label{s:Choi-Jamiolkowski}

Map-state duality, also known as the Choi-Jamio{\l}\-kowski isomorphism, relates quantum states and linear operators:

\begin{thm}[Map-state duality or Choi-Jamio{\l}kowski isomorphism]\label{thm:Choi-Jamiolkowski}
 For any pair of positive integers $n$ and $m$, there exists a bijection between the linear operators from $n$ to $m$ qubits and the states on $n+m$ qubits.
\end{thm}

In the \ZX-calculus, states are diagrams with no inputs. Therefore the Choi-Jamio\l\-kowski isomorphism as a transformation consists of just ``bending around'' the inputs of the operator so that they become outputs. This process can also be thought of as composing the operator with an apropriate entangled state. In the reverse direction, we bend around some of the outputs to become inputs, or alternatively compose the diagram with the apropriate effect.

The isomorphism implies that for any operator $A$ from $n$ to $m$ qubits and for any $n+m$-qubit state $B$,
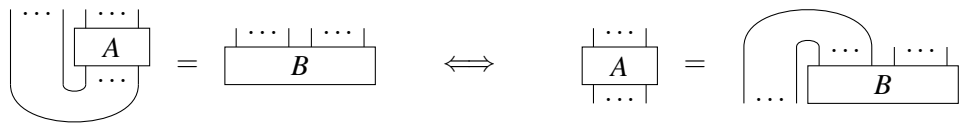
\begin{center}
 \begin{tikzpicture}
	\begin{pgfonlayer}{nodelayer}
		\node [rectangle,fill=White,draw=Black,minimum width=1cm,minimum height=0.5cm,inner sep=0pt] (0) at (-4.75, -0) {$A$};
		\node [style=none] (1) at (-5.1, 0.5) {};
		\node [style=none] (2) at (-4.75, 0.45) {$\ldots$};
		\node [style=none] (3) at (-4.4, 0.5) {};
		\node [style=none] (4) at (-5.1, -0.5) {};
		\node [style=none] (5) at (-4.75, -0.45) {$\ldots$};
		\node [style=none] (6) at (-4.4, -0.5) {};
		\node [style=none] (7) at (-5.4, -0.5) {};
		\node [style=none] (8) at (-6.1, -0.5) {};
		\node [style=none] (9) at (-5.4, 0.5) {};
		\node [style=none] (10) at (-6.1, 0.5) {};
		\node [style=none] (11) at (-5.75, 0.45) {$\ldots$};
		\node [style=none] (12) at (-3.75, -0.25) {$=$};
		\node [style=none] (13) at (-3.1, 0.25) {};
		\node [style=none] (14) at (-2.75, 0.2) {$\ldots$};
		\node [style=none] (15) at (-2.4, 0.25) {};
		\node [style=none] (16) at (-2.1, 0.25) {};
		\node [style=none] (17) at (-1.75, 0.2) {$\ldots$};
		\node [style=none] (18) at (-1.4, 0.25) {};
		\node [rectangle,fill=White,draw=Black,minimum width=2cm,minimum height=0.5cm,inner sep=0pt] (19) at (-2.25, -0.25) {$B$};
		\node [style=none] (20) at (-3.1, -0.25) {};
		\node [style=none] (21) at (-2.4, -0.25) {};
		\node [style=none] (22) at (-2.1, -0.25) {};
		\node [style=none] (23) at (-1.4, -0.25) {};
		\node [style=none] (24) at (0, -0.25) {$\Longleftrightarrow$};
		\node [rectangle,fill=White,draw=Black,minimum width=1cm,minimum height=0.5cm,inner sep=0pt] (25) at (2, -0.25) {$A$};
		\node [style=none] (26) at (1.65, 0.25) {};
		\node [style=none] (27) at (2, 0.2) {$\ldots$};
		\node [style=none] (28) at (2.35, 0.25) {};
		\node [style=none] (29) at (1.65, -0.75) {};
		\node [style=none] (30) at (2, -0.7) {$\ldots$};
		\node [style=none] (31) at (2.35, -0.75) {};
		\node [style=none] (32) at (3, -0.25) {$=$};
		\node [style=none] (33) at (3.65, 0) {};
		\node [style=none] (34) at (4.35, 0) {};
		\node [style=none] (35) at (3.65, -0.75) {};
		\node [style=none] (36) at (4.35, -0.75) {};
		\node [style=none] (37) at (4, -0.7) {$\ldots$};
		\node [style=none] (38) at (4.65, 0) {};
		\node [style=none] (39) at (5.35, 0) {};
		\node [style=none] (40) at (5, -0.05) {$\ldots$};
		\node [style=none] (41) at (5.65, 0) {};
		\node [style=none] (42) at (6, -0.05) {$\ldots$};
		\node [style=none] (43) at (6.35, 0) {};
		\node [rectangle,fill=White,draw=Black,minimum width=2cm,minimum height=0.5cm,inner sep=0pt] (44) at (5.5, -0.5) {$B$};
		\node [style=none] (45) at (4.65, -0.5) {};
		\node [style=none] (46) at (5.35, -0.5) {};
		\node [style=none] (47) at (5.65, -0.5) {};
		\node [style=none] (48) at (6.35, -0.5) {};
	\end{pgfonlayer}
	\begin{pgfonlayer}{edgelayer}
		\draw (3.center) to (6.center);
		\draw (1.center) to (4.center);
		\draw (9.center) to (7.center);
		\draw (10.center) to (8.center);
		\draw [bend right=90] (8.center) to (6.center);
		\draw [bend right=90] (7.center) to (4.center);
		\draw (13.center) to (20.center);
		\draw (15.center) to (21.center);
		\draw (16.center) to (22.center);
		\draw (18.center) to (23.center);
		\draw (33.center) to (35.center);
		\draw (34.center) to (36.center);
		\draw (26.center) to (29.center);
		\draw (28.center) to (31.center);
		\draw (38.center) to (45.center);
		\draw (39.center) to (46.center);
		\draw (41.center) to (47.center);
		\draw (43.center) to (48.center);
		\draw [bend left=90] (33.center) to (39.center);
		\draw [bend left=90] (34.center) to (38.center);
	\end{pgfonlayer}
 \end{tikzpicture}
\end{center}
This follows directly from the rule that \emph{only the topology matters}, which allows us to ``yank straight'' any inputs and outputs.

\subsection{Completeness for all stabilizer quantum mechanics}

We can now assemble the main completeness proof:

\begin{thm}\label{thm:completeness}
 The \ZX-calculus is complete for stabilizer quantum mechanics.
\end{thm}
\begin{proof}
 By theorem \ref{thm:completeness_states} we know that the \ZX-calculus is complete for stabilizer states. Now by theorem \ref{thm:Choi-Jamiolkowski}, two operators from $n$ to $m$ qubits are equal if and only if the corresponding $n+m$-qubit states are equal. Thus, given any two \ZX-calculus diagrams that represent the same operator, we can show that the diagrams are equal using the rules of the \ZX-calculus via the following sequence of steps:
 \begin{menum}
  \item Apply the Choi-Jamio{\l}kowski isomorphism to turn the operators into states.
  \item Transform the states into GS-LC and then rGS-LC form.
  \item Simplify the pair of rGS-LC diagrams.
  \item Apply the Choi-Jamio{\l}kowski isomorphism again to transform the sequence of equal states derived in the previous step back into operators.
 \end{menum}
 As the Choi-Jamio{\l}kowski isomorphism preserves equalities, this yields a sequence of steps which are valid according to the rules of the \ZX-calculus and which show that the two operators are equal. Thus, whenever two \ZX-calculus diagrams represent the same quantum mechanical state or operator, they are equal according to the rules of the \ZX-calculus, completing the proof.
\end{proof}

\section{Example}\label{s:example}

\subsection{Two \ZX-calulus diagrams for the controlled-Z operator}

In quantum circuit notation, there are two common ways of writing the controlled-Z gate in terms of controlled-NOT gates and different types of single qubit gates. The two quantum circuit diagrams translate straightforwardly to the following \ZX-calculus diagrams:
\begin{equation}\label{eq:example_diagrams}
 \begin{tikzpicture}[baseline=-0.1cm]
	\begin{pgfonlayer}{nodelayer}
		\node [style=gn] (0) at (-0.5, -0.75) {};
		\node [style=gn] (1) at (-0.5, 0.25) {};
		\node [style=gn,label={[gphase]left:$\pi/2$}] (2) at (-0.5, 0.75) {};
		\node [style=gn,label={[gphase]right:$-\pi/2$}] (3) at (0.5, -0.25) {};
		\node [style=gn,label={[gphase]right:$\pi/2$}] (4) at (0.5, 0.75) {};
		\node [style=rn] (5) at (0.5, 0.25) {};
		\node [style=rn] (6) at (0.5, -0.75) {};
		\node [style=none] (7) at (-0.5, 1) {};
		\node [style=none] (8) at (0.5, 1) {};
		\node [style=none] (9) at (-0.5, -1) {};
		\node [style=none] (10) at (0.5, -1) {};
	\end{pgfonlayer}
	\begin{pgfonlayer}{edgelayer}
		\draw (7.center) to (9.center);
		\draw (8.center) to (10.center);
		\draw (0) to (6);
		\draw (1) to (5);
	\end{pgfonlayer}
 \end{tikzpicture}
 \quad\text{and}\quad
 \begin{tikzpicture}[baseline=-0.1cm]
	\begin{pgfonlayer}{nodelayer}
		\node [style=gn] (0) at (-0.5, -0) {};
		\node [style=rn] (1) at (0.5, -0) {};
		\node [style=none] (2) at (-0.5, 0.8) {};
		\node [style=none] (3) at (0.5, 0.8) {};
		\node [style=none] (4) at (-0.5, -0.8) {};
		\node [style=none] (5) at (0.5, -0.8) {};
		\node [style=Hadamard] (6) at (0.5, 0.45) {};
		\node [style=Hadamard] (7) at (0.5, -0.45) {};
	\end{pgfonlayer}
	\begin{pgfonlayer}{edgelayer}
		\draw (2.center) to (4.center);
		\draw (3.center) to (5.center);
		\draw (0) to (1);
	\end{pgfonlayer}
 \end{tikzpicture}
\end{equation}
Since these two diagrams have been constructed to represent the same operator, we expect them to be equal. To confirm this, we use the algorithm given in theorem \ref{thm:completeness}.

\subsection{Applying the equality testing algorithm}

To bring the diagrams into GS-LC form, they first need to be mapped to the corresponding state diagrams via the Choi-Jamio{\l}kowski isomorphism. It is useful to note that
\begin{center}
 \begin{tikzpicture}
	\begin{pgfonlayer}{nodelayer}
		\node [style=none] (0) at (-3.5, 0.25) {};
		\node [style=none] (1) at (-3, -0.25) {};
		\node [style=none] (2) at (-2.5, 0.25) {};
		\node [style=none] (3) at (-1.75, -0) {$=$};
		\node [style=none] (4) at (-1, 0.5) {};
		\node [style=none] (5) at (0, 0.5) {};
		\node [style=none] (6) at (0.75, -0) {$=$};
		\node [style=none] (7) at (1.5, 0.5) {};
		\node [style=none] (8) at (2.5, 0.5) {};
		\node [style=gn] (9) at (-1, -0.25) {};
		\node [style=gn] (10) at (0, -0.25) {};
		\node [style=gn] (11) at (1.5, -0.25) {};
		\node [style=gn] (12) at (2.5, -0.25) {};
		\node [style=Hadamard] (13) at (-1, 0.2) {};
		\node [style=Hadamard] (14) at (-0.5, -0.25) {};
		\node [style=Hadamard] (15) at (2, -0.25) {};
		\node [style=Hadamard] (16) at (2.5, 0.2) {};
	\end{pgfonlayer}
	\begin{pgfonlayer}{edgelayer}
		\draw [bend right=45] (0.center) to (1.center);
		\draw [bend right=45] (1.center) to (2.center);
		\draw (4.center) to (9);
		\draw (5.center) to (10);
		\draw (7.center) to (11);
		\draw (8.center) to (12);
		\draw (9) to (10);
		\draw (11) to (12);
	\end{pgfonlayer}
 \end{tikzpicture}
\end{center}
and to convert the elements of the diagrams into those given in lemma \ref{lem:basic_elements} before transforming the diagram to a state. Thus, the first diagram becomes
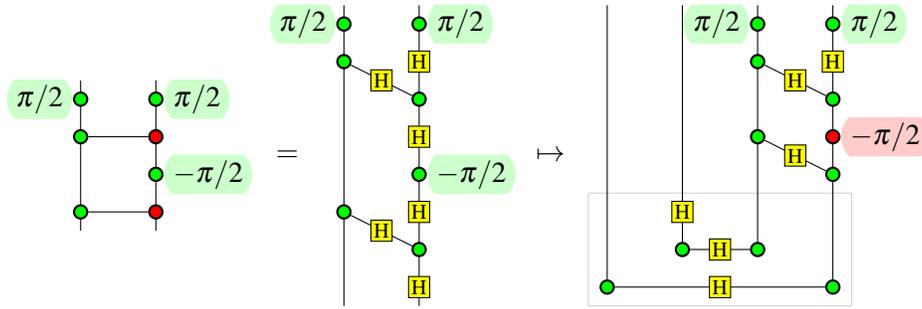
\begin{center}
 \begin{tikzpicture}
	\begin{pgfonlayer}{nodelayer}
		\node [style=gn] (0) at (-6, -0.75) {};
		\node [style=gn] (1) at (-6, 0.25) {};
		\node [style=gn,label={[gphase]left:$\pi/2$}] (2) at (-6, 0.75) {};
		\node [style=gn,label={[gphase]right:$-\pi/2$}] (3) at (-5, -0.25) {};
		\node [style=gn,label={[gphase]right:$\pi/2$}] (4) at (-5, 0.75) {};
		\node [style=rn] (5) at (-5, 0.25) {};
		\node [style=rn] (6) at (-5, -0.75) {};
		\node [style=none] (7) at (-6, 1) {};
		\node [style=none] (8) at (-5, 1) {};
		\node [style=none] (9) at (-6, -1) {};
		\node [style=none] (10) at (-5, -1) {};
		\node [style=none] (11) at (-3.25, -0) {$=$};
		\node [style=gn,label={[gphase]left:$\pi/2$}] (12) at (-2.5, 1.75) {};
		\node [style=gn] (13) at (-2.5, 1.25) {};
		\node [style=gn] (14) at (-2.5, -0.75) {};
		\node [style=gn] (15) at (-1.5, -1.25) {};
		\node [style=gn,label={[gphase]right:$-\pi/2$}] (16) at (-1.5, -0.25) {};
		\node [style=gn] (17) at (-1.5, 0.75) {};
		\node [style=gn,label={[gphase]right:$\pi/2$}] (18) at (-1.5, 1.75) {};
		\node [style=Hadamard] (19) at (-2, 1) {};
		\node [style=Hadamard] (20) at (-2, -1) {};
		\node [style=Hadamard] (21) at (-1.5, -0.75) {};
		\node [style=Hadamard] (22) at (-1.5, -1.75) {};
		\node [style=Hadamard] (23) at (-1.5, 0.25) {};
		\node [style=Hadamard] (24) at (-1.5, 1.25) {};
		\node [style=none] (25) at (-2.5, 2) {};
		\node [style=none] (26) at (-1.5, 2) {};
		\node [style=none] (27) at (-2.5, -2) {};
		\node [style=none] (28) at (-1.5, -2) {};
		\node [style=none] (29) at (0.25, 0) {$\mapsto$};
		\node [style=gn] (30) at (1, -1.75) {};
		\node [style=gn] (31) at (4, -1.75) {};
		\node [style=gn] (32) at (2, -1.25) {};
		\node [style=gn] (33) at (3, -1.25) {};
		\node [style=gn] (34) at (4, -0.25) {};
		\node [style=gn] (35) at (3, 0.25) {};
		\node [style=gn] (36) at (4, 0.75) {};
		\node [style=gn] (37) at (3, 1.25) {};
		\node [style=gn,label={[gphase]left:$\pi/2$}] (38) at (3, 1.75) {};
		\node [style=gn,label={[gphase]right:$\pi/2$}] (39) at (4, 1.75) {};
		\node [style=rn,label={[rphase]right:$-\pi/2$}] (40) at (4, 0.25) {};
		\node [style=Hadamard] (41) at (2.5, -1.25) {};
		\node [style=Hadamard] (42) at (2.5, -1.75) {};
		\node [style=Hadamard] (43) at (2, -0.75) {};
		\node [style=Hadamard] (44) at (3.5, -0) {};
		\node [style=Hadamard] (45) at (3.5, 1) {};
		\node [style=none] (46) at (1, 2) {};
		\node [style=none] (47) at (2, 2) {};
		\node [style=none] (48) at (3, 2) {};
		\node [style=none] (49) at (4, 2) {};
		\node [style=none] (50) at (0.75, -0.5) {};
		\node [style=none] (51) at (0.75, -2) {};
		\node [style=none] (52) at (4.25, -2) {};
		\node [style=none] (53) at (4.25, -0.5) {};
		\node [style=Hadamard] (54) at (4, 1.25) {};
	\end{pgfonlayer}
	\begin{pgfonlayer}{edgelayer}
		\draw [color=LightGrey] (50.center) to (51.center);
		\draw [color=LightGrey] (51.center) to (52.center);
		\draw [color=LightGrey] (52.center) to (53.center);
		\draw [color=LightGrey] (50.center) to (53.center);
		\draw (7.center) to (9.center);
		\draw (8.center) to (10.center);
		\draw (0) to (6);
		\draw (1) to (5);
		\draw (25.center) to (27.center);
		\draw (26.center) to (28.center);
		\draw (13) to (17);
		\draw (14) to (15);
		\draw (46.center) to (30);
		\draw (30) to (31);
		\draw (31) to (49.center);
		\draw (47.center) to (32);
		\draw (32) to (33);
		\draw (33) to (48.center);
		\draw (37) to (36);
		\draw (35) to (34);
	\end{pgfonlayer}
 \end{tikzpicture}
\end{center}
where the grey box in the last section of the equality encloses the GS-LC part of the diagram. The operators still outside the grey box are applied to the GS-LC state, one by one, using local complementations on the graph to change the vertex operators as necessary, until the whole diagram is in GS-LC form:
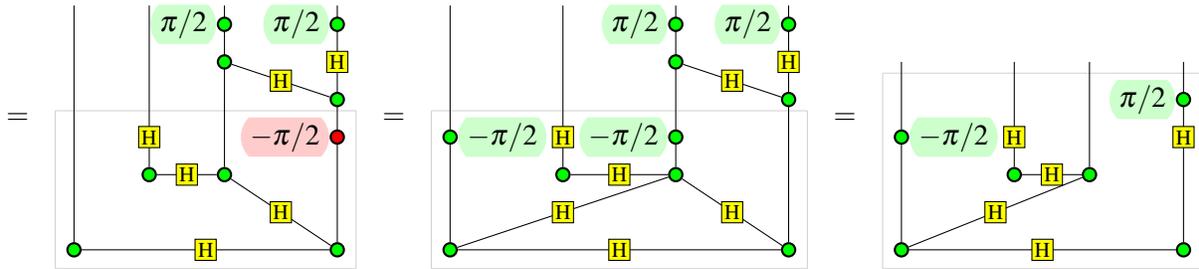
\begin{center}
 \begin{tikzpicture}
	\begin{pgfonlayer}{nodelayer}
		\node [style=none] (0) at (-3, -0) {$=$};
		\node [style=none] (1) at (-2.25, 1.5) {};
		\node [style=none] (2) at (-0.75, 1.5) {};
		\node [style=none] (3) at (0.75, 1.5) {};
		\node [style=gn] (4) at (-2.25, -1.75) {};
		\node [style=gn] (5) at (0.75, -0.75) {};
		\node [style=none] (6) at (2.25, 1.5) {};
		\node [style=Hadamard] (7) at (0, -0.75) {};
		\node [style=Hadamard] (8) at (-0, -1.75) {};
		\node [style=gn,label={[gphase]left:$\pi/2$}] (9) at (2.25, 1.25) {};
		\node [style=gn] (10) at (2.25, -1.75) {};
		\node [style=gn] (11) at (-0.75, -0.75) {};
		\node [style=Hadamard] (12) at (-0.75, -1.25) {};
		\node [style=gn,label={[gphase]right:$-\pi/2$}] (13) at (-2.25, -0.25) {};
		\node [style=Hadamard] (14) at (-0.75, -0.25) {};
		\node [style=Hadamard] (15) at (2.25, 0.75) {};
		\node [style=none] (16) at (2.5, -2) {};
		\node [style=none] (17) at (2.5, 0.1) {};
		\node [style=none] (18) at (-2.5, 0.1) {};
		\node [style=none] (19) at (-2.5, -2) {};
		\node [style=gn] (20) at (2.25, 0.25) {};
		\node [style=Hadamard] (21) at (1.5, 0.5) {};
		\node [style=Hadamard] (22) at (1.5, -1.25) {};
		\node [style=gn] (23) at (0.75, 0.75) {};
		\node [style=gn,label={[gphase]left:$\pi/2$}] (24) at (0.75, 1.25) {};
		\node [style=gn,label={[gphase]left:$-\pi/2$}] (25) at (0.75, -0.25) {};
		\node [style=gn,label={[gphase]left:$\pi/2$}] (26) at (-5.25, 1.25) {};
		\node [style=none] (27) at (-7.25, 1.5) {};
		\node [style=none] (28) at (-6.25, 1.5) {};
		\node [style=none] (29) at (-5.25, 1.5) {};
		\node [style=gn] (30) at (-7.25, -1.75) {};
		\node [style=none] (31) at (-3.5, 0.1) {};
		\node [style=gn] (32) at (-5.25, -0.75) {};
		\node [style=none] (33) at (-3.75, 1.5) {};
		\node [style=rn,label={[rphase]left:$-\pi/2$}] (34) at (-3.75, -0.25) {};
		\node [style=Hadamard] (35) at (-5.75, -0.75) {};
		\node [style=Hadamard] (36) at (-5.5, -1.75) {};
		\node [style=Hadamard] (37) at (-6.25, -0.25) {};
		\node [style=Hadamard] (38) at (-3.75, 0.75) {};
		\node [style=Hadamard] (39) at (-4.5, -1.25) {};
		\node [style=Hadamard] (40) at (-4.5, 0.5) {};
		\node [style=none] (41) at (-7.5, -2) {};
		\node [style=none] (42) at (-3.5, -2) {};
		\node [style=gn,label={[gphase]left:$\pi/2$}] (43) at (-3.75, 1.25) {};
		\node [style=none] (44) at (-7.5, 0.1) {};
		\node [style=gn] (45) at (-3.75, -1.75) {};
		\node [style=gn] (46) at (-6.25, -0.75) {};
		\node [style=gn] (47) at (-3.75, 0.25) {};
		\node [style=none] (48) at (-8, -0) {$=$};
		\node [style=gn] (49) at (-5.25, 0.75) {};
		\node [style=gn,label={[gphase]right:$-\pi/2$}] (50) at (3.75, -0.25) {};
		\node [style=none] (51) at (3.5, -2) {};
		\node [style=gn] (52) at (7.5, -1.75) {};
		\node [style=Hadamard] (53) at (7.5, -0.25) {};
		\node [style=none] (54) at (7.75, 0.6) {};
		\node [style=none] (55) at (7.75, -2) {};
		\node [style=none] (56) at (5.25, 0.75) {};
		\node [style=none] (57) at (3, -0) {$=$};
		\node [style=gn] (58) at (6.25, -0.75) {};
		\node [style=Hadamard] (59) at (5.75, -0.75) {};
		\node [style=none] (60) at (3.5, 0.6) {};
		\node [style=gn,label={[gphase]left:$\pi/2$}] (61) at (7.5, 0.25) {};
		\node [style=none] (62) at (3.75, 0.75) {};
		\node [style=Hadamard] (63) at (5.625, -1.75) {};
		\node [style=none] (64) at (6.25, 0.75) {};
		\node [style=gn] (65) at (3.75, -1.75) {};
		\node [style=Hadamard] (66) at (5.25, -0.25) {};
		\node [style=Hadamard] (67) at (5, -1.25) {};
		\node [style=gn] (68) at (5.25, -0.75) {};
		\node [style=none] (69) at (7.5, 0.75) {};
	\end{pgfonlayer}
	\begin{pgfonlayer}{edgelayer}
		\draw [color=LightGrey] (16.center) to (17.center);
		\draw [color=LightGrey] (19.center) to (16.center);
		\draw [color=LightGrey] (18.center) to (17.center);
		\draw [color=LightGrey] (18.center) to (19.center);
		\draw [color=LightGrey] (42.center) to (31.center);
		\draw [color=LightGrey] (44.center) to (41.center);
		\draw [color=LightGrey] (44.center) to (31.center);
		\draw [color=LightGrey] (41.center) to (42.center);
		\draw [color=LightGrey] (55.center) to (54.center);
		\draw [color=LightGrey] (51.center) to (55.center);
		\draw [color=LightGrey] (60.center) to (54.center);
		\draw [color=LightGrey] (60.center) to (51.center);
		\draw (1.center) to (4);
		\draw (4) to (10);
		\draw (10) to (6.center);
		\draw (2.center) to (11);
		\draw (11) to (5);
		\draw (5) to (3.center);
		\draw (4) to (5);
		\draw (5) to (10);
		\draw (23) to (20);
		\draw (30) to (45);
		\draw (27.center) to (30);
		\draw (28.center) to (46);
		\draw (45) to (33.center);
		\draw (49) to (47);
		\draw (32) to (29.center);
		\draw (46) to (32);
		\draw (32) to (45);
		\draw (62.center) to (65);
		\draw (65) to (52);
		\draw (52) to (69.center);
		\draw (56.center) to (68);
		\draw (68) to (58);
		\draw (58) to (64.center);
		\draw (65) to (58);
	\end{pgfonlayer}
 \end{tikzpicture}
\end{center}
Lastly, the vertex operators are decomposed into red and green phase operators only, so the diagram can be brought into rGS-LC form:
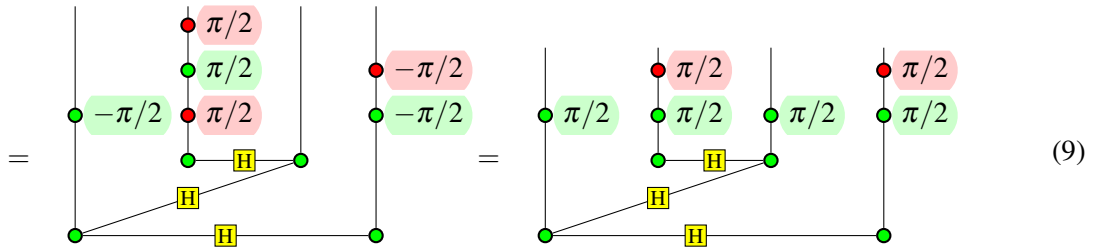
\begin{equation}\label{eq:diagram1}
 \begin{tikzpicture}[baseline=0cm]
	\begin{pgfonlayer}{nodelayer}
		\node [style=none] (0) at (-6.25, -0) {$=$};
		\node [style=none] (1) at (-5.5, 2.05) {};
		\node [style=none] (2) at (-4, 2.05) {};
		\node [style=none] (3) at (-2.5, 2.05) {};
		\node [style=gn] (4) at (-5.5, -1) {};
		\node [style=gn] (5) at (-2.5, -0) {};
		\node [style=none] (6) at (-1.5, 2.05) {};
		\node [style=Hadamard] (7) at (-3.25, -0) {};
		\node [style=Hadamard] (8) at (-3.5, -1) {};
		\node [style=gn,label={[gphase]right:$-\pi/2$}] (9) at (-1.5, 0.6) {};
		\node [style=gn] (10) at (-1.5, -1) {};
		\node [style=gn] (11) at (-4, -0) {};
		\node [style=Hadamard] (12) at (-4, -0.5) {};
		\node [style=gn,label={[gphase]right:$-\pi/2$}] (13) at (-5.5, 0.6) {};
		\node [style=rn,label={[rphase]right:$-\pi/2$}] (14) at (-1.5, 1.2) {};
		\node [style=rn,label={[rphase]right:$\pi/2$}] (15) at (-4, 0.6) {};
		\node [style=rn,label={[rphase]right:$\pi/2$}] (16) at (-4, 1.8) {};
		\node [style=gn,label={[gphase]right:$\pi/2$}] (17) at (-4, 1.2) {};
		\node [style=gn,label={[gphase]right:$\pi/2$}] (18) at (0.75, 0.6) {};
		\node [style=gn] (19) at (5.25, -1) {};
		\node [style=gn,label={[gphase]right:$\pi/2$}] (20) at (2.25, 0.6) {};
		\node [style=rn,label={[rphase]right:$\pi/2$}] (21) at (2.25, 1.2) {};
		\node [style=none] (22) at (2.25, 1.5) {};
		\node [style=none] (23) at (0, -0) {$=$};
		\node [style=gn] (24) at (3.75, -0) {};
		\node [style=Hadamard] (25) at (3, -0) {};
		\node [style=gn,label={[gphase]right:$\pi/2$}] (26) at (5.25, 0.6) {};
		\node [style=none] (27) at (0.75, 1.5) {};
		\node [style=Hadamard] (28) at (2.75, -1) {};
		\node [style=none] (29) at (3.75, 1.5) {};
		\node [style=rn,label={[rphase]right:$\pi/2$}] (30) at (5.25, 1.2) {};
		\node [style=gn] (31) at (0.75, -1) {};
		\node [style=Hadamard] (32) at (2.25, -0.5) {};
		\node [style=none] (33) at (5.25, 1.5) {};
		\node [style=gn] (34) at (2.25, -0) {};
		\node [style=gn,label={[gphase]right:$\pi/2$}] (35) at (3.75, 0.6) {};
	\end{pgfonlayer}
	\begin{pgfonlayer}{edgelayer}
		\draw (1.center) to (4);
		\draw (4) to (10);
		\draw (10) to (6.center);
		\draw (2.center) to (11);
		\draw (11) to (5);
		\draw (5) to (3.center);
		\draw (4) to (5);
		\draw (27.center) to (31);
		\draw (31) to (19);
		\draw (19) to (33.center);
		\draw (22.center) to (34);
		\draw (34) to (24);
		\draw (24) to (29.center);
		\draw (31) to (24);
	\end{pgfonlayer}
 \end{tikzpicture}
\end{equation}

Similarly, the second diagram becomes
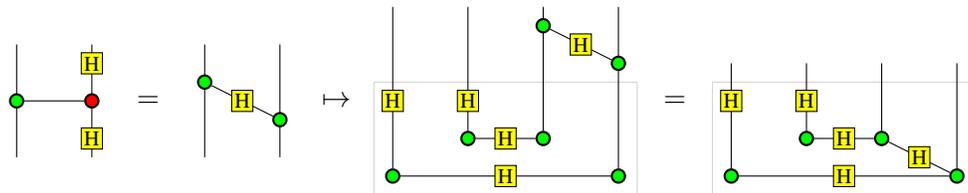
\begin{center}
 \begin{tikzpicture}
	\begin{pgfonlayer}{nodelayer}
		\node [style=gn] (0) at (-6, -0) {};
		\node [style=rn] (1) at (-5, -0) {};
		\node [style=Hadamard] (2) at (-5, 0.5) {};
		\node [style=Hadamard] (3) at (-5, -0.5) {};
		\node [style=none] (4) at (-6, 0.75) {};
		\node [style=none] (5) at (-5, 0.75) {};
		\node [style=none] (6) at (-6, -0.75) {};
		\node [style=none] (7) at (-5, -0.75) {};
		\node [style=none] (8) at (-4.25, -0) {$=$};
		\node [style=gn] (9) at (-3.5, 0.25) {};
		\node [style=gn] (10) at (-2.5, -0.25) {};
		\node [style=Hadamard] (11) at (-3, -0) {};
		\node [style=none] (12) at (-3.5, 0.75) {};
		\node [style=none] (13) at (-2.5, 0.75) {};
		\node [style=none] (14) at (-3.5, -0.75) {};
		\node [style=none] (15) at (-2.5, -0.75) {};
		\node [style=none] (16) at (-1.75, -0) {$\mapsto$};
		\node [style=gn] (17) at (-1, -1) {};
		\node [style=gn] (18) at (0, -0.5) {};
		\node [style=gn] (19) at (1, -0.5) {};
		\node [style=gn] (20) at (2, -1) {};
		\node [style=gn] (21) at (2, 0.5) {};
		\node [style=gn] (22) at (1, 1) {};
		\node [style=Hadamard] (23) at (0.5, -0.5) {};
		\node [style=Hadamard] (24) at (0.5, -1) {};
		\node [style=Hadamard] (25) at (0, -0) {};
		\node [style=Hadamard] (26) at (-1, -0) {};
		\node [style=Hadamard] (27) at (1.5, 0.75) {};
		\node [style=none] (28) at (-1, 1.25) {};
		\node [style=none] (29) at (0, 1.25) {};
		\node [style=none] (30) at (1, 1.25) {};
		\node [style=none] (31) at (2, 1.25) {};
		\node [style=none] (32) at (-1.25, 0.25) {};
		\node [style=none] (33) at (-1.25, -1.25) {};
		\node [style=none] (34) at (2.25, -1.25) {};
		\node [style=none] (35) at (2.25, 0.25) {};
		\node [style=none] (36) at (3.25, 0.25) {};
		\node [style=none] (37) at (2.75, -0) {$=$};
		\node [style=none] (38) at (5.5, 0.5) {};
		\node [style=none] (39) at (4.5, 0.5) {};
		\node [style=none] (40) at (6.75, -1.25) {};
		\node [style=Hadamard] (41) at (6, -0.75) {};
		\node [style=Hadamard] (42) at (5, -0.5) {};
		\node [style=gn] (43) at (3.5, -1) {};
		\node [style=Hadamard] (44) at (3.5, -0) {};
		\node [style=gn] (45) at (5.5, -0.5) {};
		\node [style=gn] (46) at (6.5, -1) {};
		\node [style=none] (47) at (6.75, 0.25) {};
		\node [style=none] (48) at (6.5, 0.5) {};
		\node [style=gn] (49) at (4.5, -0.5) {};
		\node [style=none] (50) at (3.25, -1.25) {};
		\node [style=none] (51) at (3.5, 0.5) {};
		\node [style=Hadamard] (52) at (4.5, -0) {};
		\node [style=Hadamard] (53) at (5, -1) {};
	\end{pgfonlayer}
	\begin{pgfonlayer}{edgelayer}
		\draw [color=LightGrey] (33.center) to (34.center);
		\draw [color=LightGrey] (34.center) to (35.center);
		\draw [color=LightGrey] (35.center) to (32.center);
		\draw [color=LightGrey] (32.center) to (33.center);
		\draw [color=LightGrey] (50.center) to (40.center);
		\draw [color=LightGrey] (40.center) to (47.center);
		\draw [color=LightGrey] (47.center) to (36.center);
		\draw [color=LightGrey] (36.center) to (50.center);
		\draw (4.center) to (6.center);
		\draw (5.center) to (7.center);
		\draw (0) to (1);
		\draw (12.center) to (14.center);
		\draw (13.center) to (15.center);
		\draw (9) to (10);
		\draw (29.center) to (18);
		\draw (30.center) to (19);
		\draw (18) to (19);
		\draw (28.center) to (17);
		\draw (17) to (20);
		\draw (31.center) to (20);
		\draw (22) to (21);
		\draw (39.center) to (49);
		\draw (38.center) to (45);
		\draw (49) to (45);
		\draw (51.center) to (43);
		\draw (43) to (46);
		\draw (48.center) to (46);
		\draw (45) to (46);
	\end{pgfonlayer}
 \end{tikzpicture}
\end{center}
which turns into
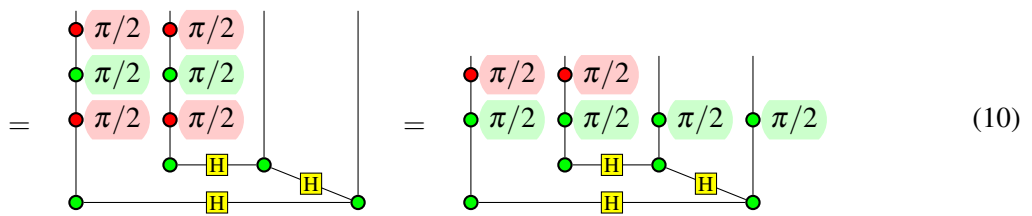
\begin{equation}\label{eq:diagram2}
 \begin{tikzpicture}[baseline=0cm]
	\begin{pgfonlayer}{nodelayer}
		\node [style=none] (0) at (-5.25, -0) {$=$};
		\node [style=none] (1) at (-2, 1.55) {};
		\node [style=none] (2) at (-3.25, 1.55) {};
		\node [style=Hadamard] (3) at (-1.375, -0.75) {};
		\node [style=Hadamard] (4) at (-2.625, -0.5) {};
		\node [style=gn] (5) at (-4.5, -1) {};
		\node [style=gn] (6) at (-2, -0.5) {};
		\node [style=gn] (7) at (-0.75, -1) {};
		\node [style=none] (8) at (-0.75, 1.55) {};
		\node [style=gn] (9) at (-3.25, -0.5) {};
		\node [style=none] (10) at (-4.5, 1.55) {};
		\node [style=Hadamard] (11) at (-2.625, -1) {};
		\node [style=rn,label={[rphase]right:$\pi/2$}] (12) at (-4.5, 0.1) {};
		\node [style=rn,label={[rphase]right:$\pi/2$}] (13) at (-4.5, 1.3) {};
		\node [style=rn,label={[rphase]right:$\pi/2$}] (14) at (-3.25, 1.3) {};
		\node [style=rn,label={[rphase]right:$\pi/2$}] (15) at (-3.25, 0.1) {};
		\node [style=gn,label={[gphase]right:$\pi/2$}] (16) at (-4.5, 0.7) {};
		\node [style=gn,label={[gphase]right:$\pi/2$}] (17) at (-3.25, 0.7) {};
		\node [style=none] (18) at (2, 0.95) {};
		\node [style=none] (19) at (0.75, 0.95) {};
		\node [style=none] (20) at (0, -0) {$=$};
		\node [style=gn] (21) at (4.5, -1) {};
		\node [style=gn] (22) at (0.75, -1) {};
		\node [style=none] (23) at (3.25, 0.95) {};
		\node [style=gn,label={[gphase]right:$\pi/2$}] (24) at (0.75, 0.1) {};
		\node [style=rn,label={[rphase]right:$\pi/2$}] (25) at (0.75, 0.7) {};
		\node [style=Hadamard] (26) at (2.625, -1) {};
		\node [style=Hadamard] (27) at (2.625, -0.5) {};
		\node [style=gn] (28) at (2, -0.5) {};
		\node [style=gn,label={[gphase]right:$\pi/2$}] (29) at (2, 0.1) {};
		\node [style=none] (30) at (4.5, 0.95) {};
		\node [style=rn,label={[rphase]right:$\pi/2$}] (31) at (2, 0.7) {};
		\node [style=Hadamard] (32) at (3.875, -0.75) {};
		\node [style=gn] (33) at (3.25, -0.5) {};
		\node [style=gn,label={[gphase]right:$\pi/2$}] (34) at (3.25, 0.1) {};
		\node [style=gn,label={[gphase]right:$\pi/2$}] (35) at (4.5, 0.1) {};
	\end{pgfonlayer}
	\begin{pgfonlayer}{edgelayer}
		\draw (2.center) to (9);
		\draw (1.center) to (6);
		\draw (9) to (6);
		\draw (10.center) to (5);
		\draw (5) to (7);
		\draw (8.center) to (7);
		\draw (6) to (7);
		\draw (18.center) to (28);
		\draw (23.center) to (33);
		\draw (28) to (33);
		\draw (19.center) to (22);
		\draw (22) to (21);
		\draw (30.center) to (21);
		\draw (33) to (21);
	\end{pgfonlayer}
 \end{tikzpicture}
\end{equation}

The last parts of \eqref{eq:diagram1} and \eqref{eq:diagram2} form a pair of rGS-LC diagrams, which we will now simplify. Numbering the qubits from left to right, we find that both diagrams have red nodes in the vertex operator of qubit 2, and that there are further red nodes in the vertex operator of qubit 4 in the first diagram and qubit 1 in the second diagram. Qubits 1 and 4 are connected in the first diagram, so we can apply the rGS-LC transformation given in proposition \ref{prop:rGS-LC_transformation2} to transfer the red node from one to the other. First, apply a local complementation about the edge $\{1,4\}$:
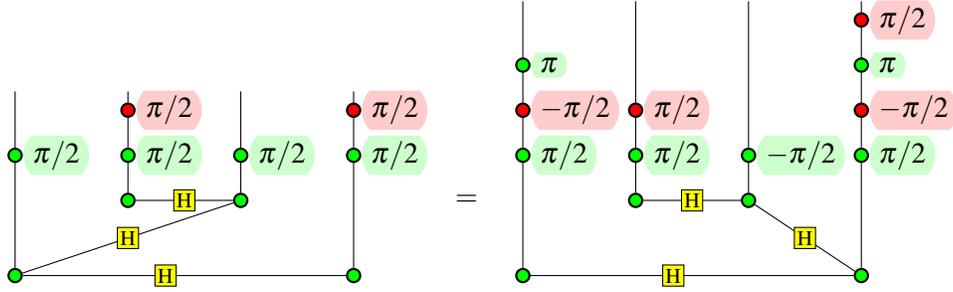
\begin{center}
 \begin{tikzpicture}
	\begin{pgfonlayer}{nodelayer}
		\node [style=gn,label={[gphase]right:$\pi/2$}] (0) at (-6, 0.6) {};
		\node [style=gn] (1) at (-1.5, -1) {};
		\node [style=gn,label={[gphase]right:$\pi/2$}] (2) at (-4.5, 0.6) {};
		\node [style=rn,label={[rphase]right:$\pi/2$}] (3) at (-4.5, 1.2) {};
		\node [style=none] (4) at (-4.5, 1.45) {};
		\node [style=gn] (5) at (-3, -0) {};
		\node [style=Hadamard] (6) at (-3.75, -0) {};
		\node [style=gn,label={[gphase]right:$\pi/2$}] (7) at (-1.5, 0.6) {};
		\node [style=none] (8) at (-6, 1.45) {};
		\node [style=Hadamard] (9) at (-4, -1) {};
		\node [style=none] (10) at (-3, 1.45) {};
		\node [style=rn,label={[rphase]right:$\pi/2$}] (11) at (-1.5, 1.2) {};
		\node [style=gn] (12) at (-6, -1) {};
		\node [style=Hadamard] (13) at (-4.5, -0.5) {};
		\node [style=none] (14) at (-1.5, 1.45) {};
		\node [style=gn] (15) at (-4.5, -0) {};
		\node [style=gn,label={[gphase]right:$\pi/2$}] (16) at (-3, 0.6) {};
		\node [style=none] (17) at (0, -0) {$=$};
		\node [style=gn,label={[gphase]right:$\pi$}] (18) at (5.25, 1.8) {};
		\node [style=none] (19) at (3.75, 2.65) {};
		\node [style=gn] (20) at (3.75, -0) {};
		\node [style=gn] (21) at (0.75, -1) {};
		\node [style=rn,label={[rphase]right:$\pi/2$}] (22) at (2.25, 1.2) {};
		\node [style=Hadamard] (23) at (2.75, -1) {};
		\node [style=gn] (24) at (5.25, -1) {};
		\node [style=Hadamard] (25) at (3, -0) {};
		\node [style=gn,label={[gphase]right:$\pi$}] (26) at (0.75, 1.8) {};
		\node [style=none] (27) at (5.25, 2.65) {};
		\node [style=gn,label={[gphase]right:$-\pi/2$}] (28) at (3.75, 0.6) {};
		\node [style=none] (29) at (0.75, 2.65) {};
		\node [style=rn,label={[rphase]right:$\pi/2$}] (30) at (5.25, 2.4) {};
		\node [style=none] (31) at (2.25, 2.65) {};
		\node [style=Hadamard] (32) at (4.5, -0.5) {};
		\node [style=gn,label={[gphase]right:$\pi/2$}] (33) at (2.25, 0.6) {};
		\node [style=gn] (34) at (2.25, -0) {};
		\node [style=rn,label={[rphase]right:$-\pi/2$}] (35) at (0.75, 1.2) {};
		\node [style=rn,label={[rphase]right:$-\pi/2$}] (36) at (5.25, 1.2) {};
		\node [style=gn,label={[gphase]right:$\pi/2$}] (37) at (0.75, 0.6) {};
		\node [style=gn,label={[gphase]right:$\pi/2$}] (38) at (5.25, 0.6) {};
	\end{pgfonlayer}
	\begin{pgfonlayer}{edgelayer}
		\draw (8.center) to (12);
		\draw (12) to (1);
		\draw (1) to (14.center);
		\draw (4.center) to (15);
		\draw (15) to (5);
		\draw (5) to (10.center);
		\draw (12) to (5);
		\draw (29.center) to (21);
		\draw (21) to (24);
		\draw (24) to (27.center);
		\draw (31.center) to (34);
		\draw (34) to (20);
		\draw (20) to (19.center);
		\draw (20) to (24);
	\end{pgfonlayer}
 \end{tikzpicture}
\end{center}
Then rewrite the vertex operators into standard form and apply a fixpoint operation about qubit 4, to get a diagram that is once again in rGS-LC form:
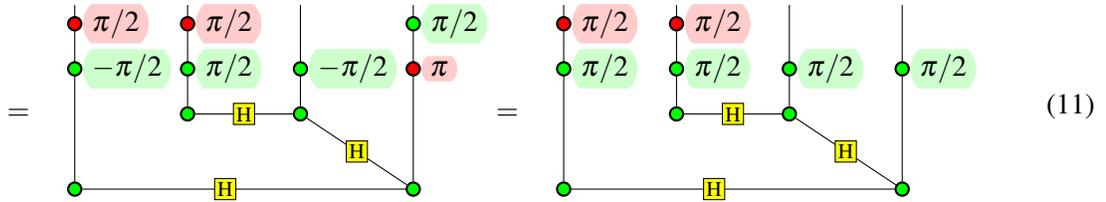
\begin{equation}\label{eq:diagram1_simplified}
 \begin{tikzpicture}[baseline=0cm]
	\begin{pgfonlayer}{nodelayer}
		\node [style=gn,label={[gphase]right:$\pi/2$}] (0) at (-1.25, 1.2) {};
		\node [style=none] (1) at (-2.75, 1.45) {};
		\node [style=gn] (2) at (-2.75, -0) {};
		\node [style=gn] (3) at (-5.75, -1) {};
		\node [style=rn,label={[rphase]right:$\pi/2$}] (4) at (-4.25, 1.2) {};
		\node [style=Hadamard] (5) at (-3.75, -1) {};
		\node [style=gn] (6) at (-1.25, -1) {};
		\node [style=Hadamard] (7) at (-3.5, -0) {};
		\node [style=none] (8) at (-1.25, 1.45) {};
		\node [style=gn,label={[gphase]right:$-\pi/2$}] (9) at (-2.75, 0.6) {};
		\node [style=none] (10) at (-5.75, 1.45) {};
		\node [style=none] (11) at (-4.25, 1.45) {};
		\node [style=Hadamard] (12) at (-2, -0.5) {};
		\node [style=gn,label={[gphase]right:$\pi/2$}] (13) at (-4.25, 0.6) {};
		\node [style=gn] (14) at (-4.25, -0) {};
		\node [style=rn,label={[rphase]right:$\pi/2$}] (15) at (-5.75, 1.2) {};
		\node [style=rn,label={[rphase]right:$\pi$}] (16) at (-1.25, 0.6) {};
		\node [style=gn,label={[gphase]right:$-\pi/2$}] (17) at (-5.75, 0.6) {};
		\node [style=none] (18) at (0, -0) {$=$};
		\node [style=none] (19) at (-6.5, -0) {$=$};
		\node [style=none] (20) at (2.25, 1.45) {};
		\node [style=gn,label={[gphase]right:$\pi/2$}] (21) at (3.75, 0.6) {};
		\node [style=Hadamard] (22) at (3, -0) {};
		\node [style=Hadamard] (23) at (2.75, -1) {};
		\node [style=none] (24) at (5.25, 1.45) {};
		\node [style=gn] (25) at (5.25, -1) {};
		\node [style=gn,label={[gphase]right:$\pi/2$}] (26) at (5.25, 0.6) {};
		\node [style=rn,label={[rphase]right:$\pi/2$}] (27) at (0.75, 1.2) {};
		\node [style=gn] (28) at (2.25, -0) {};
		\node [style=gn,label={[gphase]right:$\pi/2$}] (29) at (2.25, 0.6) {};
		\node [style=none] (30) at (0.75, 1.45) {};
		\node [style=gn] (31) at (3.75, -0) {};
		\node [style=gn,label={[gphase]right:$\pi/2$}] (32) at (0.75, 0.6) {};
		\node [style=none] (33) at (3.75, 1.45) {};
		\node [style=Hadamard] (34) at (4.5, -0.5) {};
		\node [style=rn,label={[rphase]right:$\pi/2$}] (35) at (2.25, 1.2) {};
		\node [style=gn] (36) at (0.75, -1) {};
	\end{pgfonlayer}
	\begin{pgfonlayer}{edgelayer}
		\draw (10.center) to (3);
		\draw (3) to (6);
		\draw (6) to (8.center);
		\draw (11.center) to (14);
		\draw (14) to (2);
		\draw (2) to (1.center);
		\draw (2) to (6);
		\draw (30.center) to (36);
		\draw (36) to (25);
		\draw (25) to (24.center);
		\draw (20.center) to (28);
		\draw (28) to (31);
		\draw (31) to (33.center);
		\draw (31) to (25);
	\end{pgfonlayer}
 \end{tikzpicture}
\end{equation}
The pair of diagrams given by the last parts of \eqref{eq:diagram2} and \eqref{eq:diagram1_simplified} is now simplified. In fact, these two diagrams are identical --- as expected, considering we started with two different circuit representations of the same quantum-mechanical operator.

By taking the sequence of diagrams derived here and bending outputs in those diagrams back into inputs, we could now derive a sequence of rewrites which show directly that the two diagrams given in \eqref{eq:example_diagrams} are equal.

\section{Conclusions and further work}
\label{s:conclusions}

We show that the \ZX-calculus is complete for stabilizer quantum mechanics by transforming any stabilizer diagram into a normal form consisting of a graph state and single qubit Clifford operators. A natural next question is whether this result can be extended to a larger fragment of the \ZX-calculus, e.g. by allowing nodes whose phases are multiples of $\pi/4$. It is not immediately obvious how to do this, since the completeness proof relies on several results specific to Clifford operators. It would also be interesting to investigate which rules need to be added to make the \ZX-calculus complete for general pure state qubit quantum mechanics.

The software system \texttt{Quantomatic} \cite{quantomatic} enables semi-automated manipulation of graphical calculus diagrams. It should be possible to implement the equality testing algorithm from section \ref{s:equality_testing} in this system, allowing automated comparison of diagrams. This immediately offers a new question, namely that of the computational complexity of the equality decision problem.

In \cite{coecke_phase_2011}, categorical quantum mechanics is used to analyse the origin of non-locality in stabilizer quantum mechanics as compared to Spekkens' toy theory \cite{spekkens_evidence_2007}. This, together with Pusey's work on a stabilizer formalism for the toy theory \cite{pusey_stabilizer_2012}, suggests a ``\ZX-calculus'' for that theory. It would be interesting to see whether the completeness result for the \ZX-calculus for stabilizer quantum mechanics can be reproduced in the graphical calculus for Spekkens' toy theory.

\section*{Acknowledgements}

Thanks to Ross Duncan for suggesting the problem and to Bob Coecke for his advice in preparing this paper. I acknowledge financial support from the EPSRC.

\bibliographystyle{eptcs}
\bibliography{stabilizer_completeness}

\end{document}